\documentclass[12pt]{article}

\usepackage[utf8]{inputenc}
\usepackage{adforn}
\usepackage{amsmath}
\usepackage{amsfonts}
\usepackage{amssymb}
\usepackage{alphalph}
\usepackage[titletoc]{appendix}
\usepackage{comment}
\usepackage{commath}
\usepackage{hyperref}
\usepackage{datetime}
\usepackage{footnote}
\usepackage{pdfpages}
\usepackage{pdflscape}
\usepackage[nodisplayskipstretch]{setspace}
\usepackage{amsthm}
\usepackage{fancyhdr}
\usepackage{enumerate}
\usepackage{bbm}
\usepackage{hhline}

\hypersetup{
    colorlinks=true,
    citecolor=blue,
    linkcolor=blue,
    filecolor=magenta,      
    urlcolor=blue
}

\usepackage{adjustbox}
\usepackage{booktabs}
\usepackage{caption}
\usepackage{epstopdf}
\usepackage{graphicx}
\usepackage{longtable}
\usepackage{multirow}
\usepackage{ragged2e}
\usepackage{tabularx}
\usepackage{tikz}
\usepackage{circuitikz}
\usetikzlibrary{patterns}
\usetikzlibrary{arrows}
\usetikzlibrary{decorations.markings}
\usepackage[para, flushleft]{threeparttablex}
\usepackage{float}
\usepackage{soul}
\usepackage{cancel}
\usepackage{enumitem}

\usepackage[margin=1in]{geometry}

\usepackage{natbib}

\theoremstyle{plain}
\newtheorem{theorem}{Theorem}
\newtheorem{lemma}{Lemma}

\newtheorem{proposition}{Proposition}
\newtheorem{assumption}{Assumption}

\theoremstyle{definition}
\newtheorem{definition}{Definition}

\newtheorem{remark}{Remark}
\newtheorem*{remark*}{Remark}
\newtheorem{example}{Example}

\theoremstyle{remark}

\newcommand{\expof}{w}
\newcommand{\exposet}{\mathcal W}
\newcommand{\expov}{\mathrm{w}} 

\newcommand{\nei}{\mathcal{N}}
\newcommand{\obs}{\mathrm{obs}}
\newcommand{\Zobs}{Z^{\obs}}
\newcommand{\Yobs}{Y^{\obs}}
\newcommand{\Tobs}{T^{\obs}}

\newcommand{\pval}{\mathrm{pval}}

\newcommand{\mS}{\mathcal S}
\newcommand{\foc}{\mathrm{foc}}
\newcommand{\rand}{\mathrm{rand}}
\newcommand{\ef}{\mathsf{E}_{\foc}}
\newcommand{\er}{\mathsf{E}_{\rand}}
\newcommand{\af}{\mathsf{A}_{\foc}}
\newcommand{\ar}{\mathsf{A}_{\rand}}

\newcommand{\mU}{\mathcal U}
\newcommand{\mZ}{\mathcal Z}
\newcommand{\mC}{\mathcal C}
\usepackage{algorithm}
\usepackage{algorithmic}

\title{Finite-Sample Valid Randomization Tests for \\ 
Monotone Spillover Effects}
\author{Shunzhuang Huang\thanks{University of Chicago, Booth School of Business, Chicago, IL, USA. \url{shunzhuang.huang@chicagobooth.edu}}
\and Xinran Li\thanks{University of Chicago, Department of Statistics, Chicago, IL, USA. \url{xinranli@uchicago.edu}}
\and Panos Toulis\thanks{University of Chicago, Booth School of Business, Chicago, IL, USA. \url{panos.toulis@chicagobooth.edu}\\
PT acknowledges support from NSF SES-2419009. All authors wish to thank Azeem Shaikh and Chris Hansen for valuable feedback and comments.}
}

\date{February 23, 2025}

\begin{document}
\maketitle

{\onehalfspacing
\begin{abstract}
Randomization tests have gained popularity for causal inference under network interference because they are finite-sample valid with minimal assumptions. However, existing procedures are limited as they primarily focus on the existence of spillovers through sharp null hypotheses on potential outcomes.
In this paper, we expand the scope of randomization procedures in network settings by developing new tests for the monotonicity of spillover effects.
These tests offer insights into whether spillover effects increase, decrease, or exhibit ``diminishing returns" along certain network dimensions of interest. Our approach partitions the network into multiple (possibly overlapping) parts and tests a monotone contrast hypothesis in each sub-network. The test decisions can then be aggregated in various ways depending on how each test is constructed. 
We demonstrate our method by re-analyzing a large-scale policing experiment in Colombia, which reveals evidence of monotonicity related to the ``crime displacement hypothesis". Our analysis suggests that crime spillovers on a control street increase with the number of nearby streets receiving more intense policing but diminish at higher exposure levels.
\end{abstract}}

{\bf Keywords}: causal inference, interference, monotone spillover, network, randomization test

\newpage
\section{Introduction}
In many real-world experimental settings, the treatment effect on some units may depend on 
the treatments administered to other units through a network. The presence of such spillover effects violates the classical ``no interference" assumption~\citep{cox1958planning, rubin1980randomization} and presents unique methodological challenges in the estimation of treatment effects.

To address this problem, a recent line of research in causal inference under interference has leveraged the classical Fisherian randomization test~\citep{fisher1935}. 
The key idea is to condition on a subset of units and treatment assignments for which potential outcomes are imputable under the null hypothesis of no spillover effect~\citep{athey2018exact,basse2019randomization,puelz2022graph}, and perform a {\em conditional} 
randomization test.
The resulting procedures are valid in finite samples without requiring a correct specification of the outcome model and are generally straightforward to implement in practice.

However, the scope of these procedures is currently limited to testing the sharp equality of potential outcomes under certain levels of treatment exposure; e.g., testing whether a unit's outcomes are unaffected by having zero or at least one treated neighbor, holding the unit's individual treatment fixed.
While such sharp tests can be useful for understanding the existence and magnitude of spillover effects, 
they remain silent on how the spillover effect varies as a function of treatment exposure.

In this paper, we expand the scope of existing randomization procedures to null hypotheses related to the monotonicity of spillover effects. 
We allow treatment exposure to take values in a totally ordered set ---e.g., 0, 1, 2, or more treated neighbors--- and test whether higher treatment exposure leads to better or worse outcomes.  
The key idea underlying our test is to split the network into multiple, possibly overlapping, sub-networks, and use each sub-network to test a single monotone hypothesis on treatment exposure. We then combine the separate test results while controlling for possible dependence between these tests. 
The final procedure is finite-sample valid and is straightforward to implement, especially under Bernoulli randomized designs.

We apply our procedure to re-examine a large-scale policing experiment in Medellín, Colombia conducted by~\cite{collazos2021hot}. The goal of our analysis is to test whether crime spillovers on a control street are monotone with respect to the number of neighboring treated streets.
Our results provide insights related to the ``crime displacement hypothesis" and may be relevant for crime prevention policies. Specifically, we find that a control street is negatively affected by being exposed to nearby streets that are treated with more intense policing, in a way that is monotone in the total number of treated neighboring streets. In addition, we find evidence of ``diminishing returns", where the magnitude of the crime spillover effect diminishes at higher levels of exposure to treated neighboring streets. 

Our proposed methodology builds upon existing methodologies in the randomization inference literature.  First, we leverage existing methodologies from 
conditional randomization tests under network interference~\citep{athey2018exact,basse2019randomization,puelz2022graph,basse2024randomization}.
While these procedures are finite-sample valid for testing certain types of causal effects under network interference, they are not applicable for monotone null hypotheses.
To address this limitation, we leverage recent results in testing
bounded null hypotheses under no interference~\citep{caughey2023bounded}, and extend the related methodology to settings with interference.
Furthermore, our network-splitting strategy outlined above utilizes the idea of approximate evidence factors developed by~\cite{rosenbaum2011some}. 
Notably,~\cite{zhang2021multiple} also use this idea to construct conditional randomization tests in a recursive manner. This recursive construction, however, is largely context-specific and not directly applicable to our monotone hypothesis.

The rest of the paper is organized as follows. Section~\ref{sec:setup} introduces the general setup, and provides the formal definition of our monotone hypothesis. 
We then outline the proposed method in Section~\ref{sec:overview}. In Section~\ref{sec:nonunifBRE} we present a detailed description of our method tailored to non-uniform Bernoulli designs. More general experimental designs are considered in Section~\ref{sec:method_clique}. 
In Section~\ref{sec:simu}, we show the validity and power of our tests through simulated studies. In Section~\ref{sec:realMed}, we apply our methodology to the Medellín data. Section~\ref{sec:conclusion} concludes the paper.

\section{Problem Setup}\label{sec:setup}
Consider a finite population of $N$ units indexed by $i \in [N] := \{1,2,\ldots,N\}$. The treatment for unit $i$ is denoted as $Z_i \in \{0,1\}$. The population treatment is the column vector $Z := (Z_1,Z_2,\ldots,Z_N) \in \{0,1\}^N$, and follows a known treatment assignment design $Z \sim P(Z)$. 
Vector $Z_S$ will denote the sub-vector of $Z$ that corresponds to the units in $S \subseteq [N]$.
The potential outcome for unit $i$ under population treatment $z \in \{0,1\}^N$ is denoted as $Y_i(z) \in \mathbb R$, and $Y(z) := (Y_1(z), Y_2(z), \ldots, Y_N(z))$ is the potential outcome (column) vector for the whole population. 
These potential outcomes are fixed under the randomization framework. 
Let $\Zobs \sim P(\Zobs)$ be the observed treatments and $\Yobs = Y(\Zobs)$ be the observed outcomes. Each unit $i$ may also have covariates $X_i \in \mathbb{R}^p$, which could include pre-treatment outcomes.

Under the classical ``stable unit treatment value assumption"~\citep{rubin1980randomization}, a unit's potential outcome depends solely on its treatment $Z_i$, so that $Y_i(Z) = Y_i(Z')$ for any unit $i$ and any assignments $Z,Z'$ for which $Z_i=Z_i'$. Under interference, this assumption is no longer plausible. However, without any restrictions on interference, each unit may have $2^N$ different potential outcomes, which is computationally prohibitive.
To make progress, we follow a common approach and assume a low-dimensional summary of interference~\citep{hong2006evaluating}, also known as ``treatment exposure"~\citep{aronow2017estimating} or ``effective treatment"~\citep{manski2013identification}.

\begin{assumption}[Treatment Exposure]
\label{a:exposure}
There exists a finite set $\exposet$ endowed with an equality and ordering relation, exposure mapping functions $\expof_i: \{0,1\}^N \to \exposet$, and potential outcome functions $y_i: \{0,1\} \times \exposet \to \mathbb{R}$ for $i\in[N]$, such that
\begin{equation}
\label{eq:exposure}    
Y_i(z) = y_i(z_i, \expof_i(z))~\text{for all}~i,~z.
\end{equation}
In many settings where the exposure mapping function for some unit $i$ may depend only on treatments of a subset of units $S \subseteq [N]$, i.e., $\expof_i(z) = \expof_i(z')$ for any $z,z'$ such that $z_{S} = z_{S}'$, we will overload the notation and write $\expof_i(z_S)$ in place of $\expof_i(z)$. Set $S$ may vary across $i$.
\end{assumption}
Here, $\expof_i(z)$ is the treatment exposure of unit $i$ under population treatment $z\in\{0,1\}^N$ due to interference. 
We note that the above definition of exposure mappings is not unique. For example, any monotone transformation of $\expof_i$, or any set $\exposet'$ defined more ``finely" than $ \exposet$ would also satisfy Assumption~\ref{a:exposure}. 
In this paper, we will focus on a special type of treatment exposure arising from interference due to a network between units, which we introduce below.

Let $\mathcal G = (V,E)$ be a network with vertex set $V = [N]$, edge set $E$, and adjacency matrix $A \in \{0,1\}^{N\times N}$, representing connections between units (e.g., friendship or geographical proximity). $A_{ij}=1$ if units $i$ and $j$ share a connection, and we assume this matrix to be symmetric with no self-loops.
For each $i \in V$ define the set of neighboring nodes as $\nei_i:= \{j \in V: (i,j) \in E\}$, and $\nei(S) := \bigcup_{i \in S} \nei_i$ for any $S \subseteq V$. 

Although our methodology is not tied to a particular exposure definition, we will mainly work with the following exposure function throughout the paper unless stated otherwise.
\begin{equation}\label{eq:neightreated_cont}
\expof_i(z) = \sum_{j\in[N]} A_{ij} z_j = \sum_{j \in \nei_i} z_j = \expof_i(z_{\nei_i}).
\end{equation}
The definition in Equation~\eqref{eq:neightreated_cont} implies that the potential outcomes of a unit are affected by the number of direct neighbors of $i$ that are treated under $z$. In this case, $\exposet = \{0, 1, \ldots, d_{\max}\}$ is the set of all possible exposures, where $d_{\max}$ is the maximum degree of $\mathcal G$. 
A natural ordering for this set is the usual ``less than or equal to" relation.

An important departure of our work from prior literature relates to the ordering of treatment exposures in $\exposet$. In previous work on exact randomization tests under network interference, treatment exposure on a unit $i$ is defined as whether any immediate neighbors of $i$ are treated~\citep{basse2019randomization, puelz2022graph}, or as the treatment sub-vector corresponding to all units within a certain distance from $i$~\citep{athey2018exact}.
In such cases, different values of treatment exposure just correspond to different levels of interference without a particular ordering. In our work, treatment exposures have a natural ordering, and the goal is to test whether such ordering in exposures induces an ordering on the potential outcomes as well. We turn to this question next.

\subsection{Null hypothesis of monotone spillover effects}
Suppose that we have an ordered exposure set (e.g., ``number of treated neighbors") 
$\exposet = \{ \expov_1, \expov_2, \ldots, \expov_K \}$ such that $\expov_j \leq \expov_{k} \Leftrightarrow j \leq k$, where equality is attained if and only if $j = k$.
The central goal of our paper is to test the following {\em monotone decreasing null hypothesis}:
\begin{equation}\label{eq:null_mono}
    H_0: \quad y_i(0, \expov_1) ~\geq~ y_i(0, \expov_2) ~\geq~ \cdots ~\geq~ y_i(0, \expov_K), \quad \forall i \in [N].
\end{equation}
Equivalently, we are testing $K-1$ hypotheses of the form
\begin{equation}\label{eq:null_mono_k}
    H_{0k}: \quad y_i(0, \expov_k) ~\geq~ y_i(0, \expov_{k+1}), \quad \forall i \in [N],
\end{equation}
for $k \in [K-1]$, such that $H_0 = \bigcap_{k \in [K-1]} H_{0k}$ can be defined as an intersection hypothesis. Note that the monotone increasing null hypothesis, $y_i(0, \expov_1) \leq y_i(0, \expov_2) \leq \cdots \leq y_i(0, \expov_K)$, can also be tested by flipping either the exposure labels or the sign of the outcomes.

Under the network interference with the exposure model in~\eqref{eq:neightreated_cont}, $H_0$ means that as more of $i$'s neighbors are treated, the outcome of $i$ changes in a monotone decreasing manner, even when $i$'s own treatment stays fixed.
There are many real-world settings where such a question is of scientific interest. We describe some illustrative examples below.

\begin{example}[Crime spillovers]\label{eg:crime_spill}
The ``crime displacement hypothesis" posits that crime prevention efforts in one area may shift criminal activities to other areas in situations where offenders can be mobile. Recent studies, however, also suggest an opposite effect where the benefits of crime prevention extend to nearby areas that were not directly targeted by the intervention~\citep{guerette2017assessing}. 
In this context, understanding whether spillover effects are affected monotonically by the strength of crime prevention efforts in nearby areas, and whether these effects saturate at certain levels of exposure, is of significant practical interest. While numerous studies examine the existence of crime spillover effects\footnote{See, for example, ~\cite{blattman2021place,collazos2021hot} and \cite{guerette2017assessing} for a review.}, to our best knowledge none have studied the monotonicity of these effects.
\end{example}

\begin{example}[Social networks]
During Covid-19, 
\cite{breza2021effects} conducted a cluster saturation randomized design to study the effect of stay-at-home messages through social media on people's mobility. 
The researchers observed a greater reduction in Covid-19 cases in control areas within clusters treated with high-intensity messaging compared to control areas within clusters treated with low-intensity messaging.
This observation can be framed into a monotone spillover hypothesis of Equation~\eqref{eq:null_mono} in our setup, where $Y$ measures Covid-19 cases and the exposure is the treatment intensity of the unit's cluster. 
\end{example}

Despite its significance, $H_0$ is challenging to test in finite samples via randomization tests. One key problem is that $H_0$ is a ``non-sharp hypothesis", which does not allow the imputation of all missing potential outcomes.
While recent randomization-based methods have dealt with certain classes of non-sharp hypotheses~\citep{athey2018exact, basse2019randomization, puelz2022graph}, these methods rely 
on being able to reduce a non-sharp hypothesis into a sharp hypothesis by appropriately conditioning the test on a subset of the data. Such a reduction is not possible in our setting because the ordering relation in Equation~\eqref{eq:null_mono} generally cannot ``pin down" the missing potential outcomes, except perhaps in certain limited cases, such as when outcomes are binary.

As we will see in the next section, our proposed method addresses this challenge by combining two distinct approaches in constructing randomization tests for non-sharp null hypotheses: one approach tests a ``sharper" version of $H_{0k}$ with a conditional Fisherian randomization test, and another approach extends the conditional test on the sharp null towards testing the corresponding non-sharp null  of Equation~\eqref{eq:null_mono_k}.
Below, we discuss our method on a high level, leaving details for Sections~\ref{sec:nonunifBRE} and~\ref{sec:method_clique}.

\section{Overview of Main Method}\label{sec:overview}

In this section, we provide an overview of our proposed method. 
Conceptually, our method can be described in four key steps: 

\begin{center}
    \textsc{General Procedure for Testing the Monotone Null}
\end{center}
\vspace{-15px}
\begin{enumerate}[topsep=0pt,itemsep=-0.5ex,partopsep=0.5ex,parsep=0.5ex]
    \item[{\bf Step 1.}] Split the network into $K-1$ parts in order to test $H_{0k}$ separately within each part.
    
    \item[{\bf Step 2.}]  Based on the original hypothesis $H_{0k}$ in Equation~\eqref{eq:null_mono_k}, define a sharper null hypothesis, $\widetilde H_{0k}$, that can be tested with existing conditional randomization tests~\citep{athey2018exact, puelz2022graph}.
    
    \item[{\bf Step 3.}] Adapt the test for $\widetilde H_{0k}$ towards testing $H_{0k}$ by applying the bounded null techniques of~\citet{caughey2023bounded}. This requires the use of test statistics 
    with a suitable property of {\em exposure monotonicity}, which we make concrete below.
    
    \item[{\bf Step 4.}] Combine the $K-1$ individual $p$-values obtained from each network part towards testing the main hypothesis, $H_0$. To this end, we leverage the concept of {\em stochastically larger than uniform} $p$-values of~\citet{rosenbaum2011some}.
    
\end{enumerate}
Below, we review some important aspects of this procedure, leaving technical details for the sections that follow.

\paragraph{Step 1: Splitting the network.}~The first step is to split the network, $\mathcal G = (V, E)$, into several ---possibly overlapping in nodes--- sub-networks denoted by $(\mathcal G_k)_{k \in [K]}$ where $\mathcal G_k = (V_k, E_k)$.
The idea is to test $H_{0k}$ within each sub-network $\mathcal G_k$ through a conditional randomization test, each producing a finite-sample valid $p$-value. The test for the monotone null defined in Equation~\eqref{eq:null_mono} then comes from the combination of these $p$-values, and the way $\mathcal G$ is split is designed to facilitate this combination. 
We discuss more technical details about such splitting later in Sections~\ref{sec:nonunifBRE} and~\ref{sec:method_clique}.

\paragraph{Step 2: Single contrast hypothesis.}~The second step is to focus on a sharper version of $H_{0k}$, defined as follows:
\begin{equation}\label{eq:null_eq_k}
    \widetilde H_{0k}: \quad y_i(0, \expov_k) ~=~ y_i(0, \expov_{k+1}), \quad \forall i \in [N].
\end{equation}
The null hypothesis in Equation~\eqref{eq:null_eq_k} remains a non-sharp null hypothesis as it only compares two out of the $2K$ potential outcomes.
However, $\widetilde H_{0k}$ can be tested using existing methods from the conditional randomization testing literature~\citep{athey2018exact, basse2019randomization, puelz2022graph}. 
As we review later, the key idea in these methods is to subset the data in a particular way such that 
a conditional randomization test is possible within only those control units exposed to either $\expov_k$ or $\expov_{k+1}$.

\paragraph{Step 3: Exposure-monotone test statistics.}~The next step is to adapt the tests derived for $\widetilde H_{0k}$ towards 
testing $H_{0k}$. This is possible as long as the test statistics in each individual randomization test satisfy the following {\em exposure-monotone} property. 

\begin{definition}[Exposure monotonicity]
Let $\mU \subseteq [N]$ be a subset of units and $\mZ \subseteq \{0,1\}^N$ be a subset of treatments, and write $\mC = (\mU, \mZ)$.
A test statistic, $t(z, y; \mC) : \{0,1\}^N \times \mathbb R^N \to \mathbb R$, is exposure-monotone with respect to $\mC$ in the order $(\expov, \expov')$ with $\expov \le \expov'$ if (a) its value depends only on the 
sub-vector of $y$ restricted on units in $\mU$ and the sub-vector of $z$ restricted on units in $\mU ~\cup~ \nei(\mU)$; 
and (b) for all $z \in \mZ$ and $y, \eta, \xi \in \mathbb R^N$ with $\eta_i \geq 0 \geq \xi_i ~\forall i \in \mathcal U$, 
\begin{equation}\label{eq:expo_monotone}
    t(z, y_{\eta\xi}; \mC) \geq t(z, y; \mC),
\end{equation}
where $y_{\eta\xi, i} = y_i + \mathbbm 1\{\expof_i(z) = \expov'\} \eta_i + \mathbbm 1\{\expof_i(z) = \expov\} \xi_i$. We will sometimes write the test statistic as $t(z, y; \mathcal U)$ when $\mathcal Z$ is clear from the context.

\end{definition}
\vspace{-5px}
Intuitively, an exposure-monotone test statistic operates only on a subset of units and treatment assignments, which, as explained later, are selected such that under these treatment assignments, missing potential outcomes can be imputed or bounded by the observed outcomes for the subset of units, usually termed ``focal units"~\citep{athey2018exact}.
Moreover, an exposure-monotone statistic is ``aligned" with the ordering of outcomes in the null hypothesis: its value increases towards the direction of higher exposure levels and decreases towards the direction of lower exposure levels. This definition extends the concept of ``effect increasing test statistics" developed by~\cite{caughey2023bounded} to settings with interference.

One example of an exposure-monotone test statistic is the simple difference-in-means test statistic in two exposure groups. 
Let $I_i(z, \expov) = \mathbbm{1}\{ w_i(z) = \expov\}$ indicate whether unit $i$ is exposed 
to level $\expov$ under population treatment $z$. For any $\expov \leq \expov'$, define 
\begin{equation}\label{eq:diff-in-mean}
    t_{\expov, \expov'}^{\mathrm{DiM}}(z, y; \mU) 
    = \frac{1}{\sum_{i\in\mU} I_i(z, \expov')} 
    \sum_{i\in\mU} I_i(z, \expov') \psi_1(y_i) - 
    \frac{1}{\sum_{i\in\mU} I_i(z, \expov)} 
    \sum_{i\in\mU} I_i(z, \expov) \psi_0(y_i),
\end{equation}
where $\psi_1$ and $\psi_0$ are non-decreasing functions from $\mathbb R$ to $\mathbb R$. This test statistic 
is the difference in (transformed) sample mean outcomes between units in $\mU$ exposed to $\expov'$ and $\expov$.
It also allows weighting outcomes by probabilities, such as the conditional probability of unit $i$ being exposed to $\expov'$ or $\expov$, which can be easily calculated in certain designs.

Another example is the rank-based statistic of the form
\begin{equation}\label{eq:rank-stat}
    t_{\expov, \expov'}^{\mathrm{rank}}(z, y; \mC) = 
    \sum_{i \in \mathcal U} I_i(z, \expov') \varphi(r_i(y_{\mU})),
\end{equation}
where $\varphi$ is a non-decreasing function and $r_i(y_{\mU})$ is the rank of $y_i$ within $y_{\mathcal U} = (y_i: i\in\mU)$, the outcome sub-vector for units in $\mU$. In case of tied ranks, we define $\varphi(\cdot)$ as the average value of $\varphi(\cdot)$ evaluated at those ranks with ties broken by unit ordering.\footnote{That is, $r_i > r_j$ if and only if $(y_i > y_j) \vee (y_i = y_j,~i > j)$.} We verify the exposure monotonicity of these two statistics in Appendix~\ref{appdx:exp-mono-stat}.

\paragraph{Step 4: Combination of $p$-values.}~In the last step, we combine the individual $p$-values obtained by testing $H_{0k}$ on each sub-network. As a result of network splitting, however, these $p$-values may be mutually dependent, and so to combine them properly we use the concept of {\em stochastically larger than uniform} $p$-values defined as follows.

\begin{definition}[Stochastically larger than uniform~\citep{brannath2002recursive,rosenbaum2011some}]
A $K$-dimensional random vector $(P_1,\ldots, P_{K})$ with support in the $K$-dimensional unit cube is stochastically larger than uniform if for all $(p_1,\ldots, p_{K}) \in [0,1]^{K}$,
\[
\mathbbm P(P_1 \leq p_1, \ldots, P_{K} \leq p_{K}) \leq p_1\times\cdots\times p_{K}.
\]
\end{definition}
\vspace{-5px}
Independent $p$-values are trivially stochastically larger than uniform, but the above definition allows for possibly dependent $p$-values. 
A key technical challenge in our method is to split the network in a way such that 
the resulting $p$-values will be stochastically larger than uniform. On a high level, 
our strategy will be to, first, split the node set $V$ into possibly overlapping sets $V_1, \ldots, V_K$. Then, for all $k=1, \ldots, K-1$, we sequentially compute the $p$-value for $H_{0k}$ using information only from units in $V_1 \cup V_2 \cup \cdots \cup V_k$, as explained in Sections~\ref{sec:nonunifBRE} and~\ref{sec:method_clique}. Such sequential construction leads to stochastically larger than uniform $p$-values by the following lemma.
\begin{lemma}[\cite{rosenbaum2011some}, Lemma 3]\label{lem:recur_pval}
If, for all $k$, $P_k$ is a function of $(Q_1,\ldots, Q_k)$ and $\mathbbm P(P_k \leq p_k | Q_1,\ldots Q_{k-1}) \leq p_k$ for all $p_k \in [0,1]$ and for all $(Q_1,\ldots Q_{k-1})$, then $(P_1,\ldots, P_{K})$ is stochastically larger than uniform.
\end{lemma}

We can combine stochastically larger than uniform $p$-values into a single valid $p$-value in certain ways as if they were independent. For example,  Fisher's combination rule,
\begin{equation}\label{eq:pval_FCT}
    p_{\mathrm{FCT}} := 1 - F_{\chi^2_{2(K-1)}}\Big(-2\sum_{k=1}^{K-1} \log p_k\Big),
\end{equation}
leads to a valid combined $p$-value $p_{\mathrm{FCT}}$,
where $F_{\chi^2_{2(K-1)}}(\cdot)$ is the cumulative distribution function of a $\chi^2_{2(K-1)}$ distribution. That is, if $(p_k)_k$ are stochastically larger than uniform, then $\mathbbm P(p_{\mathrm{FCT}} \leq \alpha) \leq \alpha$ for all $\alpha \in [0,1]$. There are other ways apart from Fisher's rule to aggregate these $p$-values, such as Stouffer's and Cauchy's combination rules as well as Bonferroni's method. The Bonferroni's method is valid even without splitting the network to have stochastically larger than uniform $p$-values. We discuss other combination methods in Appendix~\ref{appdx:pvalcomb} and advocate for the use of Fisher's rule in our problem.

\section{Methodology under Non-Uniform Bernoulli Design}\label{sec:nonunifBRE}
In this section, we present our main method in settings where the design follows a non-uniform Bernoulli distribution, defined as follows.

\begin{definition}[Non-uniform Bernoulli Design]\label{def:nunifBRE}
The assignment mechanism $P(Z)$ satisfies $P(Z=z) = \prod_{i \in [N]} p_i^{z_i} (1 - p_i)^{1- z_i}$, for $z \in \{0,1\}^N$, where $p_i \in [0,1]$ are known treatment probabilities.
\end{definition}
\vspace{-5px}
The unit treatment probabilities in a non-uniform Bernoulli design can be arbitrary as long as they are known and fixed. For example, $p_i$ could depend on unit covariates, past outcomes, network, or other pre-treatment features of the unit and other units. The primary restriction is that units are treated independently. Obviously, the non-uniform Bernoulli design includes the usual Bernoulli design as a special case, under which the treatment probabilities are identical for all units.

Under the non-uniform Bernoulli design, we first introduce an algorithm to test a single contrast hypothesis $H_{0k}$ in Equation~\eqref{eq:null_mono_k}, as explained in Steps 2-3 of the general procedure of the previous section.
Before presenting our test for $H_{0k}$, we define certain graph-theoretic concepts that will clarify how we utilize the structure of non-uniform Bernoulli designs to build our randomization test, and how we can extend it towards more flexible designs.

\subsection{Preliminary concepts: Module and Module set}
We begin with the concepts of {\em module} and {\em module set} that are crucial in the construction of our tests.
\begin{definition}\label{def:modules}
A {\em module} is a subset of units $\mS \subseteq [N]$ that can be partitioned into two disjoint subsets, 
namely $\mS = 
\ef(\mS) \cup \er(\mS)$, such that any pair of units $i, j\in \ef(\mS)$ is disconnected ($A_{ij}=0$), and their neighborhoods are contained in $\er(\mS)$, i.e., $ \nei(\ef(\mS))\subseteq \er(\mS)$. 
We will refer to $\ef(\mS)$ as the set of {\em eligible focal units} of module $\mS$, and $\er(\mS)$ as the set of {\em eligible randomization units} of the module. 
A {\em module set} $\mathbb{S}$ is a collection of disjoint modules $\{\mS_1, \ldots, \mS_L \}$, i.e., $\mS_\ell \cap \mS_{\ell'} = \emptyset$ for all $\ell \neq\ell'$.
\end{definition}
\vspace{-5px}
These definitions are illustrated in Figure~\ref{fig:module}. In the figure, all eligible focal units in a module share the same neighbors and thus receive the same exposure under any treatment assignment. We call such a module a {\em uniform} module. Formally, a module $\mS$ is uniform if $\nei_i = \er(\mS)$ for all $i \in \ef(\mS)$. The requirement of disjoint modules in a module set simplifies the exposition and could be relaxed as we discuss later.

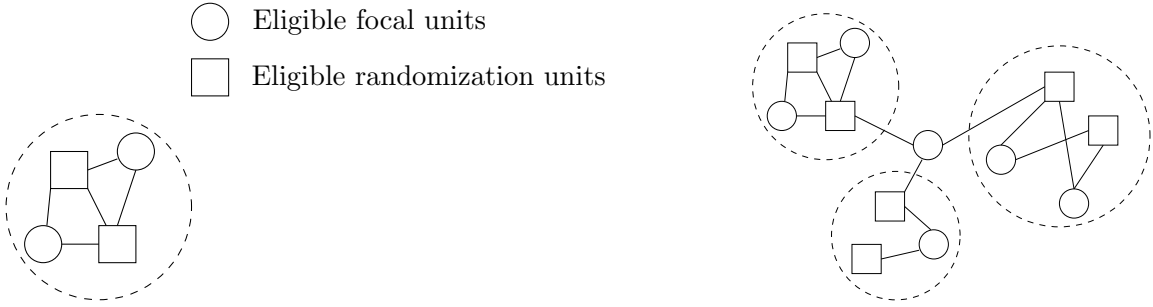
\begin{figure}[!t]
\centering
\resizebox{0.5\textwidth}{!}{%
\begin{circuitikz}
\tikzstyle{every node}=[font=\normalsize]
\draw  (7,8.5) circle (0.25cm);
\draw  (7.75,8.75) rectangle (8.25,8.25);
\draw  (7.1,9.75) rectangle (7.6,9.25);
\draw  (8.25,9.75) circle (0.25cm);

\draw [short] (7.6,9.25) -- (7.85, 8.75);
\draw [short] (7.1, 9.25) -- (7.05, 8.75);
\draw [short] (7.6, 9.5) -- (8, 9.65);

\draw [short] (8,8.75) -- (8.25,9.5);
\draw [short] (7.25,8.5) -- (7.75,8.5);
\draw [, dashed] (7.75,9) ellipse (1.25cm and 1.25cm);

\draw (9.25,11.5) circle (0.25cm);
\node [font=\small] at (11.4,11.5) {Eligible focal units};

\draw (9,11) rectangle (9.5,10.5);
\node [font=\small] at (12.2,10.75) {Eligible randomization units};
\end{circuitikz}
}
\resizebox{0.02\textwidth}{!}{%
\hspace{5px}
\begin{circuitikz}
\end{circuitikz}
}%
\resizebox{0.4\textwidth}{!}{%
\hspace{40px}
\begin{circuitikz}
\tikzstyle{every node}=[font=\normalsize]

\draw  (7,8.5) circle (0.25cm);
\draw  (7.75,8.75) rectangle (8.25,8.25);
\draw  (7.1,9.75) rectangle (7.6,9.25);
\draw  (12,7) circle (0.25cm);
\draw  (8.25,9.75) circle (0.25cm);
\draw  (10.75,7.75) circle (0.25cm);
\draw  (9.5,8) circle (0.25cm);
\draw  (12.25,8.5) rectangle (12.75,8);
\draw  (11.5,9.25) rectangle (12,8.75);

\draw  (9.6, 6.3) circle (0.25cm);
\draw  (8.2, 6.3) rectangle (8.7, 5.8);
\draw  (8.6, 7.2) rectangle (9.1, 6.7);
\draw [short] (9.1, 6.95) -- (9.55, 6.55);
\draw [short] (8.7, 6.05) -- (9.35, 6.2);
\draw [short] (9.1, 7.2) -- (9.4, 7.75);
\draw [, dashed] (8.95, 6.45) ellipse (1.1cm and 1.1cm);

\draw [short] (7.6,9.25) -- (7.85, 8.75);
\draw [short] (7.1, 9.25) -- (7.05, 8.75);
\draw [short] (7.6, 9.5) -- (8, 9.65);
\draw [short] (11,7.75) -- (12.25,8.25);
\draw [short] (11.75,8.75) -- (12,7.25);
\draw [short] (12,7.25) -- (12.5,8);
\draw [short] (10.75,8) -- (11.5,8.75);
\draw [, dashed] (11.75,8.15) ellipse (1.55cm and 1.55cm);

\draw [short] (8.25,8.5) -- (9.25,8);
\draw [short] (9.75,8) -- (11.5,9);
\draw [short] (8,8.75) -- (8.25,9.5);
\draw [short] (7.25,8.5) -- (7.75,8.5);
\draw [, dashed] (7.75,9) ellipse (1.25cm and 1.25cm);

\end{circuitikz}
}%
\caption{ {\em Left panel:} A uniform module $\mS$. Circles represent eligible focal units and squares represent eligible randomization units. {\em Right panel:} A module set $\mathbb S$ consisting of three uniform modules outlined by dashed circles.}
\label{fig:module}
\end{figure}

Constructing a module set is computationally straightforward. 
One approach is to start by randomly sampling a unit $j_1$ in $V$, define $\mS_1 = \{j_1\} \cup \nei_{j_1}$, and then 
calculate the remainder $V^{(1)} \leftarrow V \setminus  \{ \mS_1 \cup \nei(\mS_1)\}$. 
Then, sample another unit $j_2$ in $V^{(1)}$, define $\mS_2 = \{j_2\} \cup \nei_{j_2}$, calculate the 
remainder $V^{(2)} \leftarrow V^{(1)} \setminus  \{ \mS_2  \cup  \nei(\mS_2)\}$, and so on. 
The process terminates when $V^{(L)}$ is empty.
As a last step, we may augment each singleton $\ef(\mS_\ell) = \{j_\ell\}$ with units not in any modules and have exactly the same neighbors as $j_\ell$, for each $\ell=1,\ldots,L$. The resulting modules are uniform modules as well. 
In some applications where far more units always remain in control compared to those that could be treated, $\er(\mS)$ can be defined to include only the units that could be treated. See Appendix~\ref{appdx:implement_medellin} for more details.

Given a single contrast hypothesis $H_{0k}$ and population treatment assignment $z\in\{0,1\}^N$, the {\em active focal units} of a module $\mS$ are the eligible focal units that are exposed to the levels defined in the null 
hypothesis $H_{0k}$, namely
\begin{equation}
    \label{eq:active_focal}
    \af(z; \mS) := \Big\{i \in \ef(\mS): z_i = 0 ~\text{and}~ \expof_i(z) \in \{\expov_k, \expov_{k+1}\}\Big\}.
\end{equation}
A module $\mS$ is {\em active} under treatment $z$ if $\af(z; \mS) \neq \emptyset$. 
Similarly, we define the {\em active randomization units} of $\mS$  as $\ar(z; \mS) := \nei(\af(z; \mS)) \subseteq \er(\mS)$.

These definitions are illustrated in Figure~\ref{fig:activemodule} as a continuation of Figure~\ref{fig:module} under an observed treatment assignment $\Zobs$ with $\expov_k=1$ and $\expov_{k+1}=2$, indicating 1 and 2 treated neighbors respectively. Shaded circles and squares are treated under $\Zobs$. 
The module in the upper left panel is not active since under $\Zobs$ the eligible focal units of it have 
no treated neighbor. 
In the upper right module, only one eligible focal unit becomes active since the other focal unit is itself treated. The module in the bottom is active as well since the eligible focal unit is in control and has 2 treated neighbors.

We are now ready to define the testing procedure for the single contrast hypothesis $H_{0k}$. 

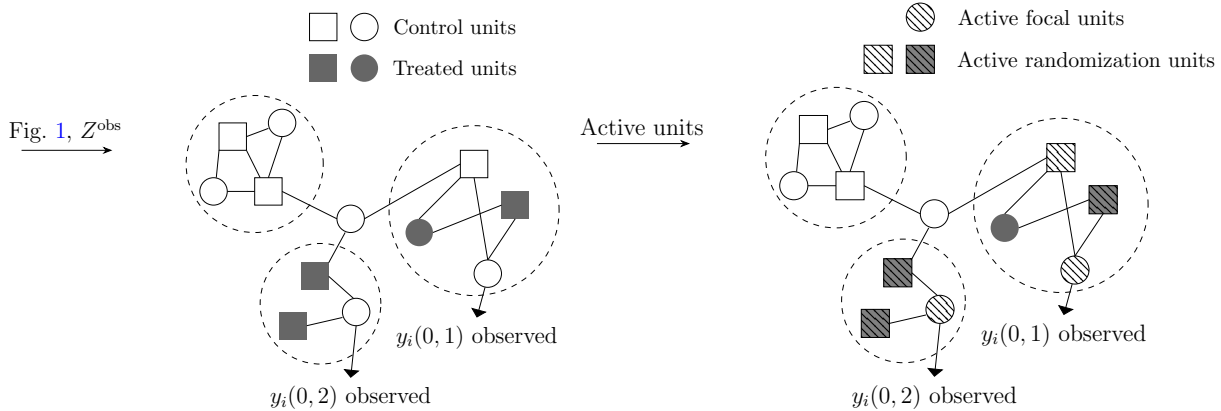
\begin{figure}[!t]
\centering
\resizebox{0.46\textwidth}{!}{%
\begin{circuitikz}
\tikzstyle{every node}=[font=\normalsize]

\node[above] at (4.3,9.25) {Fig.~\ref{fig:module}, $\Zobs$};
\draw [->, >=Stealth] (3.5,9.25) -- (5.15,9.25);

\draw  (7,8.5) circle (0.25cm);
\draw  (7.75,8.75) rectangle (8.25,8.25);
\draw  (7.1,9.75) rectangle (7.6,9.25);

\draw (12,7) circle (0.25cm);
\draw[-triangle 90] (11.95, 6.75) -- (11.8, 6.2) node[below] {$y_i(0,1)$ observed};
\draw  (8.25,9.75) circle (0.25cm);
\fill [gray!120] (10.75,7.75) circle (0.25cm);
\draw  (9.5,8) circle (0.25cm);
\fill [gray!120]  (12.25,8.5) rectangle (12.75,8);
\draw (11.5,9.25) rectangle (12,8.75);

\draw [short] (7.6,9.25) -- (7.85, 8.75);
\draw [short] (7.1, 9.25) -- (7.05, 8.75);
\draw [short] (7.6, 9.5) -- (8, 9.65);
\draw [short] (11,7.75) -- (12.25,8.25);
\draw [short] (11.75,8.75) -- (12,7.25);
\draw [short] (12,7.25) -- (12.5,8);
\draw [short] (10.75,8) -- (11.5,8.75);
\draw [, dashed] (11.75,8.15) ellipse (1.55cm and 1.55cm);

\draw  (9.6, 6.3) circle (0.25cm);
\fill [gray!120]  (8.2, 6.3) rectangle (8.7, 5.8);
\fill [gray!120] (8.6, 7.2) rectangle (9.1, 6.7);
\draw [short] (9.1, 6.95) -- (9.55, 6.55);
\draw [short] (8.7, 6.05) -- (9.35, 6.2);
\draw [short] (9.1, 7.2) -- (9.4, 7.75);
\draw [, dashed] (8.95, 6.45) ellipse (1.1cm and 1.1cm);
\draw[-triangle 90] (9.6, 6.05) -- (9.5,5.1) node[below] {$y_i(0,2)$ observed};

\draw [short] (8.25,8.5) -- (9.25,8);
\draw [short] (9.75,8) -- (11.5,9);
\draw [short] (8,8.75) -- (8.25,9.5);
\draw [short] (7.25,8.5) -- (7.75,8.5);
\draw [, dashed] (7.75,9) ellipse (1.25cm and 1.25cm);

\draw (9.75, 11.5) circle (0.25cm);
\draw (8.75, 11.75) rectangle (9.25, 11.25);
\node [font=\small] at (11.4,11.5) {Control units};

\fill [gray!120] (9.75, 10.75) circle (0.25cm);
\fill [gray!120] (8.75, 11) rectangle (9.25,10.5);
\node [font=\small] at (11.4,10.75) {Treated units};

\end{circuitikz}
}%
\resizebox{0.53\textwidth}{!}{%
\begin{circuitikz}
\tikzstyle{every node}=[font=\normalsize]

\node[above] at (4.3,9.25) {Active units};
\draw [->, >=Stealth] (3.5,9.25) -- (5.15,9.25);

\draw  (7,8.5) circle (0.25cm);
\draw  (7.75,8.75) rectangle (8.25,8.25);
\draw  (7.1,9.75) rectangle (7.6,9.25);

\filldraw[pattern=north west lines] (12,7) circle (0.25cm);
\draw[-triangle 90] (11.95, 6.75) -- (11.8, 6.2) node[below] {$y_i(0,1)$ observed};
\draw  (8.25,9.75) circle (0.25cm);
\fill [gray!120] (10.75,7.75) circle (0.25cm);
\draw  (9.5,8) circle (0.25cm);
\filldraw[pattern=north west lines][preaction={fill=gray!100}]  (12.25,8.5) rectangle (12.75,8);
\filldraw[pattern=north west lines] (11.5,9.25) rectangle (12,8.75);

\draw [short] (7.6,9.25) -- (7.85, 8.75);
\draw [short] (7.1, 9.25) -- (7.05, 8.75);
\draw [short] (7.6, 9.5) -- (8, 9.65);
\draw [short] (11,7.75) -- (12.25,8.25);
\draw [short] (11.75,8.75) -- (12,7.25);
\draw [short] (12,7.25) -- (12.5,8);
\draw [short] (10.75,8) -- (11.5,8.75);
\draw [, dashed] (11.75,8.15) ellipse (1.55cm and 1.55cm);

\filldraw[pattern=north west lines]  (9.6, 6.3) circle (0.25cm);
\filldraw[pattern=north west lines][preaction={fill=gray!100}]  (8.2, 6.3) rectangle (8.7, 5.8);
\filldraw[pattern=north west lines][preaction={fill=gray!100}] (8.6, 7.2) rectangle (9.1, 6.7);
\draw [short] (9.1, 6.95) -- (9.55, 6.55);
\draw [short] (8.7, 6.05) -- (9.35, 6.2);
\draw [short] (9.1, 7.2) -- (9.4, 7.75);
\draw [, dashed] (8.95, 6.45) ellipse (1.1cm and 1.1cm);
\draw[-triangle 90] (9.6, 6.05) -- (9.5,5.1) node[below] {$y_i(0,2)$ observed};

\draw [short] (8.25,8.5) -- (9.25,8);
\draw [short] (9.75,8) -- (11.5,9);
\draw [short] (8,8.75) -- (8.25,9.5);
\draw [short] (7.25,8.5) -- (7.75,8.5);
\draw [, dashed] (7.75,9) ellipse (1.25cm and 1.25cm);

\filldraw[pattern=north west lines]  (9.25,11.5) circle (0.25cm);
\node [font=\small] at (11.4,11.5) {Active focal units};

\filldraw[pattern=north west lines][preaction={fill=gray!100}] (9,11) rectangle (9.5,10.5);
\filldraw[pattern=north west lines] (8.25,11) rectangle (8.75,10.5);
\node [font=\small] at (12.2,10.75) {Active randomization units };

\end{circuitikz}
}%
\caption{Continuing from Figure~\ref{fig:module}. 
{\em Left panel:} $\Zobs$ is realized. Shaded color indicates treated units. {\em Right panel:} Nodes with hatched pattern represent active focal/randomization units for the null hypothesis $H_{0k}: y_i(0,1) \ge y_i(0,2)$. 
} 
\label{fig:activemodule}
\end{figure}

\subsection{Testing a single contrast hypothesis, $H_{0k}$}\label{sec:single}
To build intuition, we first present Algorithm~\ref{algo:focalrand_single} that describes in detail the steps for testing $H_{0k}$ assuming that each module in our analysis is uniform. As we will see, the use of uniform modules simplifies the  randomization distribution to a simple (clustered) Bernoulli randomization, as if we were randomizing the exposures $\expov_k, \expov_{k+1}$ on the active focal units of each module. 

The first step of Algorithm~\ref{algo:focalrand_single} (Lines 1-2) is to find the active focal and active randomization units under $\Zobs$, which define the index set of active modules denoted by $\mathcal L^\obs$. 
Note that $\mathcal L^\obs$ depends on the observed $\Zobs$, and is thus a random variable.
Lines 3-6 calculate the distributions of exposures on active focal units conditional on $\mathcal L^\obs$ being the index set of active modules. Specifically, $p_\ell$ in Line 5 is the probability that units in $\af(\Zobs; \mS_\ell)$ are exposed to $\expov_{k+1}$ conditional on $\mS_\ell$ being active with active focal units $\af(\Zobs; \mS_\ell)$.
Finally, Lines 7-14 define the main randomization test that ``shuffles" the exposures $\{\expov_k, \expov_{k+1}\}$ on the active focal units. Akin to cluster randomization, this can be simply implemented by, first, performing a Bernoulli randomization on the module level with probability $p_\ell$ (Line 10), 
and then jointly setting the exposures of all focal units in each module according to that Bernoulli randomization (Line 11). We defer the proof of validity to Theorem~\ref{thm:focaltest_mono} below.

\begin{algorithm}[!t]
\caption{Test for single contrast hypothesis, $H_{0k}:y_i(0, \expov_k) ~\geq~ y_i(0, \expov_{k+1})$, assuming uniform modules}
\label{algo:focalrand_single}
\singlespacing
\begin{algorithmic}[1]
\vspace{-0.6cm}
\REQUIRE Module set $\mathbb S = \{\mS_1, \ldots, \mS_L\}$ such that each $\mS \in \mathbb S$ is a uniform module; observed treatment $\Zobs$; observed outcome $\Yobs$ (Input).

\hspace{-24px} {\bf Output:} Finite-sample valid $p$-value for $H_{0k}$, $\pval_k$.
\item[] {\em // Preprocessing}
\STATE Given $\Zobs$, calculate $\af(\Zobs; \mS_\ell)$ and $\ar(\Zobs; \mS_\ell)$ for each $\ell \in [L]$ 
from Equation~\eqref{eq:active_focal}. Define the index set of active modules $\mathcal L^\obs = \{ \ell \in [L]: \af(\Zobs; \mS_\ell) \neq \emptyset \}$ and the set of active focal units $\mU^\obs = \bigcup_{\ell \in [L]} \af(\Zobs; \mS_\ell)$.
\STATE Define an exposure-monotone test statistic in the order $(\expov_k, \expov_{k+1})$; e.g., difference-in-means as in Equation~\eqref{eq:diff-in-mean} as $t_k(z, y) := t_{\expov_k, \expov_{k+1}}^{\mathrm{DiM}}(z, y; \mU^\obs)$. 
Calculate the observed value of the test statistic $\Tobs = t_k(\Zobs, \Yobs)$. 
\item[]
\item[] {\em // Calculate the conditional randomization distribution on active modules}
\FOR{$\ell \in \mathcal{L}^\obs$}
    \STATE Define $\mathcal T_\ell(\expov) = \big\{ z_{\mathsf{A}} \in \{0,1\}^{|\mathsf{A}|}: \expof_i(z_{\mathsf{A}}) = \expov,~ \forall i \in \af(\Zobs; \mS_\ell) \big\}$, for $\expov \in \{\expov_k, \expov_{k+1}\}$,
    where $\mathsf{A} = \ar(\Zobs; \mS_\ell)$.
    \STATE Calculate 
    \[
    p_{0\ell} = \sum_{z\in\mathcal{T}_\ell(\expov_{k})} \prod_{j \in \ar(\Zobs; \mS_\ell)} p_j^{z_j}(1-p_j)^{1-z_j}, ~p_{1\ell} = \sum_{z\in\mathcal{T}_\ell(\expov_{k+1})} \prod_{j \in \ar(\Zobs; \mS_\ell)} p_j^{z_j}(1-p_j)^{1-z_j},
    \]
    and $p_\ell = p_{1\ell}/ (p_{0\ell}+p_{1\ell})$.
\ENDFOR
\item[]
\item[]{\em // Main randomization test}
\FOR{$r=1, \ldots R$}
\STATE $Z^{(r)} \leftarrow \Zobs$.
\FOR{$\ell \in \mathcal L^\obs$}
    \STATE Sample $\iota_\ell \sim \mathrm{Bern}(p_\ell)$.
    \STATE 
    If $\iota_\ell=1$, then sample $\widetilde Z_\ell \sim \mathrm{Unif}\{\mathcal{T}_\ell(\expov_{k+1})\}$; otherwise sample $\widetilde Z_\ell \sim \mathrm{Unif}\{\mathcal{T}_\ell(\expov_{k})\}$. Update $Z^{(r)}_{i} \leftarrow \widetilde Z_i$ for all $i\in\ar(\Zobs; \mS_\ell)$.
\ENDFOR
\STATE Calculate $T^{(r)} = t_k(Z^{(r)}, \Yobs)$.
\ENDFOR
\STATE Output $p$-value: 
\begin{equation}\label{eq:pval_1}
\pval_k = \frac{1}{1+R}\Big(1 + \sum_{r=1}^R \mathbbm{1}\{ T^{(r)} \geq T^\obs \}\Big). 
\end{equation}
\end{algorithmic}
\end{algorithm}

An illustration of Algorithm~\ref{algo:focalrand_single} is shown in Figure~\ref{fig:algo1}, which continues the setup introduced earlier in Figure~\ref{fig:activemodule}. 
The right panel shows a possible randomized treatment $Z^{(r)}$ (Lines 8-12). The $\iota_\ell$ from the Bernoulli randomization in Line 10 is realized to be $1$ for the upper right module, while realized to be $0$ for the bottom module. Hence, the active focal unit ``b" is exposed to 2 treated neighbors while the active focal unit ``a" is exposed to 1 treated neighbor.
The resulting randomized treatment $Z^{(r)}$ permutes the exposures of node ``a" and ``b" under $\Zobs$. If, for example, the difference-in-means test statistic~\eqref{eq:diff-in-mean} is used with both $\psi_1$ and $\psi_0$ being the identity function, then $\Tobs = \Yobs_{\text{a}} - \Yobs_{\text{b}}$ and $T^{(r)} = \Yobs_{\text{b}} - \Yobs_{\text{a}}$.

\begin{figure}[t!]
\centering
\resizebox{0.405\textwidth}{!}{
\begin{circuitikz}
\tikzstyle{every node}=[font=\normalsize]

\draw (7,8.5) circle (0.25cm);
\draw (7.75,8.75) rectangle (8.25,8.25);
\draw (7.1,9.75) rectangle (7.6,9.25);
\filldraw[pattern=north west lines] (12,7) circle (0.25cm);
\node [font=\small] at (12.4,7.2) {b};
\draw  (8.25,9.75) circle (0.25cm);
\fill [gray!120] (10.75,7.75) circle (0.25cm);
\draw  (9.5,8) circle (0.25cm);
\filldraw[pattern=north west lines][preaction={fill=gray!100}]  (12.25,8.5) rectangle (12.75,8);
\filldraw[pattern=north west lines] (11.5,9.25) rectangle (12,8.75);

\draw [short] (11,7.75) -- (12.25,8.25);
\draw [short] (11.75,8.75) -- (12,7.25);
\draw [short] (12,7.25) -- (12.5,8);
\draw [short] (10.75,8) -- (11.5,8.75);
\draw [, dashed] (11.75,8.15) ellipse (1.55cm and 1.55cm);

\filldraw[pattern=north west lines]  (9.6, 6.3) circle (0.25cm);
\node [font=\small] at (10, 6.5) {a};
\filldraw[pattern=north west lines][preaction={fill=gray!100}]  (8.2, 6.3) rectangle (8.7, 5.8);
\filldraw[pattern=north west lines][preaction={fill=gray!100}] (8.6, 7.2) rectangle (9.1, 6.7);
\draw [short] (9.1, 6.95) -- (9.55, 6.55);
\draw [short] (8.7, 6.05) -- (9.35, 6.2);
\draw [short] (9.1, 7.2) -- (9.4, 7.75);
\draw [, dashed] (8.95, 6.35) ellipse (1.2cm and 1.2cm);

\draw [short] (7.6,9.25) -- (7.85, 8.75);
\draw [short] (7.1, 9.25) -- (7.05, 8.75);
\draw [short] (7.6, 9.5) -- (8, 9.65);
\draw [short] (8.25,8.5) -- (9.25,8);
\draw [short] (9.75,8) -- (11.5,9);
\draw [short] (8,8.75) -- (8.25,9.5);
\draw [short] (7.25,8.5) -- (7.75,8.5);
\draw [, dashed] (7.75,9) ellipse (1.25cm and 1.25cm);

\node [font=\Large] at (10,13) {Under $\Zobs$};
\node [font=\small] at (12.4,11.5) {Active focal units};
\node [font=\small] at (13.2,10.75) {Active randomization units};
\filldraw[pattern=north west lines]  (10.25,11.5) circle (0.25cm);
\filldraw[pattern=north west lines][preaction={fill=gray!100}] (10,11) rectangle (10.5,10.5);
\filldraw[pattern=north west lines] (9.25,11) rectangle (9.75,10.5);
\end{circuitikz}
}
\resizebox{0.57\textwidth}{!}{
\begin{circuitikz}
\tikzstyle{every node}=[font=\normalsize]

\node[above] at (4.3,9.25) {randomization};
\draw [->, >=Stealth] (3.5,9.25) -- (5.15,9.25);

\draw (7,8.5) circle (0.25cm);
\draw (7.75,8.75) rectangle (8.25,8.25);
\draw (7.1,9.75) rectangle (7.6,9.25);
\filldraw[pattern=north west lines] (12,7) circle (0.25cm);
\node [font=\small] at (12.4,7.2) {b};
\draw  (8.25,9.75) circle (0.25cm);
\fill [gray!120] (10.75,7.75) circle (0.25cm);
\draw  (9.5,8) circle (0.25cm);
\filldraw[pattern=north west lines][preaction={fill=gray!100}] (12.25,8.5) rectangle (12.75,8);
\filldraw[pattern=north west lines][preaction={fill=gray!100}] (11.5,9.25) rectangle (12,8.75);

\draw [short] (11,7.75) -- (12.25,8.25);
\draw [short] (11.75,8.75) -- (12,7.25);
\draw [short] (12,7.25) -- (12.5,8);
\draw [short] (10.75,8) -- (11.5,8.75);
\draw [, dashed] (11.75,8.15) ellipse (1.55cm and 1.55cm);

\filldraw[pattern=north west lines]  (9.6, 6.3) circle (0.25cm);
\node [font=\small] at (10, 6.5) {a};
\filldraw[pattern=north west lines]  (8.2, 6.3) rectangle (8.7, 5.8);
\filldraw[pattern=north west lines][preaction={fill=gray!100}] (8.6, 7.2) rectangle (9.1, 6.7);
\draw [short] (9.1, 6.95) -- (9.55, 6.55);
\draw [short] (8.7, 6.05) -- (9.35, 6.2);
\draw [short] (9.1, 7.2) -- (9.4, 7.75);
\draw [, dashed] (8.95, 6.35) ellipse (1.2cm and 1.2cm);

\draw [short] (7.6,9.25) -- (7.85, 8.75);
\draw [short] (7.1, 9.25) -- (7.05, 8.75);
\draw [short] (7.6, 9.5) -- (8, 9.65);
\draw [short] (8.25,8.5) -- (9.25,8);
\draw [short] (9.75,8) -- (11.5,9);
\draw [short] (8,8.75) -- (8.25,9.5);
\draw [short] (7.25,8.5) -- (7.75,8.5);
\draw [, dashed] (7.75,9) ellipse (1.25cm and 1.25cm);

\node [font=\Large] at (10,12.8) {Under $Z^{(r)}$};
\node [font=\small] at (12.4,11.5) {Active focal units};
\node [font=\small] at (13.2,10.75) {Active randomization units};
\filldraw[pattern=north west lines]  (10.25,11.5) circle (0.25cm);
\filldraw[pattern=north west lines][preaction={fill=gray!100}] (10,11) rectangle (10.5,10.5);
\filldraw[pattern=north west lines] (9.25,11) rectangle (9.75,10.5);
\end{circuitikz}
}
\caption{Illustration of Algorithm \ref{algo:focalrand_single} in testing $H_{0k}: y_i(0,1) \ge y_i(0,2)$.
{\em Left panel:} 
The observed outcome for node ``a" is $\Yobs_{\text{a}} = y_{\text{a}}(0,2)$ and the observed outcome for node ``b" is $\Yobs_{\text{b}} = y_{\text{b}}(0,1)$. {\em Right panel:} The randomized treatment $Z^{(r)}$ happens to permute the exposures of the two active focal units ``a" and ``b".} 
\label{fig:algo1}
\end{figure}
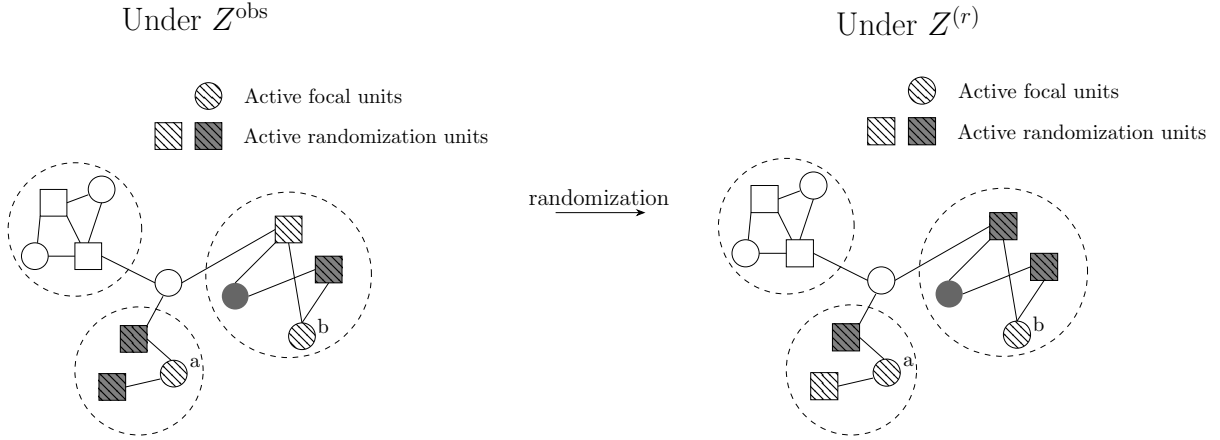

\subsubsection{Extensions of Algorithm~\ref{algo:focalrand_single}}\label{sec:extalgo1}
Before moving to test the full monotone hypothesis, here we discuss some extensions of Algorithm~\ref{algo:focalrand_single}. Further extensions to more general exposure functions and designs are presented in Appendix~\ref{appdx:gen_algo_s}. Notably, we present a way to allow certain overlaps between modules in a module set there, which could be useful in a denser network.

\paragraph{Allowing non-uniform modules.}~It is straightforward to lift the requirement of uniform modules in Algorithm~\ref{algo:focalrand_single}, and thus allowing eligible focal units to have different neighbors. 
The main challenge with such an extension is that the exposures of active focal units in the same module may not take the same value,
as opposed to Lines 4-5 and 10-11 of Algorithm~\ref{algo:focalrand_single}. Instead, the randomization distribution for $\widetilde Z_\ell$ is the conditional distribution of the non-uniform Bernoulli 
design $P$ on $\mathsf{A} = \ar(\Zobs; \mS_\ell)$ conditional on being in the set
\begin{equation}\label{eq:condZset}
    \big\{ z_{\mathsf{A}} \in \{0,1\}^{|\mathsf{A}|}: \expof_i(z_\mathsf{A}) \in \{\expov_k, \expov_{k+1}\},~\forall i \in \af(\Zobs; \mS_\ell) \big\}.
\end{equation}
This set can be enumerated whenever the space of active randomization units, namely $|\ar(\Zobs; \mS_\ell)|$, is not too large.
Alternatively, we could use rejection sampling or similar schemes to approximate the randomization distribution to arbitrary precision. Note also that~\eqref{eq:condZset} allows for other exposure functions apart from~\eqref{eq:neightreated_cont} as long as the exposure depends only on the treatments of a unit's neighbors.

\paragraph{Choice of module set.}~Algorithm~\ref{algo:focalrand_single} gives a finite-sample valid test for any (uniform) module set, but it remains unclear which module set to use. 
Based on Theorem 3 in~\cite{puelz2022graph}, which shows that the power of a conditional randomization test generally increases with both the number of focal units and the support of the conditional randomization distribution, it is better to choose a module set that gives more active focal units. However, we cannot simply maximize the observed number of active focal units under $\Zobs$, as this could raise selective inference problems. Instead, one valid approach would be to  maximize the {\em expected} number of active focal units with respect to the design from which $\Zobs$ is sampled, either through 
exact calculations or Monte-Carlo.\label{para:temp}

\subsection{Testing the full monotone hypothesis, $H_0$}
Here we present our randomization test for the full monotone hypothesis $H_0$ in~\eqref{eq:null_mono} under non-uniform Bernoulli designs. To build intuition, we begin with a simple procedure that splits the network ``far enough" to eliminate dependence between $p$-values.
We will then follow up with a more sophisticated procedure that takes the network structure into account.

\paragraph{A simple testing procedure.}~The $p$-value from Algorithm~\ref{algo:focalrand_single} is valid for testing the single contrast hypothesis $H_{0k}$. 
Thus, a straightforward way to test the full monotone null hypothesis $H_0$ is to, first, construct $K-1$ non-overlapping module sets $\mathbb S_k = \{\mS_{k, \ell}: \ell \in [L_k]\}$ for $k=1, \ldots, K-1$, such that $\mS_{k, \ell} \cap \mS_{k', \ell'} = \emptyset$ for all $k \neq k'$, $\ell \in [L_k]$ and $\ell' \in [L_{k'}]$. 
Then, we test each contrast $H_{0k}$ by Algorithm~\ref{algo:focalrand_single} applied on $\mathbb S_k$ and get $\mathrm{pval}_k$.
Since $\mathrm{pval}_k$ is only a function of data from the module set $\mathbb S_k$, these $p$-values are mutually independent under a non-uniform Bernoulli design, and can thus be combined easily into a valid $p$-value for $H_0$.
While this approach is conceptually simple, the obvious downside is that we may  discard too much network information to construct non-overlapping sub-networks.

\paragraph{A flexible testing procedure.}~A more flexible approach is to split the network in a way that allows possibly overlapping sub-networks inspired by Lemma~\ref{lem:recur_pval}. The idea is to split the network sequentially from $k=1$ to $K-1$, where 
at each step $k$ we allow $\mathbb S_k$ to have eligible randomization units that overlap with previous sub-networks. Reflecting this change, the randomization test on $\mathbb S_k$ would need to be adapted by conditioning on treatments of all previously constructed sub-networks at their observed values under $\Zobs$.

This adaptation of the randomization test requires a small extension of Algorithm~\ref{algo:focalrand_single} to allow conditioning on the treatments of certain units. This is presented in Algorithm~\ref{algo:focalrand_single_general} that takes a set of units $C$ as input on which the randomization test is conditioned (Line 3), while allowing for non-uniform modules. 
To test $H_0$, we can now apply Algorithm~\ref{algo:focalrand_single_general} sequentially to construct multiple $p$-values to be combined through Fisher's combination. 

\renewcommand{\thealgorithm}{1$'$}
\begin{algorithm}[t!]
\caption{General test for single contrast hypothesis, $H_{0k}:y_i(0, \expov_k) ~\geq~ y_i(0, \expov_{k+1})$}
\label{algo:focalrand_single_general}
\begin{algorithmic}[1]
\REQUIRE Module set $\mathbb S = \{\mS_1, \ldots, \mS_L\}$;  $\Zobs$; $\Yobs$; conditional set $C$ (Input).

\hspace{-24px} {\bf Output:} Finite-sample valid $p$-value for $H_{0k}$, $\pval_k$.
\item[] {\em // Preprocessing}
\STATE Same as Algorithm~\ref{algo:focalrand_single}.
\item[] {\em // Calculate conditional randomization distribution on active modules}
\FOR{$\ell \in \mathcal{L}^\obs$}
    \STATE Calculate the randomization distribution for active randomization units supported on $\{0,1\}^{|\mathsf{A}|}$ where $\mathsf{A} = \ar(\Zobs;\mS_\ell)$, conditional on treatment status of units in $C$:
    \begin{equation}\label{eq:rdist_single}
    r_{\ell, k}(z_{\mathsf{A}}) ~\propto~ \mathbbm 1\{\expof_i(z_{\mathsf{A}}) \in \{\expov_k, \expov_{k+1}\},~\forall i \in \af(\Zobs; \mS_\ell) \} P_\ell(z_{\mathsf{A}}|C),
    \end{equation}
    where under the non-uniform Bernoulli design
    \[
    P_\ell(z_{\mathsf{A}}| C) \propto \mathbbm 1\{z_{\mathsf{A} \cap C} = \Zobs_{\mathsf{A} \cap C}\} \prod_{j \in \mathsf{A} \setminus C} p_j^{z_j} (1-p_j)^{1-z_j}.
    \]
    \vspace{-0.8cm}
\ENDFOR
\item[] {\em // Main randomization test}
\STATE Same as Algorithm~\ref{algo:focalrand_single}, except that Lines 10-11 are changed to sample $\widetilde Z_\ell$ from $r_{\ell, k}(\cdot)$.
\end{algorithmic}
\end{algorithm}
\renewcommand{\thealgorithm}{\arabic{algorithm}}
\addtocounter{algorithm}{-1}

The procedure to test $H_0$ is shown in Algorithm~\ref{algo:focalrand_multiple}, and can be described as follows. For each $k$, in Line 2, we first construct module set $\mathbb S_k$ such that the eligible focal units for each module in $\mathbb S_k$ do not overlap with $\mathbb S_{<k}$, the units in all module sets constructed prior to step $k$.
Importantly, to gain more power, we allow overlap between the eligible randomization units across module sets. That is, for $k' < k$, both eligible focal units and eligible randomization units built at step $k'$ can become eligible randomization units for the modules built at step $k$. 
The potential benefit in power is illustrated in the left panel of Figure~\ref{fig:Focalalgo_multiple_graph}, where the grey square (labeled as ``c") is a treated eligible randomization unit from a module in $\mathbb{S}_1$ that is re-used in the subsequent module set $\mathbb{S}_2$ as another eligible randomization unit. In contrast, under the simple procedure described above, this unit 
could not be included in $\mathbb S_2$, so that the left module cannot be formed.
Next, in Line 3 we apply the single contrast hypothesis test (Algorithm~\ref{algo:focalrand_single_general}) on module set $\mathbb S_k$ conditioning on the treatments in set $C=\mathbb S_{<k}$. As a result, in the right panel of Figure~\ref{fig:Focalalgo_multiple_graph}, the randomization test conditions on unit ``c" being treated.
Finally, Line 5 combines the $p$-values $(\pval_k)_{k=1}^{K-1}$ by Fisher's rule. 
Theorem~\ref{thm:focaltest_mono} below shows that these $p$-values are stochastically larger than uniform, 
and thus the combined $p$-value from Algorithm~\ref{algo:focalrand_multiple} leads to a finite-sample valid $p$-value for $H_0$.
The proof is in Appendix~\ref{a:proof}.

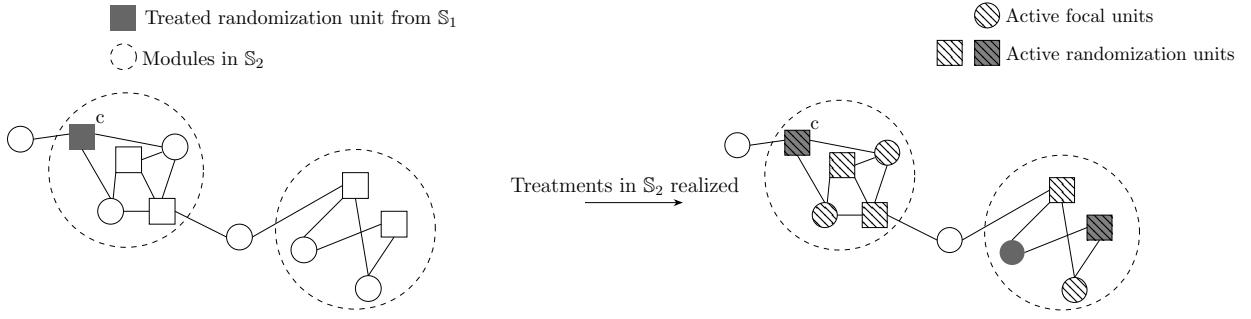
\begin{figure}[t!]
\begin{minipage}{0.4\textwidth}
\resizebox{0.95\textwidth}{!}{
\begin{circuitikz}
\tikzstyle{every node}=[font=\normalsize]
\draw  (7,8.5) circle (0.25cm);
\draw  (7.75,8.75) rectangle (8.25,8.25);
\draw  (7.1,9.75) rectangle (7.6,9.25);
\draw  (12,7) circle (0.25cm); 
\draw  (8.25,9.75) circle (0.25cm);
\draw  (10.75,7.75) circle (0.25cm);
\draw  (9.5,8) circle (0.25cm);
\draw  (12.25,8.5) rectangle (12.75,8);
\draw  (11.5,9.25) rectangle (12,8.75);
\fill [gray!120] (6.2,10.2) rectangle (6.7,9.7); 
\node [font=\small] at (6.8, 10.3) {c};
\draw  (5.25,9.9) circle (0.25cm); 

\draw [short] (7.6,9.25) -- (7.85, 8.75);
\draw [short] (7.1, 9.25) -- (7.05, 8.75);
\draw [short] (7.6, 9.5) -- (8, 9.65);
\draw [short] (11,7.75) -- (12.25,8.25);
\draw [short] (11.75,8.75) -- (12,7.25);
\draw [short] (12,7.25) -- (12.5,8);
\draw [short] (10.75,8) -- (11.5,8.75);
\draw [, dashed] (11.75,8.15) ellipse (1.55cm and 1.55cm); 

\draw [short] (8.25,8.5) -- (9.25,8);
\draw [short] (9.75,8) -- (11.5,9);
\draw [short] (8,8.75) -- (8.25,9.5);
\draw [short] (7.25,8.5) -- (7.75,8.5);
\draw [short] (6.45,9.7) -- (7,8.75); 
\draw [short] (6.7,10) -- (8,9.75); 
\draw [short] (6.2,10) -- (5.5,9.9); 
\draw [, dashed] (7.3, 9.3) ellipse (1.5cm and 1.5cm); 

\node [font=\small] at (10.7, 12.25) {Treated randomization unit from $\mathbb S_1$ }; 
\node [font=\small] at (8.8,11.45) {Modules in $\mathbb S_{2}$};
\fill [gray!120] (7,12.5) rectangle (7.5,12);
\draw [, dashed] (7.25,11.5) circle (0.25cm);
\end{circuitikz}
}%
\end{minipage}
\begin{minipage}{0.6\textwidth}
\resizebox{1\textwidth}{!}{%
\begin{circuitikz}
\tikzstyle{every node}=[font=\normalsize]

\node[above] at (3,8.75) {Treatments in $\mathbb S_2$ realized}; 
\draw [->, >=Stealth] (2.2,8.75) -- (4.15,8.75);

\filldraw[pattern=north west lines] (7,8.5) circle (0.25cm);
\filldraw[pattern=north west lines] (7.75,8.75) rectangle (8.25,8.25);
\filldraw[pattern=north west lines] (7.1,9.75) rectangle (7.6,9.25);
\filldraw[pattern=north west lines] (12,7) circle (0.25cm);
\filldraw[pattern=north west lines] (8.25,9.75) circle (0.25cm);
\fill[gray!120] (10.75,7.75) circle (0.25cm);
\draw (9.5,8) circle (0.25cm);
\filldraw[pattern=north west lines][preaction={fill=gray!100}] (12.25,8.5) rectangle (12.75,8);
\filldraw[pattern=north west lines] (11.5,9.25) rectangle (12,8.75);
\filldraw[pattern=north west lines][preaction={fill=gray!100}] (6.2,10.2) rectangle (6.7,9.7); 
\draw  (5.25,9.9) circle (0.25cm); 
\node [font=\small] at (6.8, 10.3) {c};

\draw [short] (11,7.75) -- (12.25,8.25);
\draw [short] (11.75,8.75) -- (12,7.25);
\draw [short] (12,7.25) -- (12.5,8);
\draw [short] (10.75,8) -- (11.5,8.75);
\draw [, dashed] (11.75,8.15) ellipse (1.55cm and 1.55cm);

\draw [short] (7.6,9.25) -- (7.85, 8.75);
\draw [short] (7.1, 9.25) -- (7.05, 8.75);
\draw [short] (7.6, 9.5) -- (8, 9.65);
\draw [short] (8.25,8.5) -- (9.25,8);
\draw [short] (9.75,8) -- (11.5,9);
\draw [short] (8,8.75) -- (8.25,9.5);
\draw [short] (7.25,8.5) -- (7.75,8.5);
\draw [short] (6.45,9.7) -- (7,8.75); 
\draw [short] (6.7,10) -- (8,9.75); 
\draw [short] (6.2,10) -- (5.5,9.9); 
\draw [, dashed] (7.3,9.3) ellipse (1.5cm and 1.5cm); 

\node [font=\small] at (12.1, 12.5) {Active focal units};
\node [font=\small] at (12.92, 11.75) {Active randomization units};

\filldraw[pattern=north west lines]  (10.25,12.5) circle (0.25cm);
\filldraw[pattern=north west lines][preaction={fill=gray!100}] (10,12) rectangle (10.5,11.5);
\filldraw[pattern=north west lines] (9.25,12) rectangle (9.75,11.5);
\end{circuitikz}
}%
\end{minipage}
\caption{Illustration of Algorithm \ref{algo:focalrand_multiple} showing overlap between module sets.
{\em Left:} Each dashed circle represents a module in $\mathbb S_2$, while the shaded square, labeled as ``c", is an eligible randomization unit that overlaps with those in $\mathbb S_1$, and is realized to be treated.
{\em Right:} Treatments in $\mathbb S_2$ are realized. The randomization test is then conditional on unit ``c" being treated.}
\label{fig:Focalalgo_multiple_graph}
\end{figure}

\begin{algorithm}[t!]
\caption{Test for monotone hypothesis, $H_0: y_i(0, \expov_1) \geq y_i(0, \expov_2)\geq\cdots \geq y_i(0, \expov_K)$}
\label{algo:focalrand_multiple}
\begin{algorithmic}[1]
\REQUIRE Exposure set $\exposet = \{\expov_1, \expov_2,\ldots, \expov_K\}$; observed treatment $\Zobs$; observed outcome $\Yobs$ (Input).

\hspace{-24px} {\bf Output:} Finite-sample valid $p$-value for $H_{0}$, $p_{\mathrm{FCT}}$.
\FOR{$k = 1,2,\ldots, K-1$}
    \STATE Construct module set $\mathbb S_k = \{\mS_{k,1}, \ldots, \mS_{k, L_k}\}$ such that $\ef(\mS_{k, \ell}) \cap \mathbb S_{<k} = \emptyset$ for all $\ell \in [L_k]$, where $\mathbb S_{<k} := \bigcup_{k'<k} \bigcup_{\ell' \in [L_{k'}]} \mS_{k', \ell'}$ with the convention $\mathbb S_{<1} := \emptyset$.
    \STATE Apply Algorithm~\ref{algo:focalrand_single_general} with module set $\mathbb S_k$, observed treatment $\Zobs$ and outcome $\Yobs$, and conditional set $\mathbb S_{<k}$. Get $p$-value $\pval_k$.
\ENDFOR
\STATE Output the sequence of $p$-values $(\pval_k)_{k=1}^{K-1}$ and the combined $p$-value $p_{\mathrm{FCT}}$ using~\eqref{eq:pval_FCT}.
\end{algorithmic}
\end{algorithm}

\begin{theorem}\label{thm:focaltest_mono}
Suppose that the treatment assignment follows a non-uniform Bernoulli design 
in Definition \ref{def:nunifBRE}. 
Then, each $p$-value from Algorithm~\ref{algo:focalrand_multiple}, $\pval_k$, is valid for  $H_{0k}: y_i(0, \expov_k) \ge y_i(0, \expov_{k+1})~\forall i \in [N]$.
Moreover, these $p$-values are stochastically larger than uniform, so the combined $p$-value from Line 5 of Algorithm~\ref{algo:focalrand_multiple} is valid under the monotone spillover hypothesis $H_0$. Specifically, 
\[
\mathbbm P\big(p_{\mathrm{FCT}} \leq \alpha ~|~ H_{0} \big) \leq \alpha,~\text{for all}~\alpha \in [0,1],
\]
where the probability is with respect to treatment randomization from the Bernoulli design.
\end{theorem}

\begin{remark}[Covariate adjustment]\label{rem:covadj}
At each step $k$ in Algorithm~\ref{algo:focalrand_multiple}, we can replace the outcome $Y_i$ input into Algorithm~\ref{algo:focalrand_single_general} by $Y_i - \widehat f_{k-1}(X_i)$ for all $i \in \bigcup_{\ell \in [L_k]} \mS_{k,\ell}$, for some model $\widehat f_{k-1}$ fitted using data $(\Yobs_j, X_j)$ for $j \in \mathbb S_{<k}$ only. For example, $\widehat f_{k-1}$ can be a linear regression of $\Yobs_j$ on $X_j$ using $j \in \mathbb S_{<k}$, or simply subtracting the pre-treatment baseline outcome $Y_i^{\text{pre}}$ from $\Yobs_i$. 
The test remains valid because both $\widehat f_{k-1}$ and $X_i$ does not change with $\Zobs_i$ for $i \in \bigcup_{\ell \in [L_k]} \mS_{k,\ell} \setminus \mathbb S_{<k}$ (recall that the randomization test at step $k$ conditions on $\Zobs_{\mathbb S_{<k}}$), so that $H_{0k}$ implies the monotonicity of the adjusted potential outcomes: $y_i(0, \expov_k) - \widehat f_{k-1}(X_i) \geq y_i(0, \expov_{k+1}) - \widehat f_{k-1}(X_i)$ for all $i \in \bigcup_{\ell \in [L_k]} \mS_{k,\ell} \setminus \mathbb S_{<k}$.
\end{remark}

\begin{remark}[Relaxation of Algorithm~\ref{algo:focalrand_multiple}]
We can relax the requirement in Line 2 of Algorithm~\ref{algo:focalrand_multiple} to allow some eligible randomization units constructed at step $k'<k$ to also serve as eligible focal units at step $k$ for testing $H_{0k}$, provided they are observed to be in control. This is because we condition on the treatments of $\mathbb S_{<k}$ so the re-used units should be in control to realize the desired control potential outcomes, and the randomization occurs only at their neighbors, which could be outside $\mathbb S_{<k}$ to avoid a degenerate randomization distribution.
\end{remark}

\begin{remark}[Comparison with previous literature]
The result in Theorem~\ref{thm:focaltest_mono} advances the growing literature of randomization tests under interference~\citep{athey2018exact, basse2019randomization, puelz2022graph}. Prior methods cannot directly test monotone spillover hypotheses such as $H_0$, and are mainly designed to test simpler null hypotheses of the form defined in Equation~\eqref{eq:null_eq_k}. 
Moreover, while methods like the ``biclique method" of~\citet{puelz2022graph} can be adapted to test monotone spillovers, they can be computationally demanding, particularly for large networks. We will discuss this in more detail in Section~\ref{sec:method_clique}. Our method is also related to the one proposed in~\cite{zhang2021multiple}, which provides a general recipe for recursively constructing conditional randomization tests. Their method also builds on similar results to~\cite{rosenbaum2011some} as we do in Step 4 of Section~\ref{sec:overview}.
The key challenge in the recursive construction, however, is that it is context-specific.
It is unclear how their approach could be adapted to our monotone spillover hypothesis, where each individual hypothesis involves a contrast of exposure levels.
\end{remark}

\section{Methodology under Arbitrary Designs}\label{sec:method_clique}
The procedure outlined in Algorithm~\ref{algo:focalrand_multiple} relies heavily on the structure of the non-uniform Bernoulli design of Definition~\ref{def:nunifBRE}.
In this section, we discuss how to test $H_0$ under an arbitrary experimental design $P(\cdot)$. 

As before, we begin with a valid test for $\widetilde H_{0k}$, which 
under general designs is possible through the biclique testing procedure of~\cite{puelz2022graph}. Roughly speaking, this procedure translates a null hypothesis on treatment exposures into a bipartite graph ---known as the ``null exposure graph"--- and conditions the randomization test
on a biclique of this graph. 
The details are provided in Appendix~\ref{appdx:biclique_test}.
Additionally, using similar arguments to Theorem~\ref{thm:focaltest_mono}, the biclique test will be valid for $H_{0k}$ under an exposure-monotone test statistic.

With the biclique test for $H_{0k}$, we can test the full monotone null $H_0$ 
based on the idea of network splitting outlined in Section~\ref{sec:overview}.
Concretely, consider a partition of the network $\mathcal G = (V, E)$ into $\{\mathcal G_k\}_{k\in[K]}$, where $\mathcal G_k = (V_k, E_k)$, such that $V_k \cap V_{k'} = \emptyset$ for $k\neq k' \in [K]$.
Then, we can conduct a biclique test of the null $H_{0k}$ 
within sub-network $\mathcal G_k$, sequentially for $k=1,2,\ldots, K-1$, and combine the resulting $p$-values using Fisher's combination rule.
Importantly, in order to apply Lemma~\ref{lem:recur_pval} and guarantee the validity of the combination,
when testing $H_{0k}$ on $\mathcal G_k = (V_k, E_k)$, we want to ensure that the exposure function $\expof_i(\cdot)$, for $i\in V_k$, is computable using only treatments of units in $V_{\leq k} := \bigcup_{k' \leq k} V_{k'}$, in the sense that $\expof_i(z) = \expof_i(z_{V_{\leq k}})$. 
Thus, before conducting the biclique test on $\mathcal G_k$, we may simply remove units in $V_k$ whose exposures depend on treatments of units outside $V_{\leq k}$. 

The partition of the network is not unique, however. 
Heuristically, the partition could aim to maximize the connections between nodes within each sub-network, and minimize the connections between nodes in different sub-networks.
Several existing algorithms in the community detection literature work in this way, such as 
the Leiden algorithm in \cite{traag2019leiden}. 
We discuss this algorithm in Appendix~\ref{appdx:assignment} and provide a way to select the partition based on statistics of null exposure graphs. Importantly, the selection procedure can help improve our tests' power without affecting their validity.

Given a choice of the partition, Algorithm~\ref{algo:general_design} presents our proposed randomization test for the monotone null $H_0$ under a general design $P(\cdot)$.
The validity of the algorithm is given in Theorem~\ref{thm:generaltest} that follows. The proof is presented in Appendix~\ref{a:proof}.

\begin{algorithm}[t!]
\caption{Test for  $H_0: y_i(0, \expov_1) \geq y_i(0, \expov_2)\geq\cdots \geq y_i(0, \expov_K)$ under general design.}
\label{algo:general_design}
\begin{algorithmic}[1]
\REQUIRE Exposure set $\exposet = \{\expov_1, \expov_2,\ldots, \expov_K\}$; network $\mathcal G$ and a partition $\{\mathcal G_k\}_{k \in [K]}$ with $\mathcal G_k = (V_k, E_k)$;  design $P(\cdot)$; observed treatment $\Zobs$ (Input).

\hspace{-24px} {\bf Output:} Finite-sample valid $p$-value for $H_{0}$, $p_{\mathrm{FCT}}$.
\FOR{$k = 1,2,\ldots, K-1$}
    \STATE Define $V_k' \leftarrow \{i \in V_k: \expof_i(z) ~\text{is computable with}~ (z_i: i \in V_{\leq k}) \}$.
    \STATE Run the biclique test on $V_k'$, using as 
    input (i) the conditional distribution of $Z_{V_k}$ given $(\Zobs_i)_{i\in V_{<k}}$ where $V_{<k} := \bigcup_{k'<k} V_{k'}$, denoted as $P_k(\cdot)$;
    and (ii) an exposure-monotone test statistic $t_k$ in the order $(\expov_k, \expov_{k+1})$.
    \STATE From the biclique test above, obtain the biclique decomposition $\{\mathcal C_j^k\}_{j\in J_k}$ (defined in Appendix~\ref{appdx:biclique_test}) and the corresponding $p$-value $\pval_k$.
\ENDFOR
\STATE Output the sequence of $p$-values $(\pval_k)_{k=1}^{K-1}$ and the combined $p$-value $p_{\mathrm{FCT}}$ using~\eqref{eq:pval_FCT}.
\end{algorithmic}
\end{algorithm}

\begin{theorem}\label{thm:generaltest}
For any $k \in [K-1]$ and $\alpha_k \in [0,1]$ in Algorithm~\ref{algo:general_design},
\[
\mathbbm P_{Z^{\mathrm{obs}}_{V_k} \sim P_k(\cdot)} \left(\pval_k \leq \alpha_k ~\big|~ \{\mathcal C_j^k\}_{j\in J_k}, ~(\Zobs_i)_{i \in V_{<k}}, ~H_{0k}\right) \leq \alpha_k,
\]
where $P_k(\cdot)$ is defined in Line 3 of Algorithm~\ref{algo:general_design} as $P_k(Z_{V_k}) = P(Z_{V_k} | Z^{\mathrm{obs}}_{V_{<k}})$.
Moreover, the $p$-values $(\pval_k)_k$ are stochastically larger than uniform, so the combined $p$-value from Line 6 of Algorithm~\ref{algo:general_design} is valid under the monotone spillover hypothesis $H_0$. That is, $\mathbbm P(p_{\mathrm{FCT}} \leq \alpha ~|~ H_{0} ) \leq \alpha$ for all $\alpha \in [0,1]$,
where the probability is taken with respect to design the $P(\cdot)$.
\end{theorem}

Algorithm~\ref{algo:general_design} can handle general designs but could be computationally intensive. 
The main computational challenge is to build the null exposure graph and execute the biclique test based on the conditional distribution specified recursively in Line 3.
When the treatment follows a non-uniform Bernoulli design as in Section~\ref{sec:nonunifBRE}, this conditional distribution is still a non-uniform Bernoulli design on the set $V_k$ and is computationally easy to sample from to build the null exposure graph.
Similarly, this conditional distribution has a simple structure whenever the experiment follows a completely randomized design, or a clustered or blocked completely randomized design. 
Under general designs, however, the form of this conditional distribution may be complex.
On the other hand, under the non-uniform Bernoulli design, the module-based test in Section~\ref{sec:nonunifBRE} is significantly easier computationally compared to the biclique-based test and can also achieve higher power as shown in the simulations in Section~\ref{sec:simu_results} and Appendix~\ref{appdx:assign_power} below, by leveraging the structure of the design.

\section{Simulated Studies: Validity and Power}\label{sec:simu}

In this section, we examine our methods in simulations. 
We will use as a basis a large-scale ``hotspots" policing experiment in Medellín, Colombia conducted by~\cite{collazos2021hot} to test a monotone spillover hypothesis related to the crime displacement hypothesis discussed in Example~\ref{eg:crime_spill}. 
First, we describe some background information on the experiment, 
and then demonstrate our randomization procedures using both simulated and real outcomes. 

\subsection{Background}
The units of the experiment were $N=37,055$ street segments in Medellín. Based on past crime data and input from the police, $967$ of them were selected as ``hotspot" streets. Among these hotspots, $384$ were randomly assigned to treatment consisting of a 50-80\% increase in police patrol time over a period of six months. Treatment assignment was subject to certain complex constraints imposed by the police.\footnote{For instance, the police limited the number of treated hotspots across policing stations. Moreover, due to an inadvertent coding error, 7 hotspots were always assigned to treatment~\citep{collazos2021hot}.} 
To simplify, we approximate the design through a non-uniform Bernoulli design with unit treatment probabilities calculated from the true assignment mechanism. We think this approximation is adequate as the unit-level treatments according to the true assignment mechanism are uncorrelated.

The outcomes of interest are post-treatment crime counts on five types of crime: homicides, assaults, car theft, motorbike theft, and personal robbery, as well as a crime index weighted by the relative average prison sentence.\footnote{$0.550$ for homicides, $0.112$ for assaults, $0.221$ for car and motorbike theft, and $0.116$ for robbery.}
Crime spillovers of any type are possible due to the adjacency between nearby streets. 
Let $d(i,j)$ denote the distance between two streets $i$ and $j$, defined as the geographic distance 
between their midpoints.
We consider the following exposure mapping function akin to the definition in Equation~\eqref{eq:neightreated_cont}:
\[
\expof_i(z) = \sum_{j \in \nei_i} z_j, \quad \nei_i = \{j \in [N]: d(i,j) \leq r\}.
\]
This definition counts how many streets that are neighbors of street $i$ are treated, 
where a neighbor of $i$ is any street within $r$ meters of $i$.
We choose $r = 225$, following the analysis in~\cite{collazos2021hot, puelz2022graph}.
Figure~\ref{fig:Med-histdeg} (left panel) shows the geography of Medellín and the arrangement of 
treated and control hotspot streets as they were observed in the actual experiment.
The right panel shows a histogram of the number of neighboring hotspot streets (which we refer to as ``degree") across all streets, revealing a pattern typically found in scale-free networks~\citep{kolaczyk2009statistical}.

\begin{figure}[!t]
    \centering
    \includegraphics[width = 0.38\textwidth]{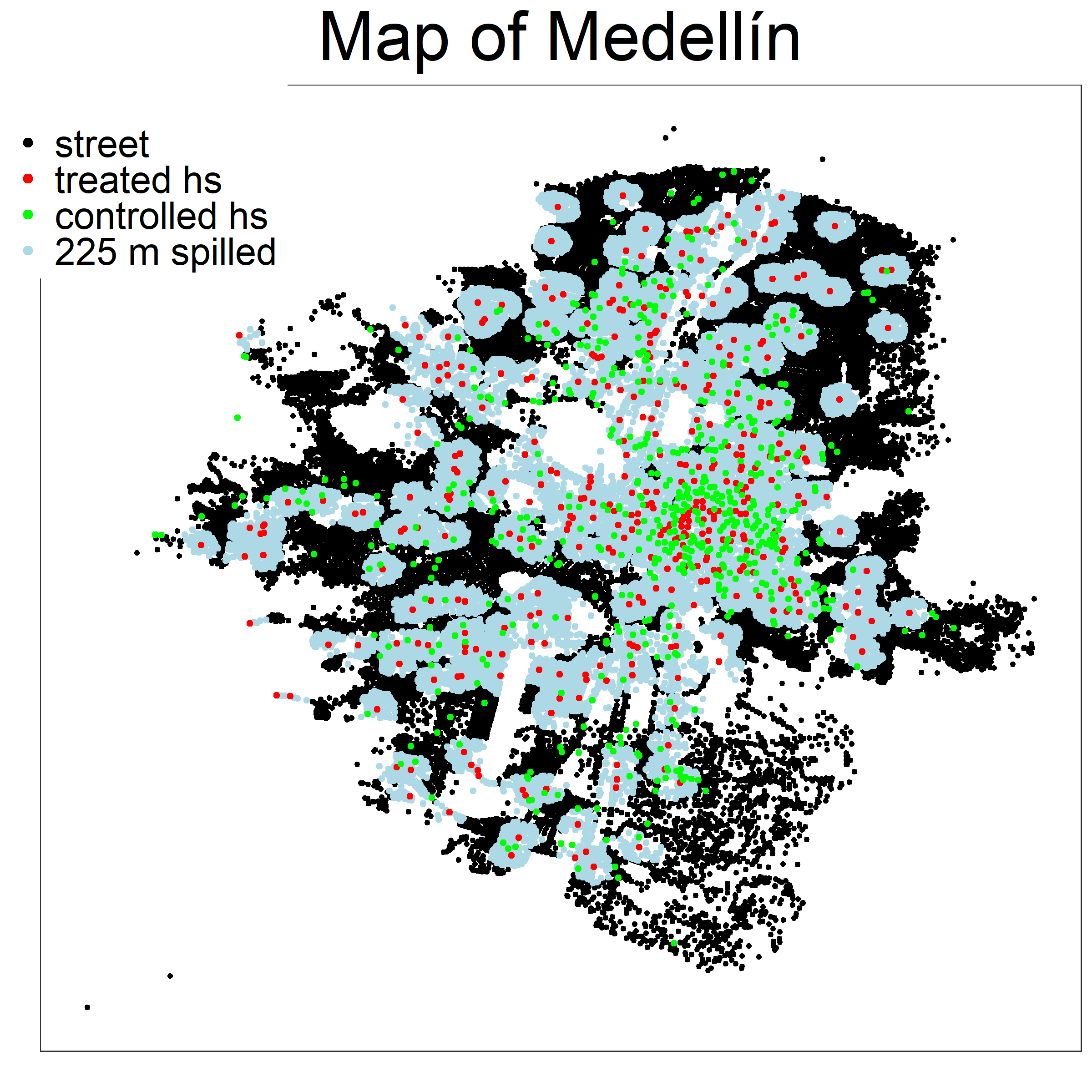}~~
    \includegraphics[width = 0.38\textwidth]{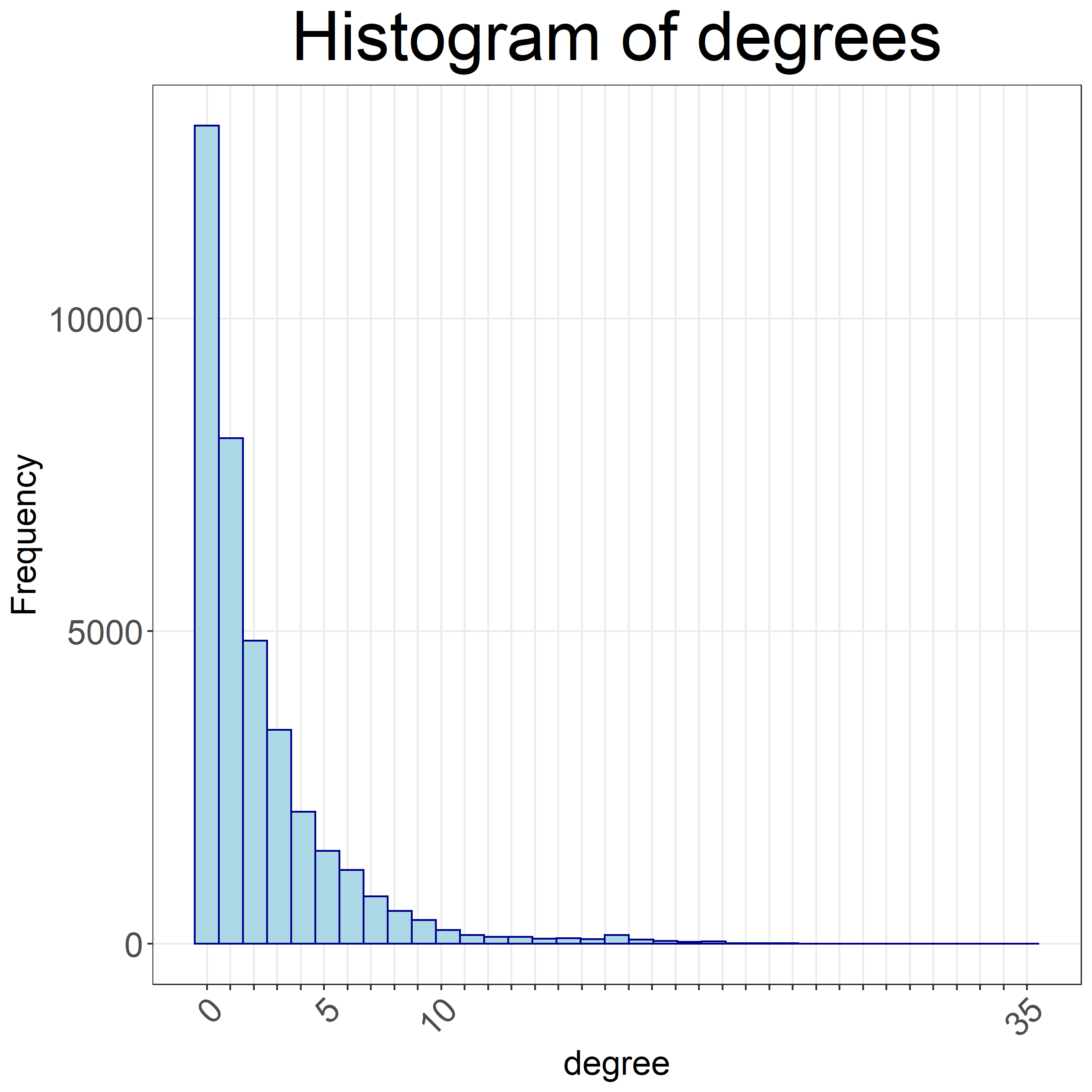}
    \caption{Map of Medellín and histogram of degrees for all units.}
    \label{fig:Med-histdeg}
\end{figure}

Based on these observations, we define $\exposet = \{0,1,2,[\geq 3]\}$ as the set of possible exposures, where ``$[\geq 3]$" denotes any exposure with at least 3 treated neighbors. That is, all units with $\expof_i \geq 3$ are grouped into one exposure category labeled as ``$[\geq 3]$". Under Assumption~\ref{a:exposure}, this implies that the controlled potential outcomes of a street are the 
same for all exposures that include at least 3 treated neighbors. 
This assumption is limiting but it simplifies the analysis and can be justified through an empirical check using randomization tests; see Appendix~\ref{appdx:grouping3} for details.
Moreover, to check the robustness of our results, in Appendix~\ref{sec:Med_5H} we consider an alternative definition of the exposure set, $\exposet = \{0, 1, 2, 3, 4, [\geq 5]\}$, where ``$[\geq 5]$" is defined similarly and allows for more exposure levels.

Given the exposure set, we consider the following monotone null hypothesis:
\begin{equation}\label{eq:mononull_med}
\begin{gathered}
    H_0 = \bigcap_{k \in [3]} H_{0k}, ~\text{where}\\
    H_{01}: y_i(0, 0) \geq y_i(0, 1),\quad H_{02}: y_i(0, 1) \geq y_i(0, 2), \quad H_{03}: y_i(0, 2)\geq y_i(0, [\geq 3]),\quad \forall i.
\end{gathered}
\end{equation}
The potential outcome $y_i(0,\expov)$ represents a certain crime occurrence when $i$ is in control and receives exposure $\expov$.
That is, we are testing the hypothesis that a street is benefited from having treated neighbor streets, 
and that this benefit is weakly monotonic. 
In the real data analysis in Section~\ref{sec:realMed}, we will also consider the above hypothesis
in the opposite direction:
\begin{equation}\label{eq:H0_crime_displacement}
\begin{gathered}
    H_0' = \bigcap_{k \in [3]} H_{0k}', ~\text{where}\\
    H_{01}': y_i(0, 0) \leq y_i(0, 1),\quad H_{02}': y_i(0, 1) \leq y_i(0, 2), \quad H_{03}': y_i(0, 2)\leq y_i(0, [\geq 3]),\quad \forall i.
\end{gathered}
\end{equation}
This hypothesis formulates the crime displacement hypothesis: a street is negatively impacted from having treated neighbor streets, and the  impact is weakly monotonic. 
We note that a rejection of an individual hypothesis $H_{0k}$ generally implies a non-rejection of $H_{0k}'$. However, this need not be true for the combined hypotheses $H_0$ and $H_0'$, since the combination typically involves a non-linear transformation of individual $p$-values.

\subsection{Simulation DGPs}\label{sec:simu_dgp}
We consider four different data generating processes (DGP) for the monotone hypothesis~\eqref{eq:mononull_med}, each showcasing a particular aspect of our proposed methodology. 
For each DGP, the control potential outcomes are generated as $y_i(0,0) \overset{iid}{\sim} \mathrm{Gamma}(\widehat \alpha, \widehat \beta)$, where $(\widehat \alpha, \widehat \beta)$ are calibrated on the real Medellín data by matching the mean and variance of the observed post-treatment crime index among units that are in control ($\Zobs_i=0$) and receive no exposure, i.e., $\expof_i(\Zobs) = 0$. The remaining potential outcomes are defined as follows.
{\small \begin{align}
y_i(z_i, \expof_i) & = y_i(0,0) \exp\{-z_i + \tau \expof_i (1 - 0.5 z_i) \}, & \tag{DGP1} \\
y_i(z_i, \expof_i) & = y_i(0,0) \exp\{-z_i + \tau  [\expof_i - 2\mathbbm{1}(\expof_i=1)] (1 - 0.5 z_i) \},& \tag{DGP2}\\
y_i(z_i, \expof_i) & = y_i(0,0) \exp\{-z_i + \tau \expof_i \cdot (1 - 0.5 z_i) + \theta d_i\}, & \tag{DGP3} \\
y_i(z_i, \expof_i) & = y_i^*(0,0) \exp\{-z_i + \tau \expof_i (1 - 0.5 z_i) \}, & \tag{DGP4}\\
\text{with}&~y_i^*(0,0) = y_i(0, 0)+ \frac{|G^{r\prime}_{i,\cdot}\varepsilon|}{\sqrt{\sum_j G^r_{i, j}}}, ~\varepsilon_i \overset{iid}{\sim} N(0,1). & \nonumber
\end{align}}DGP1 is used to illustrate the validity of our approach for $\tau \leq 0$, and its power against $\tau > 0$.
DGP2 is a variation of DGP1 with heterogeneity and non-monotonicity in the spillover effect. 
Specifically, when $\tau>0$, the spillover effect turns from positive to negative when $\expof_i=1$ for control streets, and reverses for treated streets. DGP3 introduces network-related confounding where the degree $d_i$ of street $i$ affects outcomes 
at a level controlled by parameter $\theta$.
DGP4 introduces confounding through network-correlated errors, where $G^r \in \{0, 1\}^{N\times N}$ with $G^r_{i,j} = \mathbbm 1\{d(i,j) \leq r\}$ as its $(i,j)$-th element and $G^r_{i,\cdot}$ as its $i$-th row.

For test statistics in the randomization tests, we consider the simple difference-in-means (``DiM") defined in~\eqref{eq:diff-in-mean} and the Stephenson rank sum statistics defined in~\eqref{eq:rank-stat} with $\varphi(r) = \binom{r-1}{s-1}$ if $r\geq s$ and $\varphi(r) = 0$ otherwise, for $s=5,20$. We refer to these as ``rs5" and ``rs20", respectively. Apart from Fisher's combination of $p$-values, we also consider Stouffer's combination $p_{\mathrm{Stouf}} =\Phi(\sum_k n_k\Phi^{-1}(p_k) / \|n\|_2)$, where $\Phi(\cdot)$ denotes the standard Normal CDF and $n = (n_1, n_2, n_3)$ with $n_k$ being the expected number of units with exposures $\expov_k$ or $\expov_{k+1}$ for $k = 1,2,3$. Other weighting schemes are also possible.
Additionally, as a baseline procedure, we consider a one-sided {\tt t}-test on the coefficient on $\expof_i$ in the linear regression model of $\log Y_i$ on treatment $Z_i$, exposure $\expof_i$ and their interaction $Z_i \cdot \expof_i$. 

\subsection{Simulation results}\label{sec:simu_results}
The simulation results for all DGPs are presented in Table~\ref{tab:simu_resulttables}. The results are calculated over 2,000 replications, corresponding to a $\pm0.5\%$ sampling error. 
In DGP1, we see that all methods are valid, including the baseline OLS procedure. The rank sum statistics 
is generally low-powered, partly due to homogeneity in the effects, but the difference-in-means test statistic leads to a randomization test with a power comparable to OLS (roughly 70\%). The Fisher and Stouffer combination rules differ only slightly.

In DGP2, our randomization tests are all valid in the absence of spillover effect ($\tau=0$). 
In contrast to DGP1, negative $\tau$'s lead to a violation of the monotonicity hypothesis as well. 
We see that the randomization tests with the difference-in-means statistic  are powerful against these alternatives, 
whereas OLS is not. The randomization test with the rank-sum statistic is again low-powered but is clearly better compared to OLS against those alternatives. 
Interestingly, there is a noticeable difference between the Fisher and Stouffer combination rules in terms of power under DGP2, due to the behavior of the transformation $\Phi^{-1}(p)$ when $p\to0$ and $1$. We discuss more about this issue in Appendix~\ref{appdx:pvalcomb}.

\renewcommand{\arraystretch}{0.95}
\begin{table}[!t]
\centering
{
\begin{tabular}{cccccccccc}
  \hline
\multirow{10}{*}{DGP1}  & Combination & Test statistic & 
\multicolumn{7}{c}{Spillover effect  $\tau$} \\
\cmidrule{2-10}
&  &  & -0.5 & -0.2 & -0.1 & 0 & 0.1 & 0.2 & 0.5 \\ 
\cmidrule{2-10}
  & \multirow{3}{*}{Fisher} & DiM & 0.00 & 0.00 & 0.40 & 5.15 & 24.65 & 62.25 & 100.00 \\ 
  & & rs5 & 1.20 & 2.85 & 3.45 & 5.70 & 6.55 & 7.40 & 13.60 \\ 
  & & rs20 & 0.20 & 1.10 & 2.25 & 4.90 & 9.55 & 17.20 & 50.55 \\ 
  \cmidrule{2-10}
  & \multirow{3}{*}{Stouffer} & DiM & 0.00 & 0.00 & 0.35 & 5.15 & 28.25 & 68.70 & 100.00 \\ 
  & & rs5 & 1.30 & 2.95 & 3.35 & 5.75 & 6.50 & 8.35 & 12.85 \\ 
  & & rs20 & 0.00 & 0.95 & 2.15 & 5.85 & 10.20 & 16.95 & 48.45 \\ 
  \cmidrule{2-10}
  & & OLS & 0.00 & 0.00 & 0.35 & 4.90 & 40.65 & 88.30 & 100.00 \\ 
\hline\\
\hline
\multirow{10}{*}{DGP2}  & Combination &  & \multicolumn{7}{c}{Spillover effect  $\tau$} \\
\cmidrule{2-10}
&  &  & -0.5 & -0.2 & -0.1 & 0 & 0.1 & 0.2 & 0.5 \\ 
\cmidrule{2-10}
 & \multirow{3}{*}{Fisher} & DiM & 94.45 & 18.60 & 6.10 & 4.15 & 16.60 & 52.65 & 99.85 \\ 
  & & rs5 & 2.85 & 2.90 & 4.40 & 4.95 & 6.40 & 7.05 & 13.80 \\ 
  & & rs20 & 7.50 & 2.45 & 3.85 & 5.15 & 9.70 & 16.80 & 59.15 \\ 
  \cmidrule{2-10}
  & \multirow{3}{*}{Stouffer} & DiM & 1.30 & 10.85 & 7.55 & 4.30 & 3.05 & 0.60 & 0.00 \\ 
  & & rs5 & 3.80 & 4.10 & 5.05 & 5.20 & 4.75 & 5.15 & 4.85 \\ 
  & & rs20 & 2.95 & 4.15 & 4.80 & 5.60 & 4.90 & 5.10 & 3.20 \\ 
  \cmidrule{2-10}
  & & OLS & 0.00 & 0.00 & 0.35 & 4.90 & 36.35 & 79.50 & 100.00 \\ 
\hline\\
\hline
\multirow{7}{*}{DGP3}  & Combination &  & \multicolumn{7}{c}{Degree confounding $\theta$ ($\tau=0$ throughout)} \\
\cmidrule{2-10}
&  &  & -0.3 & -0.2 & -0.1 & 0 & 0.1 & 0.2 & 0.3 \\ 
\cmidrule{2-10}
 & \multirow{3}{*}{Fisher} & DiM & 5.20 & 5.15 & 4.85 & 5.40 & 4.95 & 4.95 & 5.25 \\ 
  & & rs5 & 4.55 & 5.00 & 5.05 & 4.10 & 5.35 & 5.25 & 4.00 \\ 
  & & rs20 & 4.50 & 4.85 & 4.95 & 4.25 & 4.55 & 5.65 & 4.00 \\ 
  \cmidrule{2-10}
  & & OLS & 0.00 & 0.00 & 0.00 & 5.10 & 94.05 & 100.00 & 100.00 \\ 
\hline\\
\hline
\multirow{7}{*}{DGP4}  & Combination & tstat & \multicolumn{7}{c}{Network correlation $r$ ($\tau=0$ throughout)} \\
\cmidrule{2-10}
&  &  & 0 & 50 & 100 & 150 & 225 & 350 & 400 \\ 
\cmidrule{2-10}
 & \multirow{3}{*}{Fisher} & DiM & 4.67 & 4.60 & 4.80 & 4.50 & 5.65 & 5.03 & 4.50 \\
  & & rs5 & 4.92 & 5.20 & 5.10 & 4.35 & 5.25 & 3.92 & 3.60 \\ 
  & & rs20 & 4.67 & 4.90 & 5.30 & 4.75 & 5.00 & 4.47 & 4.35 \\ 
  \cmidrule{2-10}
  & & OLS & 5.13 & 13.90 & 23.95 & 29.80 & 34.10 & 37.14 & 38.50 \\ 
   \hline
\end{tabular}
}
\caption{Simulation results for DGP1-DGP4 under 2,000 replications. All tests are conducted at the 5\% level and 
the reported values are rejection probabilities (in \%).}
\label{tab:simu_resulttables}
\end{table}

In DGP3 and DGP4, we fix $\tau=0$ (no spillover effect) and vary the degree of network confounding through parameters $\theta$ and $r$, respectively. We therefore expect all valid procedures to reject at 5\% across all these parameters.
Indeed, the randomization-based tests are all valid across all settings.
However, the OLS-based tests are significantly distorted. For instance, in DGP3, the 
OLS test over-rejects at a level 94\%  when $\theta=0.1$. An increasing $\theta$ parameter leads 
to a spurious correlation between outcomes and network degrees, which confounds OLS.
In DGP4, OLS over-rejects from 13.9\% to 38.50\% as we increase $r$ from 50 to 400. 
An increasing parameter $r$ leads to an increasingly long-range correlation between error terms in the potential outcomes, leading in turn to a spurious correlation between observed outcomes and treatment exposures.
See Appendix~\ref{appdx:degree_medellin} for empirical evidence of the confounders.

\begin{remark}[Saturated regression]
One can also run a regression saturated at both treatment and exposure, i.e., $Y_i = \sum_{z\in\{0,1\}}\sum_{k\in[K]} \mathbbm 1\{ Z_i = z, \expof_i = \expov_k\} \beta_{z,k} + \varepsilon_i$, and test the monotone hypothesis~\eqref{eq:mononull_med} by testing $\beta_{0,k} \geq \beta_{0,k'}~\forall k \leq k'$ using the limiting distribution of $\widehat \beta_{z,k}$. The validity of this approach again depends critically on the correct specification of the model, without which the $\beta$ does not represent the true causal effect (DGP3).
Even under the correct specification, establishing the limiting distribution can also be challenging when, for example, the correlation between outcomes is too strong, as in DGP4. 
\end{remark}

\section{Application: Testing Monotonicity in Crime Spillovers}\label{sec:realMed}
We now turn to testing the monotone nulls defined in Equations~\eqref{eq:mononull_med} and~\eqref{eq:H0_crime_displacement} using the real Medellín data, where we specify $\exposet = \{0, 1, 2, [\geq 3]\}$. The robustness results for the alternative specification $\exposet = \{0, 1, 2, 3, 4, [\geq 5]\}$ are presented in Appendix~\ref{sec:Med_5H}.

The left panel of Figure~\ref{fig:hist-crimes} displays the histograms of the post-treatment counts of the five types of crime and the crime index. Generally, all these outcomes are inflated at zero, and some of them take a few large values. Also, many of the pre-treatment baseline outcomes are identical to the post-treatment ones. As a result, we adjust each outcome by its pre-treatment value and use the difference between post- and pre-treatment crime or crime index, $\Delta Y_i = Y_i^{\text{post}} - Y_i^{\text{pre}}$, as the (adjusted) outcome in test statistics. The histogram of $\Delta Y_i$ is displayed in the right panel in Figure~\ref{fig:hist-crimes}. For test statistics, we again use the difference-in-means and the rank sum statistics in Section~\ref{sec:simu_dgp}. We pay particular attention to the rank sum statistics due to their ability to detect uncommon-but-dramatic responses to treatment and often superior power under treatment effect heterogeneity~\citep{caughey2023bounded}, which could be useful for detecting the violation of individual-level monotonicity given the distribution of $\Delta Y_i$ in our setting. For brevity, in the following we will focus on crime index as the outcome. Results for other outcomes are presented in Appendix~\ref{appdx:more_emp_3H}.

To account for the uncertainty in constructing module sets, we repeat the construction $2,000$ times independently, run Algorithm~\ref{algo:focalrand_multiple} for each construction, and report twice the median of all the $2,000$ $p$-values for~\eqref{eq:mononull_med} in the third column in Table~\ref{tab:Medellin_crimeindex}, which is a valid but potentially conservative $p$-value\footnote{See, for example, Lemma 4.1 in \cite{chernozhukov2018generic}. In Appendix~\ref{appdx:more_emp_3H} we show the histograms of these $p$-values.}. 
Figure~\ref{fig:hist-fu} shows descriptive histograms of the numbers of eligible and active focal units across the $2,000$ replications.

\begin{figure}[!t]
    \centering
    \begin{tabular}{cc}
        \includegraphics[width = 0.5\textwidth]{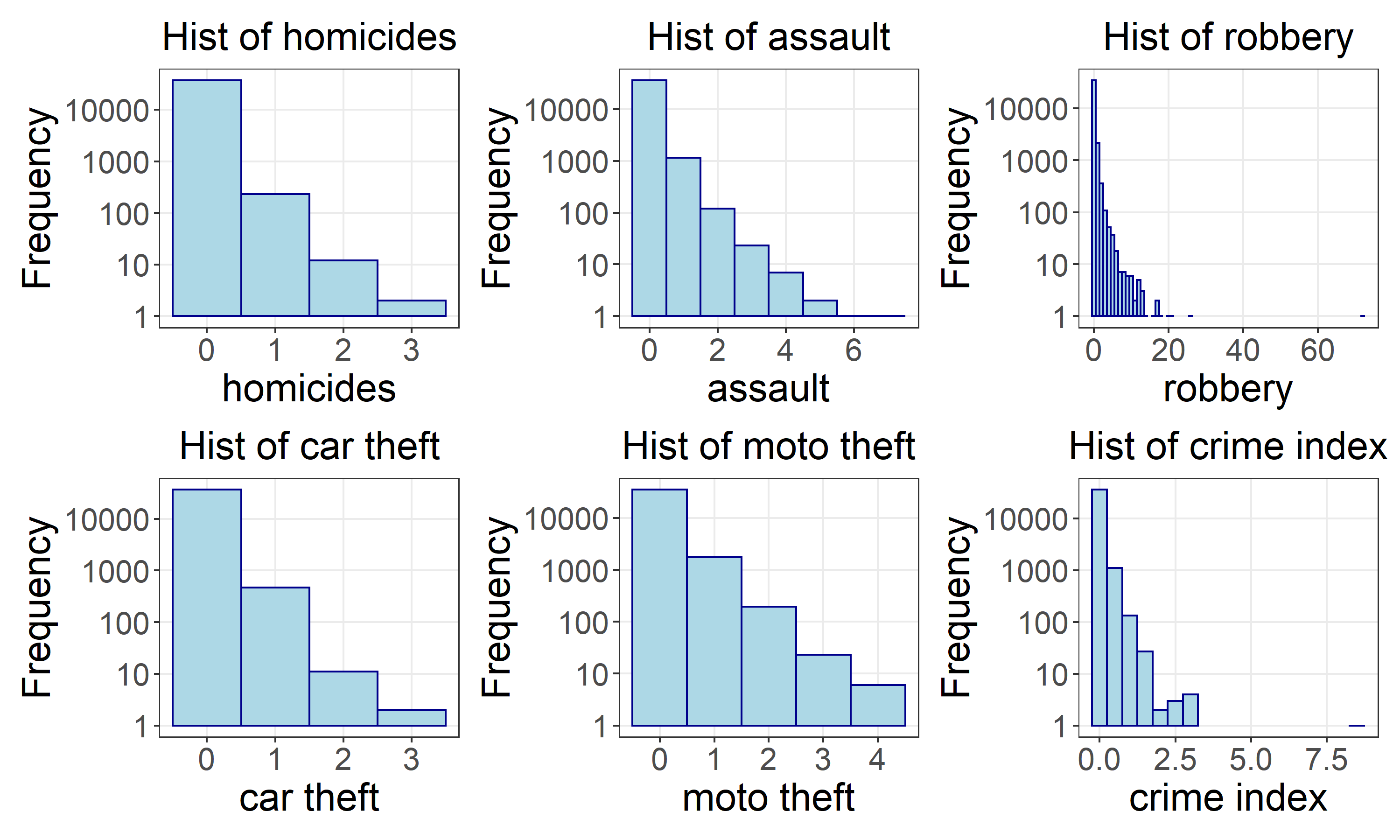} &
        \includegraphics[width = 0.5\textwidth]{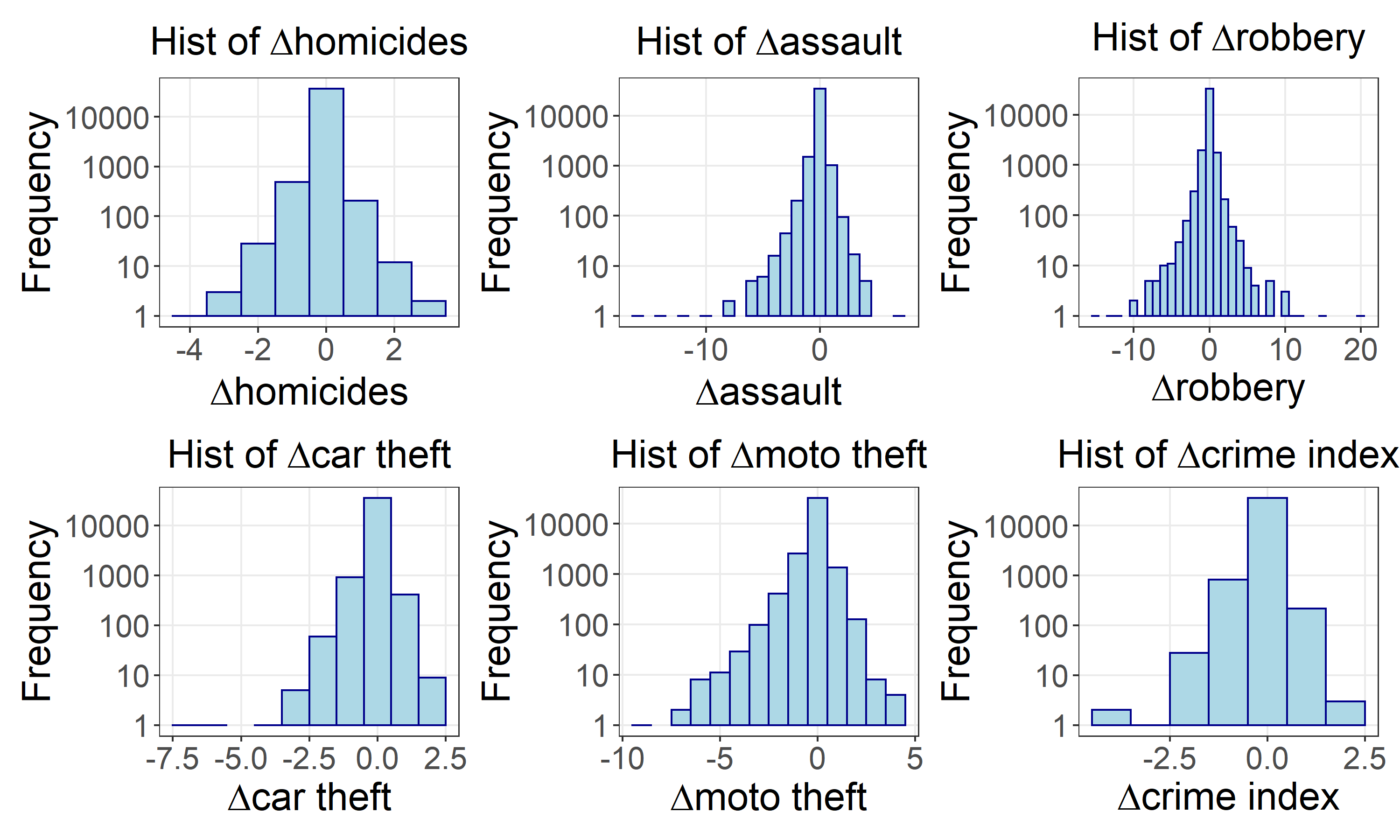} \\
    \end{tabular}
    \caption{Histogram of crimes (left) and differences in post- and pre-treatment crimes (right).}
    \label{fig:hist-crimes}
\end{figure}

\begin{figure}[!t]
    \centering
    \includegraphics[width = 0.9\textwidth]{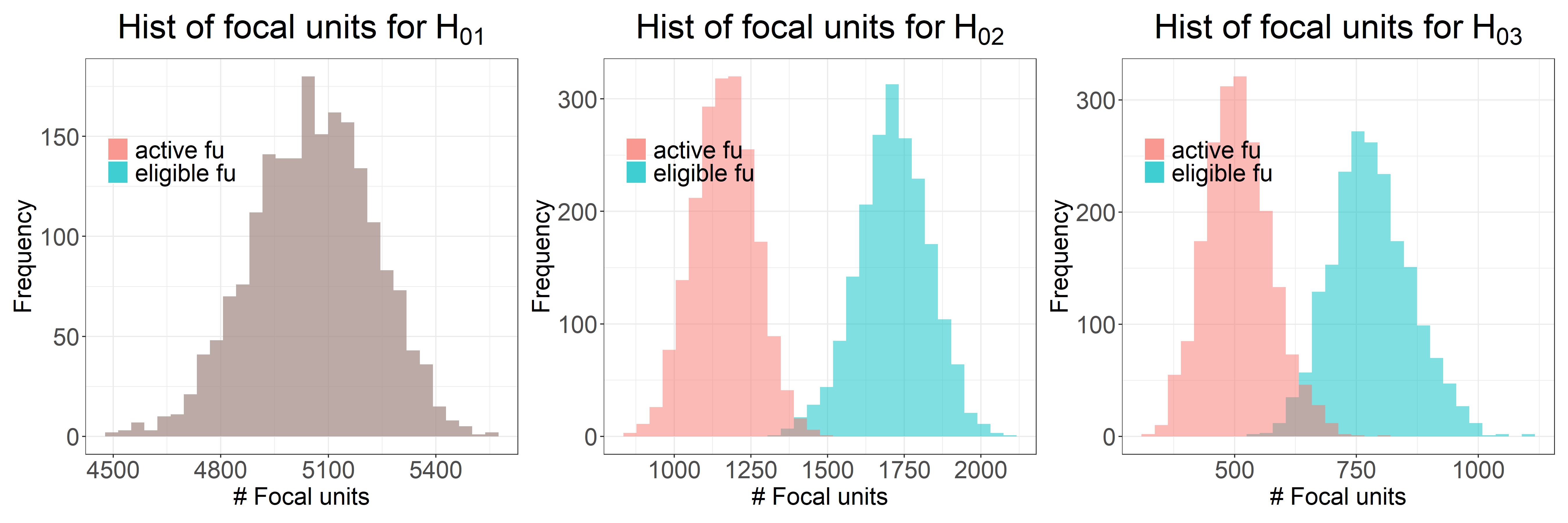}
    \caption{Histogram of number of focal units across the $2,000$ constructions.}
    \label{fig:hist-fu}
\end{figure}

To get a sense of how the focal units are distributed on the city map, we display two realizations of the algorithm in Figure~\ref{fig:3H_realization}. In general, there are far more eligible/active focal units for hypotheses involving lower levels of exposure than for those involving higher levels. 
Also, focal units for lower exposure level hypotheses tend to be located at the outskirts of the city, while focal units for higher exposure level hypotheses tend to be located at the center of the city. 
A plausible explanation for this observation is that there is higher density of hotspot streets at the city center, which is considered more disturbing, compared to the outskirts.
We also apply the idea discussed in Section~\ref{para:temp} to conduct tests on the module set that maximizes the expected number of active focal units among the $2,000$ module sets, both for the beneficial hypothesis~\eqref{eq:mononull_med} and the crime displacement hypothesis~\eqref{eq:H0_crime_displacement}. These results are presented in the fourth and fifth column of Table~\ref{tab:Medellin_crimeindex}. 

\begin{figure}[!t]
    \centering
    \includegraphics[width = 0.40\textwidth]{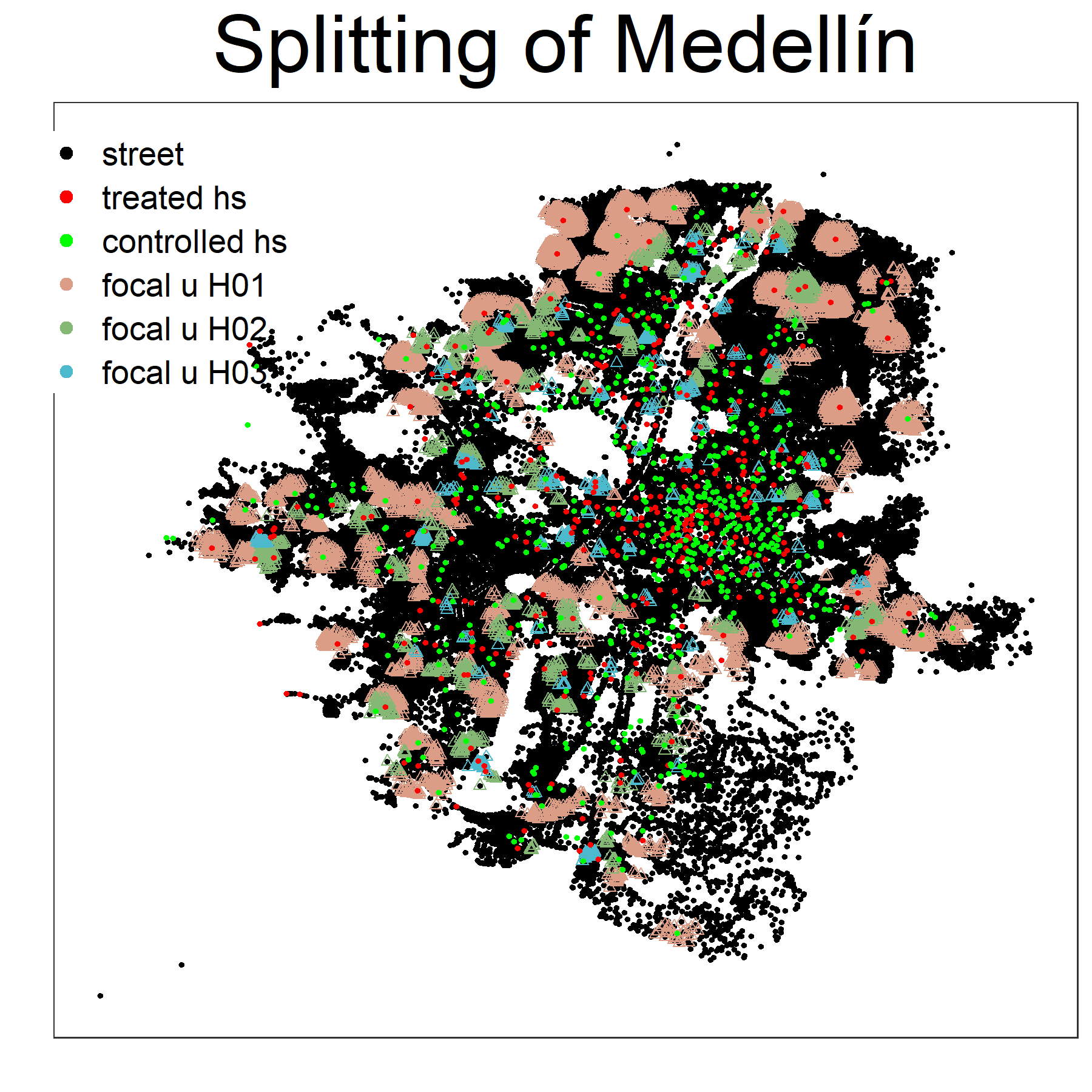}
    \includegraphics[width = 0.40\textwidth]{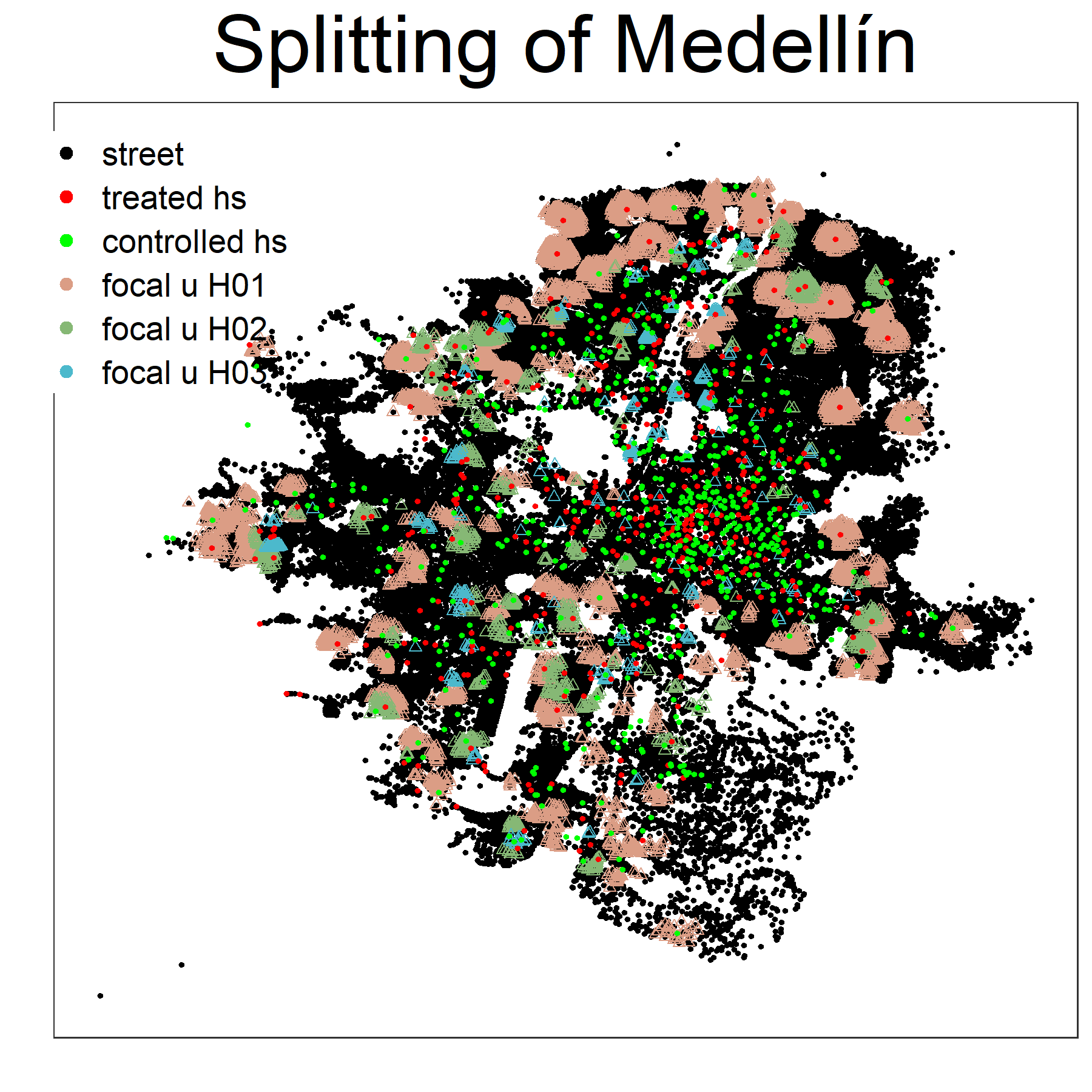}
    \caption{Two realizations of focal units for testing under exposures $\exposet = \{0, 1, 2, [\geq 3]\}$.}
    \label{fig:3H_realization}
\end{figure}

Overall, our randomization procedures consistently reject the beneficial spillover hypothesis~\eqref{eq:mononull_med} that the crime outcome is monotonically decreasing in the number of neighboring treated streets. Moreover, $p$-values tend to increase from $H_{01}$ to $H_{03}$, and the strongest rejection signal comes from $H_{01}$, so that the failure of monotonicity mainly comes from the individual null hypothesis $H_{01}: y_i(0, 0) \geq y_i(0,1)$. 
This suggests that there is negative crime spillover from nearby policing, supporting the crime displacement hypothesis which we fail to reject. Moreover, our results also provide some evidence of ``diminishing returns", as more intense nearby policing may not necessarily lead to increased crime displacement.

\renewcommand{\arraystretch}{0.95}
\begin{table}[!t]
\centering
\begin{tabular}{lcccc}
  \hline
Hypothesis & test stat & Twice median & Max AFUs & Crime disp.\ \\
  \hline
$H_{01}$ & \multirow{4}{*}{DiM} & 0.148 & \textbf{0.011} & 0.989 \\ 
  $H_{02}$ &  & 0.856 & 0.538 & 0.462 \\ 
  $H_{03}$ &  & 0.406 & 0.200 & 0.800 \\ 
  $H_0$ by FCT &  & 0.179 & \textbf{0.036} & 0.919 \\ 
  \hline
  $H_{01}$ & \multirow{4}{*}{rs5} & 0.074 & \textbf{0.008} & 0.992 \\ 
  $H_{02}$ &  & 0.128 & 0.055 & 0.945 \\ 
  $H_{03}$ &  & 0.404 & 0.182 & 0.818 \\ 
  $H_0$ by FCT &  & \textbf{0.029} & \textbf{0.004} & 0.997 \\ 
  \hline
  $H_{01}$ & \multirow{4}{*}{rs20} & 0.159 & \textbf{0.011} & 0.989 \\ 
  $H_{02}$ &  & 0.220 & 0.164 & 0.836 \\ 
  $H_{03}$ &  & 0.513 & 0.280 & 0.720 \\ 
  $H_0$ by FCT &  & 0.095 & \textbf{0.019} & 0.984 \\ 
   \hline
\end{tabular}
\caption{Third column: Twice the median of $p$-values for testing~\eqref{eq:mononull_med} across $2,000$ module set constructions. Fourth column: $p$-values for testing~\eqref{eq:mononull_med} using the module sets that maximize the expected number of active focal units. Fifth column: $p$-values for testing the crime displacement hypothesis~\eqref{eq:H0_crime_displacement} using the module sets that maximize the expected number of active focal units. \textbf{Bold} denotes value below $0.05$.} 
\label{tab:Medellin_crimeindex}
\end{table}

\section{Concluding Remarks}\label{sec:conclusion}
In this paper, we extend randomization tests for spillover effects in network settings to test the monotonicity of spillover effects, leveraging the idea of network splitting and the ``bounded null" perspective of randomization tests. 
An interesting future direction would be to study the power of our tests.
Power analysis is particularly challenging in the current setup because it
depends on the network topology, the alternative hypothesis, and the module set construction. Another valuable direction would be to formalize the test for diminishing returns of spillover effects. 
Although our current randomization tests can provide evidence for or against diminishing returns, they do not directly test the second-order derivative of the potential outcome function, which would be crucial for a formal test on diminishing returns.

\singlespacing
\bibliographystyle{chicago}

\begin{thebibliography}{}

\bibitem[\protect\citeauthoryear{Aronow and Samii}{Aronow and Samii}{2017}]{aronow2017estimating}
Aronow, P.~M. and C.~Samii (2017).
\newblock Estimating average causal effects under general interference, with application to a social network experiment.
\newblock {\em Annals of Applied Statistics\/}~{\em 11\/}(4), 1912--1947.

\bibitem[\protect\citeauthoryear{Athey, Eckles, and Imbens}{Athey et~al.}{2018}]{athey2018exact}
Athey, S., D.~Eckles, and G.~W. Imbens (2018).
\newblock Exact p-values for network interference.
\newblock {\em Journal of the American Statistical Association\/}~{\em 113\/}(521), 230--240.

\bibitem[\protect\citeauthoryear{Basse, Ding, Feller, and Toulis}{Basse et~al.}{2024}]{basse2024randomization}
Basse, G., P.~Ding, A.~Feller, and P.~Toulis (2024).
\newblock Randomization tests for peer effects in group formation experiments.
\newblock {\em Econometrica\/}~{\em 92\/}(2), 567--590.

\bibitem[\protect\citeauthoryear{Basse, Feller, and Toulis}{Basse et~al.}{2019}]{basse2019randomization}
Basse, G.~W., A.~Feller, and P.~Toulis (2019).
\newblock Randomization tests of causal effects under interference.
\newblock {\em Biometrika\/}~{\em 106\/}(2), 487--494.

\bibitem[\protect\citeauthoryear{Birnbaum}{Birnbaum}{1954}]{birnbaum1954combining}
Birnbaum, A. (1954).
\newblock Combining independent tests of significance.
\newblock {\em Journal of the American Statistical Association\/}~{\em 49\/}(267), 559--574.

\bibitem[\protect\citeauthoryear{Blattman, Green, Ortega, and Tob{\'o}n}{Blattman et~al.}{2021}]{blattman2021place}
Blattman, C., D.~P. Green, D.~Ortega, and S.~Tob{\'o}n (2021).
\newblock Place-based interventions at scale: The direct and spillover effects of policing and city services on crime.
\newblock {\em Journal of the European Economic Association\/}~{\em 19\/}(4), 2022--2051.

\bibitem[\protect\citeauthoryear{Brannath, Posch, and Bauer}{Brannath et~al.}{2002}]{brannath2002recursive}
Brannath, W., M.~Posch, and P.~Bauer (2002).
\newblock Recursive combination tests.
\newblock {\em Journal of the American Statistical Association\/}~{\em 97\/}(457), 236--244.

\bibitem[\protect\citeauthoryear{Breza, Stanford, Alsan, Alsan, Banerjee, Chandrasekhar, Eichmeyer, Glushko, Goldsmith-Pinkham, Holland, et~al.}{Breza et~al.}{2021}]{breza2021effects}
Breza, E., F.~C. Stanford, M.~Alsan, B.~Alsan, A.~Banerjee, A.~G. Chandrasekhar, S.~Eichmeyer, T.~Glushko, P.~Goldsmith-Pinkham, K.~Holland, et~al. (2021).
\newblock Effects of a large-scale social media advertising campaign on holiday travel and covid-19 infections: a cluster randomized controlled trial.
\newblock {\em Nature medicine\/}~{\em 27\/}(9), 1622--1628.

\bibitem[\protect\citeauthoryear{Caughey, Dafoe, Li, and Miratrix}{Caughey et~al.}{2023}]{caughey2023bounded}
Caughey, D., A.~Dafoe, X.~Li, and L.~Miratrix (2023).
\newblock {Randomisation inference beyond the sharp null: bounded null hypotheses and quantiles of individual treatment effects}.
\newblock {\em Journal of the Royal Statistical Society Series B: Statistical Methodology\/}, qkad080.

\bibitem[\protect\citeauthoryear{Chernozhukov, Demirer, Duflo, and Fernandez-Val}{Chernozhukov et~al.}{2018}]{chernozhukov2018generic}
Chernozhukov, V., M.~Demirer, E.~Duflo, and I.~Fernandez-Val (2018).
\newblock Generic machine learning inference on heterogeneous treatment effects in randomized experiments, with an application to immunization in india.
\newblock Technical report, National Bureau of Economic Research.

\bibitem[\protect\citeauthoryear{Collazos, Garc{\'\i}a, Mej{\'\i}a, Ortega, and Tob{\'o}n}{Collazos et~al.}{2021}]{collazos2021hot}
Collazos, D., E.~Garc{\'\i}a, D.~Mej{\'\i}a, D.~Ortega, and S.~Tob{\'o}n (2021).
\newblock Hot spots policing in a high-crime environment: An experimental evaluation in medellin.
\newblock {\em Journal of Experimental Criminology\/}~{\em 17}, 473--506.

\bibitem[\protect\citeauthoryear{Cox}{Cox}{1958}]{cox1958planning}
Cox, D.~R. (1958).
\newblock Planning of experiments.

\bibitem[\protect\citeauthoryear{Fisher}{Fisher}{1935}]{fisher1935}
Fisher, R.~A. (1935).
\newblock {\em {The Design of Experiments}}.
\newblock Oliver and Boyd.

\bibitem[\protect\citeauthoryear{Guerette and Bowers}{Guerette and Bowers}{2017}]{guerette2017assessing}
Guerette, R.~T. and K.~J. Bowers (2017).
\newblock Assessing the extent of crime displacement and diffusion of benefits: A review of situational crime prevention evaluations.
\newblock {\em Crime Opportunity Theories\/}, 529--566.

\bibitem[\protect\citeauthoryear{Heumos, Schaar, Lance, Litinetskaya, Drost, Zappia, L{\"u}cken, Strobl, Henao, Curion, et~al.}{Heumos et~al.}{2023}]{heumos2023best}
Heumos, L., A.~C. Schaar, C.~Lance, A.~Litinetskaya, F.~Drost, L.~Zappia, M.~D. L{\"u}cken, D.~C. Strobl, J.~Henao, F.~Curion, et~al. (2023).
\newblock Best practices for single-cell analysis across modalities.
\newblock {\em Nature Reviews Genetics\/}~{\em 24\/}(8), 550--572.

\bibitem[\protect\citeauthoryear{Hong and Raudenbush}{Hong and Raudenbush}{2006}]{hong2006evaluating}
Hong, G. and S.~W. Raudenbush (2006).
\newblock Evaluating kindergarten retention policy: A case study of causal inference for multilevel observational data.
\newblock {\em Journal of the American Statistical Association\/}~{\em 101\/}(475), 901--910.

\bibitem[\protect\citeauthoryear{Kolaczyk}{Kolaczyk}{2009}]{kolaczyk2009statistical}
Kolaczyk, E. (2009).
\newblock Statistical analysis of network data: Methods and models.
\newblock {\em Springer Series In Statistics\/}, 386.

\bibitem[\protect\citeauthoryear{Liu and Xie}{Liu and Xie}{2019}]{liu2019cauchy}
Liu, Y. and J.~Xie (2019).
\newblock Cauchy combination test: a powerful test with analytic p-value calculation under arbitrary dependency structures.
\newblock {\em Journal of the American Statistical Association\/}.

\bibitem[\protect\citeauthoryear{Logan, LaCasse, and Lunday}{Logan et~al.}{2023}]{logan2023social}
Logan, A.~P., P.~M. LaCasse, and B.~J. Lunday (2023).
\newblock Social network analysis of twitter interactions: a directed multilayer network approach.
\newblock {\em Social Network Analysis and Mining\/}~{\em 13\/}(1), 65.

\bibitem[\protect\citeauthoryear{Manski}{Manski}{2013}]{manski2013identification}
Manski, C.~F. (2013).
\newblock Identification of treatment response with social interactions.
\newblock {\em The Econometrics Journal\/}~{\em 16\/}(1), S1--S23.

\bibitem[\protect\citeauthoryear{Puelz, Basse, Feller, and Toulis}{Puelz et~al.}{2022}]{puelz2022graph}
Puelz, D., G.~Basse, A.~Feller, and P.~Toulis (2022).
\newblock A graph-theoretic approach to randomization tests of causal effects under general interference.
\newblock {\em Journal of the Royal Statistical Society Series B: Statistical Methodology\/}~{\em 84\/}(1), 174--204.

\bibitem[\protect\citeauthoryear{Rosenbaum}{Rosenbaum}{2011}]{rosenbaum2011some}
Rosenbaum, P.~R. (2011).
\newblock Some approximate evidence factors in observational studies.
\newblock {\em Journal of the American Statistical Association\/}~{\em 106\/}(493), 285--295.

\bibitem[\protect\citeauthoryear{Rubin}{Rubin}{1980}]{rubin1980randomization}
Rubin, D.~B. (1980).
\newblock Randomization analysis of experimental data: The fisher randomization test comment.
\newblock {\em Journal of the American statistical association\/}~{\em 75\/}(371), 591--593.

\bibitem[\protect\citeauthoryear{Stouffer, Suchman, DeVinney, Star, and Williams~Jr}{Stouffer et~al.}{1949}]{stouffer1949}
Stouffer, S.~A., E.~A. Suchman, L.~C. DeVinney, S.~A. Star, and R.~M. Williams~Jr (1949).
\newblock The american soldier: Adjustment during army life.(studies in social psychology in world war ii), vol. 1.

\bibitem[\protect\citeauthoryear{Traag, Waltman, and Van~Eck}{Traag et~al.}{2019}]{traag2019leiden}
Traag, V.~A., L.~Waltman, and N.~J. Van~Eck (2019).
\newblock From louvain to leiden: guaranteeing well-connected communities.
\newblock {\em Scientific reports\/}~{\em 9\/}(1), 5233.

\bibitem[\protect\citeauthoryear{Zhang, Phillips, Rogers, Baker, Chesler, and Langston}{Zhang et~al.}{2014}]{zhang2014imbea}
Zhang, Y., C.~A. Phillips, G.~L. Rogers, E.~J. Baker, E.~J. Chesler, and M.~A. Langston (2014).
\newblock On finding bicliques in bipartite graphs: a novel algorithm and its application to the integration of diverse biological data types.
\newblock {\em BMC bioinformatics\/}~{\em 15}, 1--18.

\bibitem[\protect\citeauthoryear{Zhang and Zhao}{Zhang and Zhao}{2024}]{zhang2021multiple}
Zhang, Y. and Q.~Zhao (2024).
\newblock Multiple conditional randomization tests for lagged and spillover treatment effects.
\newblock {\em Biometrika\/}, asae042.

\end{thebibliography}

\newpage
\appendix
\renewcommand\thefigure{A\arabic{figure}}
\setcounter{figure}{0}
\setcounter{page}{1}
\renewcommand{\thepage}{A-\arabic{page}}
\renewcommand\thetable{A\arabic{table}}
\setcounter{table}{0}
\normalsize
\doublespacing

{\centering
    \section*{Supplementary Material}
    \vspace{1cm}
}

\section{Proof of Results}\label{a:proof}
\subsection{Details of exposure monotone statistics}\label{appdx:exp-mono-stat}
Verifying the exposure monotonicity follows from arguments similar to those in Proposition 1 and 2 of~\cite{caughey2023bounded}. Let $\eta_i \geq 0 \geq \xi_i$, $i \in \mathcal U$ and define $N_{z,\expov} = \sum_{i\in\mU} I_i(z, \expov)$. 
To show that the difference-in-means statistic~\eqref{eq:diff-in-mean} satisfies the exposure-monotone property, note that
\begin{align*}
&t_{\expov, \expov'}^{\mathrm{DiM}}(z, y_{\eta\xi}; \mathcal C)  - t_{\expov, \expov'}^{\mathrm{DiM}}(z, y; \mathcal C) \\
&= \sum_{i\in\mU} \frac{I_i(z, \expov')}{N_{z,\expov'}}(\psi_1(y_i + \eta_i) - \psi_1(y_i)) - \sum_{i\in\mU} \frac{I_i(z, \expov)}{N_{z,\expov}} (\psi_0(y_i + \xi_i) - \psi_0(y_i)) \ge 0,
\end{align*}
because both $\psi_1(\cdot)$ and $\psi_0(\cdot)$ are non-decreasing, as desired. From the expression above it is clear that $\psi_1$ and $\psi_0$ can also vary across $i$, as long as they are non-decreasing.

To show the rank-sum statistic~\eqref{eq:rank-stat} satisfies the exposure-monotone property, recall when there are ties we define the value of $\varphi(\cdot)$ to be the average of $\varphi(\cdot)$ evaluated at those ranks with ties broken by unit ordering, so it suffices to consider the ranks defined by
\[
r_i(y_\mU) = \sum_{j \in \mU} \delta_{ij}(y_i, y_j), \quad \text{where}~\delta_{ij}(y_i, y_j) = \mathbbm 1\{y_i > y_j\} + \mathbbm 1\{y_i = y_j, i \geq j\}.
\]
The desired conclusion now follows from Lemma A3 in~\cite{caughey2023bounded} by mapping our $(y_{\mU}, \mU, I_i(z, \expov'))$ to their $(\boldsymbol{y}, [n], \mathbbm 1\{z_i=1\})$.

\subsection{Proof of Theorem~\ref{thm:focaltest_mono}}
\begin{proof}[Proof of Theorem~\ref{thm:focaltest_mono}]
We will proceed the proof in three steps. Throughout $\{\mathbb S_k\}_k$ is considered fixed.

\underline{1. $\pval_k$ is valid for $\widetilde H_{0k}$:} We will apply Theorem 1 in~\cite{basse2019randomization}. Under the non-uniform Bernoulli design, it suffices to consider treatments within $\mathbb S_{\leq k}$. Given $\mathbb S_{\leq k}$, define 
\[
\mathbb C_k = \{(\mathcal U_k, ~\mathcal Z_k):~ \mathcal U_k \subseteq \mathbb S_{\leq k}, ~\mathcal Z_k \subseteq \{0,1\}^{|\mathbb S_{\leq k}|}\},
\]
the space of ``conditioning event"~\citep{basse2019randomization}. The preprocessing step in Algorithm~\ref{algo:focalrand_single_general} defines the following ``conditioning mechanism" that maps any $Z_{\mathbb S_{\leq k}} \in \{0,1\}^{|\mathbb S_{\leq k}|}$ into a (degenerate) distribution over $\mathbb C_k$:
\[
m_k(\mathcal U_k, \mathcal Z_k | Z_{\mathbb S_{\leq k}} ) = f_k(\mathcal U_k | Z_{\mathbb S_{\leq k}}) g_k(\mathcal Z_k | \mathcal U_k, Z_{\mathbb S_{\leq k}})
\]
where 
{\footnotesize \begin{align*}
    f_k(\mathcal U_k | Z_{\mathbb S_{\leq k}}) &= \mathbbm 1\bigg\{\mathcal U_k = \bigcup_{\ell \in [L_k]}\af(Z_{\mathbb S_{\leq k}}; \mS_{k, \ell}) \bigg\} \\
    &= \mathbbm 1\bigg\{\mathcal U_k = \Big\{i \in \bigcup_{\ell \in [L_k]} \ef(\mS_{k, \ell}) : Z_{i} = 0, ~\expof_i(Z_{\mathbb S_{\leq k}}) \in \{\expov_k, \expov_{k+1}\} \Big\} \bigg\} \\
    g_k(\mathcal Z_k | \mathcal U_k, Z_{\mathbb S_{\leq k}} ) &\\
    = \mathbbm 1 \bigg\{\mathcal Z_k = &\Big\{Z' \in \mathrm{supp}(P) \cap \mathbb S_{\leq k}:~ Z'_{\mathbb S_{\leq k}\setminus(\nei(\mathcal U_k)\setminus \mathbb S_{<k})} = Z_{\mathbb S_{\leq k}\setminus(\nei(\mathcal U_k)\setminus \mathbb S_{<k})},~ \expof_i(Z_{\mathbb S_{\leq k}}') \in \{\expov_k, \expov_{k+1}\},~\forall i \in \mathcal U_k \Big\} \bigg\}.
\end{align*}}The test conditions on treatments $Z_{\mathbb S_{<k}}$ so we remove it from $\nei(\mathcal U_k)$ in $\mathbb S_{\leq k} \setminus (\nei(\mathcal U_k)\setminus \mathbb S_{<k}) = (\mathbb S_{\leq k}\setminus \nei(\mathcal U_k)) \cup \mathbb S_{< k}$.
Both $\expof_i(Z_{\mathbb S_{\leq k}})$ and $\expof_i(Z_{\mathbb S_{\leq k}}')$ are well-defined because $\expof_i(Z)$ depends only on treatments of $\nei_i$, or treatments of any other units included in $\er(\mS_{k,\ell})$, 
and by construction these units are included in $\er(\mS_{k,\ell}) \subseteq \mathbb S_{k} \subseteq \mathbb S_{\leq k}$ for all $i \in \ef(\mS_{k,\ell})$. 
The potential outcomes relevant for $\widetilde H_{0k}$, $\{y_i(0, \expov_k), y_i(0, \expov_{k+1})\}$, are imputable for each $i \in \mathcal U_k$ under any $Z_{\mathbb S_{\leq k}} \in \mathcal Z_k$ by construction. Since the test statistic uses only outcomes in $\mathcal U_k$, the test statistic is also imputable.

It remains to show that the randomization distribution in~\eqref{eq:rdist_single} coincides with $\mathbbm P(Z_{\mathbb S_{\leq k}} | \mathcal C_k^\obs) \propto m_k(\mathcal C_k^\obs | Z_{\mathbb S_{\leq k}}) P(Z_{\mathbb S_{\leq k}})$ where $\mathcal C_k^\obs = (\mathcal U_k^\obs, \mathcal Z_k^\obs) \sim m_k(\cdot|\Zobs_{\mathbb S_{\leq k}})$. That is, the proposed randomization distribution coincides with the conditional distribution of $Z_{\mathbb S_{\leq k}}$ induced by the conditioning mechanism and the design.

Given $\Zobs_{\mathbb S_{\leq k}}$ and the induced set of all active focal units $\af(\Zobs_{\mathbb S_{\leq k}}; \mathbb S_{k}) := \bigcup_{\ell \in [L_k]} \af(\Zobs_{\mathbb S_{\leq k}}; \mS_{k, \ell})$ which equals $\mathcal U_k^\obs$ almost surely, as well as the set of all active randomization units $\ar(\Zobs_{\mathbb S_{\leq k}}; \mathbb S_{k}) := \bigcup_{\ell \in [L_k]} \ar(\Zobs_{\mathbb S_{\leq k}}; \mS_{k, \ell})$, the $r_{k, \ell}(\cdot)$ in~\eqref{eq:rdist_single} defines a distribution on all units in $\mathbb S_{\leq k}$ by 
\begin{align*}
    r_k(Z_{\mathbb S_{\leq k}}=z_{\mathbb S_{\leq k}}) &\propto \mathbbm 1\{z_i = Z^{\mathrm{obs}}_i,~ \forall i \in \mathbb S_{\leq k}\setminus( \ar(\Zobs_{\mathbb S_{\leq k}}; \mathbb S_{k}) \setminus \mathbb S_{<k}) \} \\
    &~~~~~~ \times \mathbbm 1\{\expof_i(z_{\mathbb S_{\leq k}}) \in \{\expov_k, \expov_{k+1}\},~\forall i \in \af(\Zobs_{\mathbb S_{\leq k}}; \mathbb S_{k})\} \\
    &~~~~~~ \times P(z_{\mathbb S_{\leq k}}).
\end{align*}
Hence, it suffices to show that the two indicators above defines the same conditioning mechanism as $m_k$. That is,
\begin{equation}\label{eq:indicators}
\begin{split}
    &\mathbbm 1\{z_i = Z^{\mathrm{obs}}_i,~ \forall i \in \mathbb S_{\leq k} \setminus ( \ar(\Zobs_{\mathbb S_{\leq k}}; \mathbb S_{k}) \setminus \mathbb S_{<k} ) \} \\
    &~~~~~~ \times \mathbbm 1\{\expof_i(z_{\mathbb S_{\leq k}}) \in \{\expov_k, \expov_{k+1}\},~\forall i \in \af(\Zobs_{\mathbb S_{\leq k}}; \mathbb S_{k})\} \\
    =&~ f_{k}(\mathcal U_k^\obs | z_{\mathbb S_{\leq k}}) g_k(\mathcal Z_k^\obs | \mathcal U_k^\obs, z_{\mathbb S_{\leq k}}),
\end{split}
\end{equation}
almost surely. Firstly note that 
\begin{equation}\label{eq:fk_equiv}
    f_{k}(\mathcal U_k^\obs | z_{\mathbb S_{\leq k}}) = 1 \iff \af(\Zobs_{\mathbb S_{\leq k}}; \mS_{k, \ell}) = \af(z_{\mathbb S_{\leq k}}; \mS_{k, \ell}) ~\forall \ell \in [L_k].
\end{equation}
Suppose the LHS of~\eqref{eq:indicators} is $1$, then it's easy to see that $\af(z_{\mathbb S_{\leq k}}; \mS_{k, \ell}) = \af(\Zobs_{\mathbb S_{\leq k}}; \mS_{k, \ell})$ for all $\ell \in [L_k]$ by~\eqref{eq:active_focal}, which implies $f_{k}(\mathcal U_k^\obs | z_{\mathbb S_{\leq k}}) = 1$ by~\eqref{eq:fk_equiv}. Also, we can see $z_{\mathbb S_{\leq k}} \in \mathcal Z^\obs$, which implies $g_k(\mathcal Z_k^\obs | \mathcal U_k^\obs, z_{\mathbb S_{\leq k}}) = 1$. Hence RHS is also $1$. 
On the other hand, suppose the LHS of~\eqref{eq:indicators} is $0$, then either for some unit $i \in \mathbb S_{\leq k} \setminus ( \ar(\Zobs_{\mathbb S_{\leq k}}; \mathbb S_{k}) \setminus \mathbb S_{<k} )$, $z_i \neq \Zobs_i$, or for some unit $i \in \af(\Zobs_{\mathbb S_{\leq k}}; \mathbb S_k)$, $\expof_i(z_{\mathbb S_{\leq k}}) \notin \{\expov_k, \expov_{k+1}\}$. For the first case, 
\begin{align*}
    &\big\{Z' \in \mathrm{supp}(P) \cap \mathbb S_{\leq k}:~ Z'_{\mathbb S_{\leq k}\setminus(\nei(\mathcal U_k)\setminus \mathbb S_{<k})} = z_{\mathbb S_{\leq k}\setminus(\nei(\mathcal U_k)\setminus \mathbb S_{<k})},~ \expof_i(Z_{\mathbb S_{\leq k}}') \in \{\expov_k, \expov_{k+1}\},~\forall i \in \mathcal U_k \big\} \\
    &\qquad \neq \mathcal Z^\obs,
\end{align*}
because for any $z_{\mathbb S_{\leq k}}' \in \mathcal Z^\obs$ we will have $z'_i = \Zobs_i \neq z_i$ for that unit $i$, so that $g_k(\mathcal Z_k^\obs | \mathcal U_k^\obs, z) = 0$. 
For the second case, there exists $\ell \in [L_k]$ and $i \in \af(\Zobs_{\mathbb S_{\leq k}}; \mS_{k, \ell})$ such that $\expof_i(z_{\mathbb S_{\leq k}}) \notin \{\expov_k, \expov_{k+1}\}$. But that means $\af(z_{\mathbb S_{\leq k}}; \mS_{k, \ell}) \neq \af(\Zobs_{\mathbb S_{\leq k}}; \mS_{k, \ell})$ so that $f_{k}(\mathcal U_k^\obs | z_{\mathbb S_{\leq k}}) = 0$ by~\eqref{eq:fk_equiv}. As a result, \eqref{eq:indicators} holds, and by Theorem~1 in \cite{basse2019randomization}, for any $\alpha \in [0, 1]$,
\[
\mathbbm P\left( \pval_k \leq \alpha ~\big|~ \Zobs_{\mathbb S_{< k}}, \mathcal U^\obs_k; \widetilde H_{0k} \right) \leq \alpha,
\]
where the probability is taken with respect to the design conditional on realizations of $\Zobs_{\mathbb S_{< k}}$ and $\mathcal U_k^\obs$ being the active focal units. This implies for any $\alpha \in [0, 1]$,
\[
\mathbbm P\left( \pval_k \leq \alpha ~\big|~ \Zobs_{\mathbb S_{< k}}; \widetilde H_{0k} \right) \leq \alpha,
\]
where the probability is taken with respect to the design conditional on realizations of $\Zobs_{\mathbb S_{< k}}$ solely.

\underline{2. $p_k$ is valid for $H_{0k}$:}
This follows from a similar argument to Theorem A1 of~\cite{caughey2023bounded}, by recognizing the two exposure levels $\expov_k$ and $\expov_{k+1}$ as ``controlled" and ``treated" separately as in the binary treatment setting with no interference. The $p$-value $\pval_k$ in \eqref{eq:pval_1} can be equivalently written as (ignoring conditioning and subscripts for brevity)
\[
\pval_k = \mathbbm P_{\expof \sim r_{k, \expof}} \big(t(\expof, \Yobs ; \mathcal U_k^\obs) \geq t(\expof^\obs, \Yobs; \mathcal U_k^\obs) \big),
\]
where $\expof = (\expof_i)_{i \in \mathcal U_k^\obs}$ is the (random) vector of exposures for units in $\mathcal U_k^\obs$, $\mathcal U_k^\obs = \af(\Zobs_{\mathbb S_{\leq k}}; \mathbb S_k)$ the set of all active focal units almost surely, $\expof_i^\obs = \expof_i(\Zobs_{\mathbb S_{\leq k}})$ the observed exposures, and 
\[
r_{k, \expof}(\expof = (\expov^i)_{i \in \mathcal U_k^\obs}) = \sum_{z_{\mathbb S_{\leq k}}} \mathbbm 1\{\expof_i(z_{\mathbb S_{\leq k}}) = \expov^i,~\forall i \in \mathcal U_k^\obs\} r_k(z_{\mathbb S_{\leq k}}), \quad \forall (\expov^i)_{i \in \mathcal U_k^\obs} \in \{\expov_{k}, \expov_{k+1}\}^{|\mathcal U_k^\obs|},
\]
is the distribution over the exposures of $\mathcal U_k^\obs$ induced by the randomization distribution $r_k$. 
Let $\tau_{k,i} := y_i(0, \expov_{k+1}) - y_i(0, \expov_k)$ be the (unobserved) true ``treatment effect" of exposure $\expov_{k+1}$ versus $\expov_{k}$. Suppose $H_{0k}$ holds, then we have $\tau_{k,i} \leq 0$ for all $i$. Note that for all $i \in \mathcal U_k^\obs$,
\begin{align*}
    &Y^{\mathrm{obs}}_i - y_i(0, \expov_{k+1}) = \mathbbm 1\{\expof_i^{\mathrm{obs}} = \expov_k\} (-\tau_{k,i}) \geq 0 \\
    &Y^{\mathrm{obs}}_i - y_i(0, \expov_k) = \mathbbm 1\{\expof_i^{\mathrm{obs}} = \expov_{k+1}\} (\tau_{k,i}) \leq 0.
\end{align*}
Since $Y^{\mathrm{obs}}_i = y_i(0, \expof_i) + \mathbbm 1\{\expof_i = \expov_{k+1}\} (Y^{\mathrm{obs}}_i - y_i(0, \expov_{k+1})) + \mathbbm 1\{\expof_i = \expov_{k}\} (Y^{\mathrm{obs}}_i - y_i(0, \expov_{k}))$, 
where we write $y_i(0, \expof_i) := \mathbbm 1\{\expof_i=\expov_{k+1}\} y_i(0, \expov_{k+1}) + \mathbbm 1\{\expof_i=\expov_{k}\} y_i(0, \expov_{k})$ as the shorthand for the observed outcome under the randomly drawn $w_i$,
by the definition of exposure monotonicity, $t( \expof, Y^{\mathrm{obs}}; \mathcal U_k^\obs ) \geq t( \expof, (y_i(0, \expof_i))_i; \mathcal U_k^\obs )$.
Hence, 
\begin{align*}
    \pval_k &= \mathbbm P_{\expof \sim r_{k, \expof}} \big(t(\expof, \Yobs ; \mathcal U_k^\obs) \geq t(\expof^\obs, \Yobs; \mathcal U_k^\obs) \big) \\
    &\geq \mathbbm P_{\expof \sim r_{k, \expof}} \big(t( \expof, (y_i(0, \expof_i))_i; \mathcal U_k^\obs ) \geq t( \expof^{\mathrm{obs}}, Y^{\mathrm{obs}}; \mathcal U_k^\obs ) \big) =: \pval^*_k.
\end{align*}
Under $H_{0k}$, $\pval^*_k$ stochastically dominates Uniform $(0,1)$ following the same reasoning in step 1 with only a modification that we impute missing potential outcomes using both $\Yobs$ and $(\tau_{k,i})_{i}$. Hence, for any $\alpha \in [0,1]$,
\[
\mathbbm P\left( \pval_k \leq \alpha \big| \Zobs_{\mathbb S_{< k}}; H_{0k} \right) \leq \mathbbm P\left( \pval^*_k \leq \alpha \big| \Zobs_{\mathbb S_{< k}}; H_{0k} \right) \leq \alpha.
\]

\underline{3. Stochastically larger than Uniform:}
Recall in Algorithm~\ref{algo:focalrand_multiple} from each step $k=1,\ldots,K-1$, we apply Algorithm~\ref{algo:focalrand_single_general} with module set $\mathbb S_k$, observed treatment $\Zobs$, conditional set $\mathbb S_{<k}$, and get $p$-value $\pval_k$. Essentially we only input the observed treatment $\Zobs_{\mathbb S_{\leq k}}$ on $\mathbb S_{\leq k}$.
As a result, $\pval_k$ can be written as a function of $\pval_k(\Zobs_{\mathbb S_{\leq k}}) = \pval_k(\Zobs_{\mathbb S_{1}}, \ldots, \Zobs_{\mathbb S_{k}})$. The result follows from Lemma~\ref{lem:recur_pval} by taking $Q_k = \Zobs_{\mathbb S_{k}}$.
\end{proof}

\subsection{Proof of Theorem~\ref{thm:generaltest}}
\begin{proof}[Proof of Theorem \ref{thm:generaltest}]
The fact that each $\pval_k$ is valid under $H_{0k}$, that is,
\begin{equation}\label{eq:pval_general}
    \mathbbm P_{Z^{\mathrm{obs}}_{V_k} \sim P_k(\cdot)} \left(\pval_k \leq \alpha_k ~\big|~ \{\mathcal C_j^k\}_{j\in J_k}, ~(\Zobs_i)_{i \in V_{<k}}, ~H_{0k}\right) \leq \alpha_k,
\end{equation}
for all $k$ and $\alpha_k \in [0,1]$, follows directly from Proposition \ref{prop:bicliquert_mono} in Appendix~\ref{appdx:biclique_test} by taking $P(\cdot)$ as $P(\cdot | Z^{\mathrm{obs}}_{\cup_{k'<k} V_{k'}})$. 

In addition, the $k$-th biclique decomposition $\{\mathcal C_j^k\}_{j \in J_k}$ is computed using the $P_k(Z_{V_k}) = P(Z_{V_k} | Z^{\mathrm{obs}}_{V_{<k}})$ (and the partition of the network, which we considered as fixed), so that it is a function of $Z^{\mathrm{obs}}_{V_{<k}}$.
Given the $k$-th biclique decomposition $\{\mathcal C_j^k\}_{j \in J_k}$, the biclique test defines a conditioning mechanism $s_k$ that maps any $Z_{V_k}$ into the unique biclique $\mathcal C = (\mathcal U, \mathcal Z) \in \{\mathcal C_j^k\}_{j \in J_k}$ such that $Z_{V_k} \in \mathcal Z$. $s_k$ maps the observed $\Zobs_{V_k}$ into the $\mathcal C^{k, \obs} = (\mathcal U^{k, \obs}, \mathcal Z^{k, \obs})$ that is used in the randomization test, so that $\mathcal C^{k, \obs}$ is a function of $\Zobs_{V_k}$. Moreover, the $k$-th randomization distribution (Line 3 of Algorithm~\ref{algo:cliquetest}) is given by
\[
r_k(Z_{V_k}) \propto \mathbbm 1\{Z_{V_k} \in \mathcal Z^{k, \obs}\} P_k(Z_{V_k}) = \mathbbm 1\{Z_{V_k} \in \mathcal Z^{k, \obs}(Z^{\mathrm{obs}}_{V_k})\} P(Z_{V_k} | Z^{\mathrm{obs}}_{V_{<k}}),
\]
which shows that $r_k(\cdot)$ is a function of $Z^{\mathrm{obs}}_{V_{\leq k}}$. Altogether, the $k$-th $p$-value $\pval_k$ is a function of $Z^{\mathrm{obs}}_{V_{\leq k}}$. Taking $Q_k = Z^{\mathrm{obs}}_{V_k}$ in Lemma \ref{lem:recur_pval}, from \eqref{eq:pval_general}, we conclude that $(\pval_k)_{k}$ are stochastically larger than uniform, which finishes the proof.
\end{proof}

\section{Combination of $p$-values}\label{appdx:pvalcomb}
There are many other ways one can combine the $p$-values resulting from each sub-hypothesis $H_{0k}$ apart from Fisher's rule. A general condition is given by the lemma below.
\begin{lemma}[\cite{rosenbaum2011some}, Lemma 1]\label{lem:combine}
If $f: [0,1]^{K} \to \mathbb R$ is monotone decreasing in each of the $K$ coordinates and $(P_1,\ldots, P_{K})$ is stochastically larger than uniform, then $f(P_1,\ldots, P_{K}) \lesssim^{\mathrm{1st}} f(P_1^*, \ldots, P_{K}^*)$ where $P_k^* \overset{iid}{\sim} U(0,1)$ and $\lesssim^{\mathrm{1st}}$ denotes first-order stochastic dominance. 
\end{lemma}

Therefore one can compare the observed value of $f(P_1,\ldots, P_{K})$ with quantiles of $f(P_1^*, \ldots, P_{K}^*)$ for $P_k^* \overset{iid}{\sim} U(0,1)$, which can be calculated exactly or approximated by Monte-Carlo. 
Some examples include Stouffer's weighted $z$-score \citep{stouffer1949} and the Cauchy combination rule~\citep{liu2019cauchy}. Existing theory has pointed out that no single method of combining independent tests of significance is optimal in general~\citep{birnbaum1954combining}. 

In our context, we prefer Fisher's combination for the behavior of its log-transformation when $p$-values are close to $0$ or $1$ (see Table~\ref{tab:comb_extreme}).
When the null is violated at some of the sub-hypothesis $H_{0k}$, and the signal strength is quite strong in the sense that $p_k \to 0$ and $p_{k'} \to 1$ for all $k' \neq k$ (or the other way), where $p_k$ is the $p$-value for testing $H_{0k}$, both Stouffer and Cauchy combination involve expression of $+\infty - \infty$, making the combined $p$-value unstable or even powerless depending on the weighting schemes. This is the case in DGP2 of Section~\ref{sec:simu_results} when $\tau \neq 0$. In that setup, $p_1$ provides an opposite and more powerful signal of the monotone null compared to $p_2$ and $p_3$, as shown in Figure~\ref{fig:DGP2_allrej} where we plot the rejections of each individual hypothesis and the combined rejection of the monotone hypothesis. Here we truncate each $p$-value to be within $[\epsilon, 1-\epsilon]$ with $\epsilon = 10^{-4}$ before the CDF transformation $\Phi^{-1}(\cdot)$, and the truncation binds when $|\tau|$ is large. When using the expected number of units with exposures $\expov_k$ or $\expov_{k+1}$ for $k = 1,2,3$ to be the weights, the combination happens to offset the signals from the three $p$-values, so that the Stouffer's method becomes almost powerless in rejecting $\tau \neq 0$. For other weighting schemes it could be possible that one of the $p$-values becomes overly dominant, so that the Stouffer's method is again powerless in rejecting $\tau < 0$ or $\tau > 0$. In any case, the combination might yield undesirable results.
\begin{table}[!h]
\centering
\begin{tabular}{l|c|c|c}
\hline
     &  $p_k \to 0$ & $p_{k'} \to 1$ & Combination of $p_k$ and $p_{k'}$\\
   \hline
   Fisher & $\log p_k \to -\infty$ & $\log p_{k'} \to 0$ & FCT $\to -\infty$ \\
   Stouffer & $\Phi^{-1}(p_k) \to -\infty$ & $\Phi^{-1}(p_{k'}) \to +\infty$ & Stouffer $\to +\infty -\infty$ \\
   Cauchy & $\tan((0.5-p_k) \pi) \to -\infty$ & $\tan((0.5-p_{k'}) \pi) \to +\infty$ & CCT $\to +\infty -\infty$ \\
\hline
\end{tabular}
\caption{Behavior of transformed $p$-values under extreme values for Fisher (FCT), Stouffer and Cauchy combination rule (CCT).}
\label{tab:comb_extreme}
\end{table}

\begin{figure}[!ht]
    \centering
    \includegraphics[width = 0.4\textwidth]{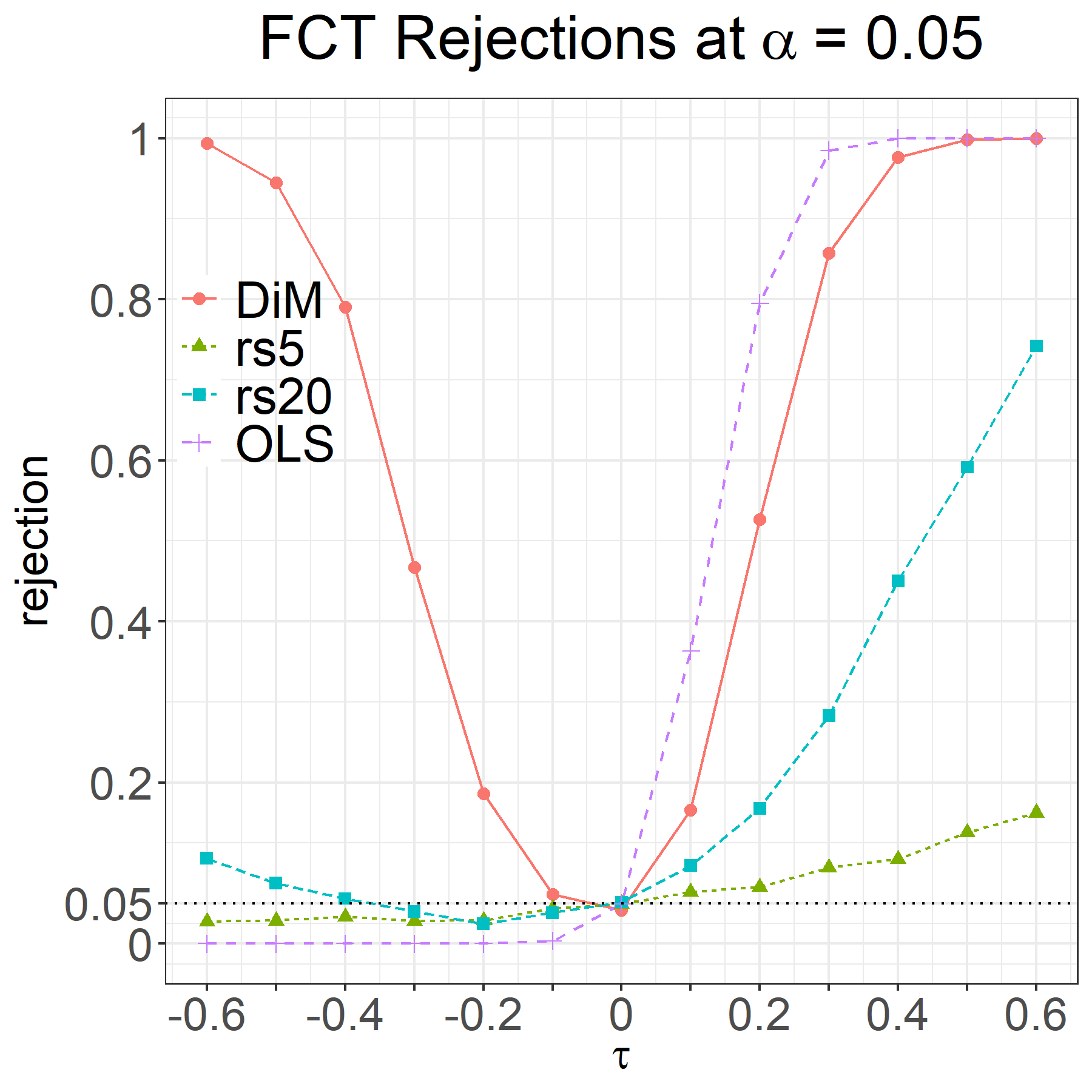}
    \includegraphics[width = 0.4\textwidth]{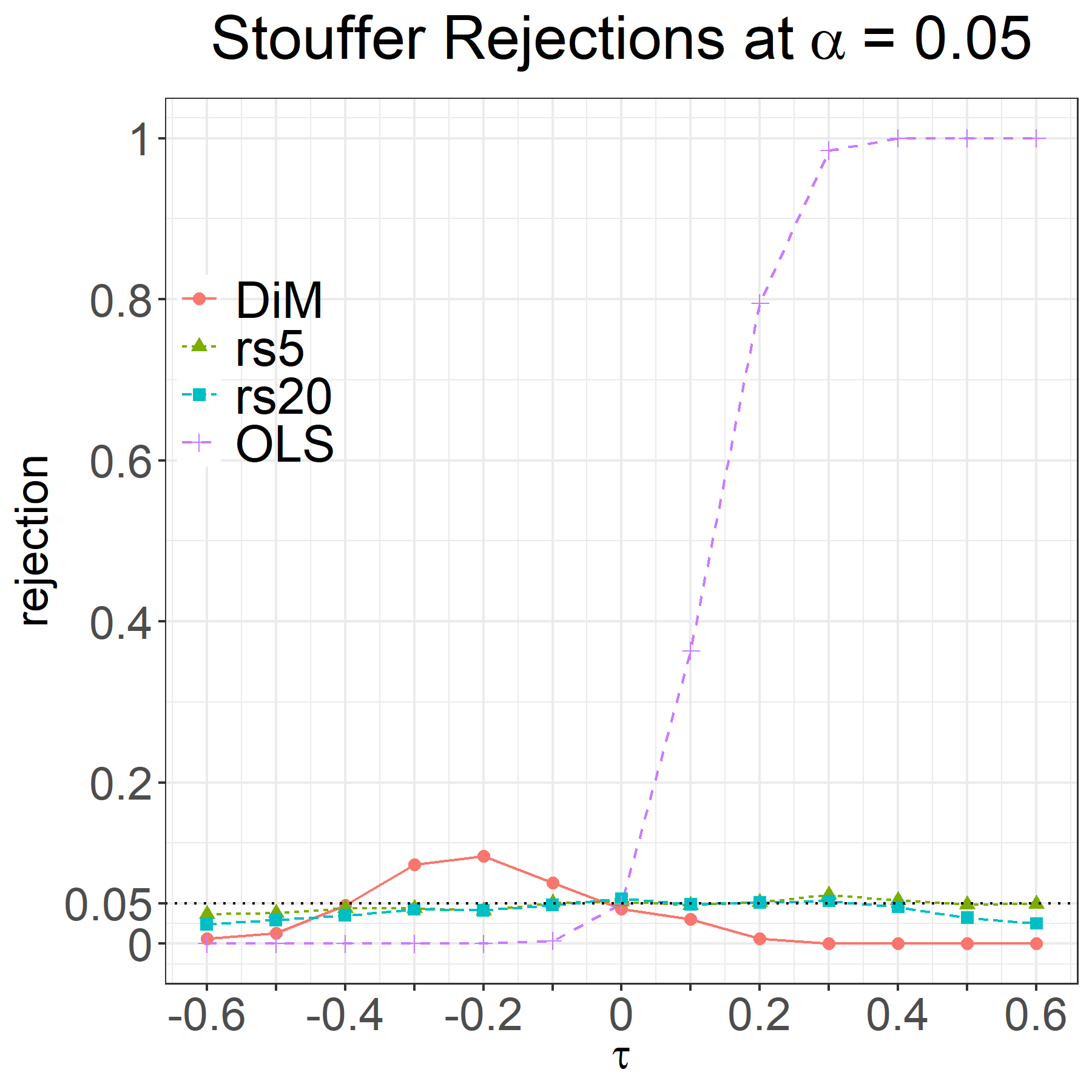}
    \includegraphics[width = 1.0\textwidth]{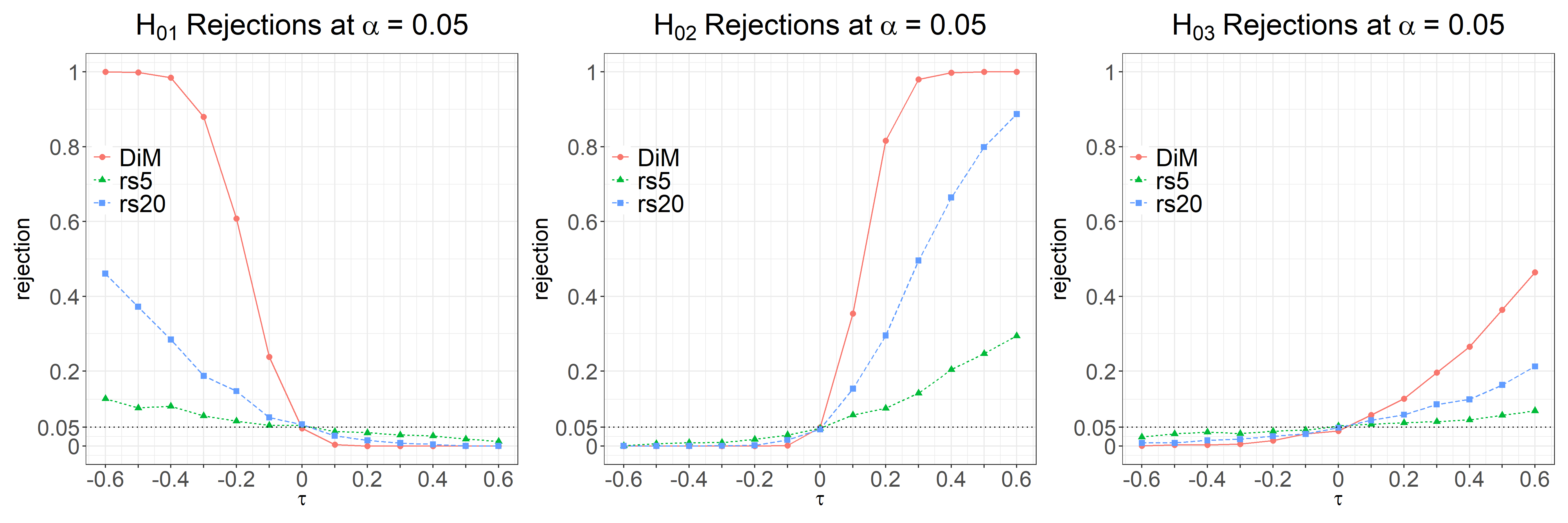}
    \caption{Rejection of overall monotone hypothesis and individual hypotheses under DGP2.}
    \label{fig:DGP2_allrej}
\end{figure}

Another approach would be to not split the network at all, test each individual null hypothesis $H_{0k}$ on the entire network, and combine the resulting $p$-values using Bonferroni's method. Although Bonferroni's rule usually leads to power loss, this approach may be competitive in settings where we suspect that the violation of monotonicity is mainly due to one particular contrast, e.g., having ``1 neighbor treated" vs.\ ``0 neighbors treated", or testing all but a few of the sub-hypotheses is difficult due to, for example, lack of eligible or active focal units after splitting the network.

One final remark is that the default Fisher's combination puts equal weights on each of the $p$-values. This is plausible without any prior knowledge on the spillover effect. Some interesting exceptions may exist when, for example, we suspect ``diminishing returns'' in the spillover effect, or the violation of monotonicity is mainly due to one particular contrast of the exposures. 
In such cases, it may be better to put more weights on some hypotheses over others as that could lead to more power compared to the default Fisher's combination rule. The weighted Fisher's combination does not suffer from the issues observed in Stouffer's combination, where rejection signals are offset or some $p$-values become overly dominant. As long as one of the sub-tests provides a strong rejection signal, the weighted Fisher's combination will tend to reject the overall monotone hypothesis, as illustrated in Table~\ref{tab:comb_extreme}.

\section{More on the Module-based Algorithm}
\subsection{Further extensions of Algorithm~\ref{algo:focalrand_single_general}}\label{appdx:gen_algo_s}
\paragraph{Exposures beyond neighbors}~For $\expof_i(z)$ that depends on the treatments of units beyond $\nei_i$, we can analogously define $\er(\mS_\ell)$ as the union of units whose treatment status determines $\expof_i(z)$ for all $i \in \ef(\mS_\ell)$, i.e., $\er(\mS_\ell)$ is a set of units such that for all $i \in \ef(\mS_\ell)$, $y_i(z) = y_i(z')$ for all $z,z' \in \{0,1\}^N$ such that $z_{\er(\mS_\ell) \cup\{i\}} = z'_{\er(\mS_\ell) \cup\{i\}}$. The remaining definitions of modules, module sets, active focal/randomization units and procedures in Algorithm~\ref{algo:focalrand_single_general} are kept unchanged.

\paragraph{Designs beyond non-uniform Bernoulli}~The assumption of Bernoulli design is used to calculate the conditional randomization distribution over active focal units in Lines 4-5 and Lines 8-12 in Algorithm~\ref{algo:focalrand_single}, or Lines 2-4 in Algorithm~\ref{algo:focalrand_single_general}. In particular, when sampling from the randomization distribution in Lines 8-12 of Algorithm~\ref{algo:focalrand_single} or Lines 2-4 in Algorithm~\ref{algo:focalrand_single_general}, we rely on the assumption that treatments on $\er(\mS_\ell)$ for different $\ell$ are independent, which follows from the Bernoulli design. We can relax the Bernoulli design to a ``clustered" one where the design has a factorization $P(Z) = P_{c}(Z_{\mathbb S^c}) \prod_{\ell} P_\ell(Z_{\mS_\ell})$ by appropriately choosing the modules, so that treatments across $(\mS_\ell)_\ell$ are independent of each other, while not affecting the validity of the randomization test. Under the clustered design, the randomization distribution for $\widetilde Z_\ell$ can be calculated exactly or approximated to any precision using the same procedure as in Section~\ref{sec:extalgo1}. Note, however, that such a calculation is tractable only when the size of each cluster, $|\mS_\ell|$, is moderate. Otherwise, calculating the randomization distribution will be computationally challenging. In such cases, it's better to use the clique-based approach in Section~\ref{sec:method_clique}, in which the cluster naturally provides a partition of the whole network.

\paragraph{Overlaps in modules}~In Definition~\ref{def:modules} we define a module set $\mathbb S$ to be a collection of disjoint modules $\{\mS_1, \ldots, \mS_L\}$ so that $\mS_\ell \cap \mS_{\ell'} = \emptyset$ for all $\ell \neq \ell'$. It turns out that requiring modules to have disjoint eligible randomization units, i.e., $\er(\mS_\ell) \cap \er(\mS_{\ell'}) = \emptyset$ for all $\ell \neq\ell'$, is sufficient for a randomization test by a conditioning argument, leading to a more efficient use of the network information. 
Formally, we call a collection of modules $\mathbb S = \{\mS_1, \ldots, \mS_L\}$ a {\em generalized module set} if (1) $\er(\mS_\ell) \cap \er(\mS_{\ell'}) = \emptyset$ for all $\ell \neq\ell'$, and (2) for all $\ell$, $\er^*(\mS_\ell) := \er(\mS_\ell) \setminus \bigcup_{\ell' \neq \ell} \mS_{\ell'} \neq \emptyset$. 
In a generalized module set, eligible focal units in one module can serve as eligible randomization units for another module, and vice versa.
However, eligible randomization units for different modules remain disjoint, which implies that eligible focal units for different modules are also disjoint. Requirement (2) further states that each module should also have some eligible randomization units, $\er^*(\mS_\ell)$, that do not overlap with other modules. More specifically, $\er^*(\mS_\ell)$ does not overlap with the eligible focal units of other modules. As we will see soon, requirement (2) is imposed only to rule out degenerate randomization distributions on modules but does not otherwise affect the validity of the test.

An example of a generalized module set is shown in Figure~\ref{fig:ovlap_module}. In the figure, node ``a" serves as both an eligible focal unit in module $\mS_1$ and an eligible randomization unit in module $\mS_2$, while node ``b" serves as both an eligible focal unit in module $\mS_2$ and an eligible randomization unit in module $\mS_1$. The remaining eligible randomization units are therefore $\er^*(\mS_1) = \{\mathrm{r11},\mathrm{r12}\}$ and $\er^*(\mS_2)=\{\mathrm{r2}\}$. Nevertheless, $\er(\mS_1) \cap \er(\mS_2) = \emptyset$ still holds.

\begin{figure}[h]
\centering
\resizebox{0.5\textwidth}{!}{
\begin{tikzpicture}

\draw[red, dashed] (-0.5,-0.75) ellipse (2cm and 2.63cm);
\draw[cyan, dashed] (0.6,-1.4) ellipse (2.3cm and 1.5cm);

\node[red] at (-2,1.5) {$\mS_1$};
\node[cyan] at (3.2,-1) {$\mS_2$};

\node[circle, red, draw, minimum size=0.8cm] (a) at (-0.5,-0.8) {};
\node[rectangle, cyan, draw, minimum size=0.9cm] (ra) at (-0.5,-0.8) {};
\node at (-0.5,-0.8) {a};

\node[circle, cyan, draw, minimum size=0.8cm] (b) at (0.4,-2.2) {};
\node[rectangle, red, draw, minimum size=0.9cm] (rb) at (0.4,-2.2) {};
\node at (0.4,-2.2) {b};

\node[rectangle, red, draw, minimum size=0.8cm] (rs11) at (-1.5,0.65) {};
\node at (-1.5,0.65) {r11};
\node[rectangle, red, draw, minimum size=0.8cm] (rs12) at (0.5,0.65) {};
\node at (0.5,0.65) {r12};
\node[rectangle, cyan, draw, minimum size=0.8cm] (rs2) at (2,-1) {};
\node at (2,-1) {r2};

\draw (rs11) -- (ra);
\draw (rs12) -- (ra);
\draw (ra) -- (rb);
\draw (rb) -- (rs2);

\node[right, font=\large] at (3,1) {\tikz\draw[radius=0.2cm] circle; Eligible focal units};
\node[right, font=\large] at (3,0.3) {\tikz\draw (0,0) rectangle (0.4cm,0.4cm); Eligible randomization units};

\end{tikzpicture}
}
\caption{A generalized module set $\mathbb S$ consisting of two modules outlined by {\color{red} red} and {\color{cyan} blue} dashed circles.}
\label{fig:ovlap_module}
\end{figure}
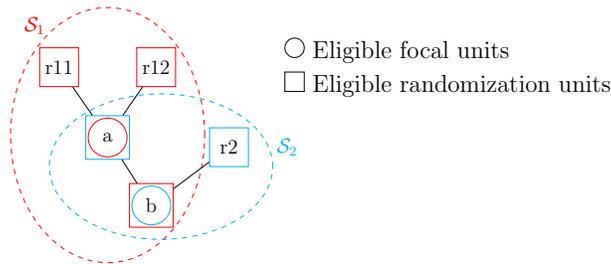

Given the hypothesis to test $H_{0k}$, a generalized module set $\mathbb S = \{\mS_\ell: \ell \in[L]\}$ and the observed treatment $\Zobs$, the active focal and randomization units are defined in the same way as in the main text. The randomization for each active module $\mS_\ell$ on $\ar(\Zobs; \mS_\ell)$, however, should further condition on the treatments in $\ar(\Zobs; \mS_\ell) \setminus \er^*(\mS_\ell)$ at their observed values. Equivalently, we condition on the treatments of all eligible focal units in the generalized module set $\mathbb S$. Hence, the modified test can be easily implemented using Algorithm~\ref{algo:focalrand_single_general} by first augmenting the conditional set $C$ with $\bigcup_{\ell \in [L]} \ef(\mS_\ell)$.

An illustration of the modified test for $H_{0k}: y_i(0,1) \geq y_i(0,2)$, i.e., having 1 versus 2 treated neighbors, is given in Figure~\ref{fig:ovlap_module_rand}. The left panel depicts the active focal and randomization units under the observed $\Zobs$ for the generalized module set in Figure~\ref{fig:ovlap_module}. Only module $\mS_1$ is active with the observed $\Yobs_{\mathrm{a}} = y_{\mathrm{a}}(0, 2)$, since $\Zobs_{\mathrm{b}} = 1$. Because node ``b" serves as an eligible focal unit for module $\mS_2$, its treatment is conditioned at $\Zobs_{\mathrm{b}} = 1$ throughout the randomizations.
Thus, the support of the randomization distribution on $\ar(\Zobs; \mS_1)$ is $(z_{\mathrm{r11}}, z_{\mathrm{r12}}, z_{\mathrm{b}}) \in \{(0,0,1), (0,1,1), (1,0,1)\}$, leading to $\expof_{\mathrm{a}} = 1,1,2$ respectively. The right panel displays $(z_{\mathrm{r11}}, z_{\mathrm{r12}}, z_{\mathrm{b}}) = (0,1,1)$.

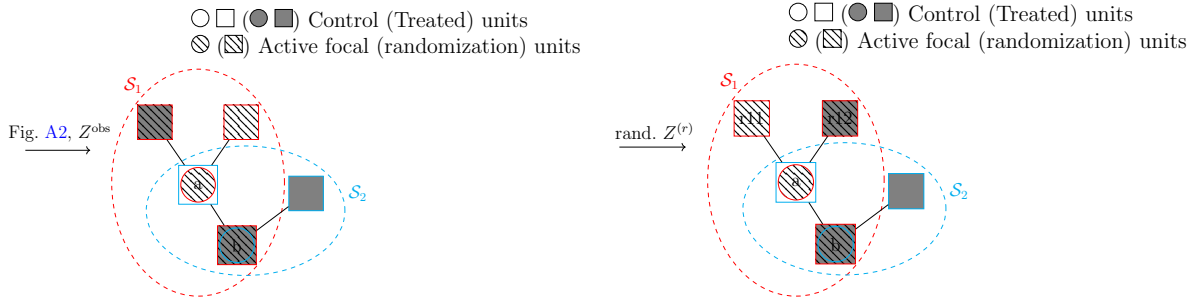
\begin{figure}[h]
\centering
\resizebox{0.48\textwidth}{!}{
\begin{tikzpicture}

\node[above] at (-3.7,0.05) {Fig.~\ref{fig:ovlap_module}, $\Zobs$};
\draw [->] (-4.5,0) -- (-3,0);

\draw[red, dashed] (-0.5,-0.75) ellipse (2cm and 2.63cm);
\draw[cyan, dashed] (0.6,-1.4) ellipse (2.3cm and 1.5cm);

\node[red] at (-2,1.5) {$\mS_1$};
\node[cyan] at (3.2,-1) {$\mS_2$};

\node[circle, draw=red, pattern=north west lines, minimum size=0.8cm] (a) at (-0.5,-0.8) {};
\node[rectangle, cyan, draw, minimum size=0.9cm] (ra) at (-0.5,-0.8) {};
\node at (-0.5,-0.8) {a};

\node[rectangle, draw=red, pattern=north west lines, preaction={fill=gray!100}, minimum size=0.9cm] (rb) at (0.4,-2.2) {};
\node[circle, cyan, draw, minimum size=0.8cm] (b) at (0.4,-2.2) {};
\node at (0.4,-2.2) {b};

\node[rectangle, draw=red, pattern=north west lines, preaction={fill=gray!100}, minimum size=0.8cm] (rs11) at (-1.5,0.65) {};
\node[rectangle, draw=red, pattern=north west lines, minimum size=0.8cm] (rs12) at (0.5,0.65) {};
\node[rectangle, draw=cyan, preaction={fill=gray!100}, minimum size=0.8cm] (rs2) at (2,-1) {};

\draw (rs11) -- (ra);
\draw (rs12) -- (ra);
\draw (ra) -- (rb);
\draw (rb) -- (rs2);

\node[right, font=\large] at (-0.8,3) {\tikz\draw[radius=0.2cm] circle; \tikz\draw (0,0) rectangle (0.4cm,0.4cm); (\tikz\draw[radius=0.2cm, preaction={fill=gray!100}] circle; \tikz\draw[preaction={fill=gray!100}] (0,0) rectangle (0.4cm,0.4cm);) Control (Treated) units};
\node[right, font=\large] at (-0.8,2.4) {\tikz\draw[pattern=north west lines, radius=0.2cm] circle; (\tikz\draw[pattern=north west lines] (0,0) rectangle (0.4cm,0.4cm);) Active focal (randomization) units};

\end{tikzpicture}
}
\resizebox{0.48\textwidth}{!}{
\begin{tikzpicture}

\node[above] at (-3.7,0.05) {rand.\ $Z^{(r)}$};
\draw [->] (-4.5,0) -- (-3,0);

\draw[red, dashed] (-0.5,-0.75) ellipse (2cm and 2.63cm);
\draw[cyan, dashed] (0.6,-1.4) ellipse (2.3cm and 1.5cm);

\node[red] at (-2,1.5) {$\mS_1$};
\node[cyan] at (3.2,-1) {$\mS_2$};

\node[circle, draw=red, pattern=north west lines, minimum size=0.8cm] (a) at (-0.5,-0.8) {};
\node[rectangle, cyan, draw, minimum size=0.9cm] (ra) at (-0.5,-0.8) {};
\node at (-0.5,-0.8) {a};

\node[rectangle, draw=red, pattern=north west lines, preaction={fill=gray!100}, minimum size=0.9cm] (rb) at (0.4,-2.2) {};
\node[circle, cyan, draw, minimum size=0.8cm] (b) at (0.4,-2.2) {};
\node at (0.4,-2.2) {b};

\node[rectangle, draw=red, pattern=north west lines, minimum size=0.8cm] (rs11) at (-1.5,0.65) {};
\node at (-1.5,0.65) {r11};
\node[rectangle, draw=red, pattern=north west lines, preaction={fill=gray!100}, minimum size=0.8cm] (rs12) at (0.5,0.65) {};
\node at (0.5,0.65) {r12};
\node[rectangle, draw=cyan, preaction={fill=gray!100}, minimum size=0.8cm] (rs2) at (2,-1) {};

\draw (rs11) -- (ra);
\draw (rs12) -- (ra);
\draw (ra) -- (rb);
\draw (rb) -- (rs2);

\node[right, font=\large] at (-0.8,3) {\tikz\draw[radius=0.2cm] circle; \tikz\draw (0,0) rectangle (0.4cm,0.4cm); (\tikz\draw[radius=0.2cm, preaction={fill=gray!100}] circle; \tikz\draw[preaction={fill=gray!100}] (0,0) rectangle (0.4cm,0.4cm);) Control (Treated) units};
\node[right, font=\large] at (-0.8,2.4) {\tikz\draw[pattern=north west lines, radius=0.2cm] circle; (\tikz\draw[pattern=north west lines] (0,0) rectangle (0.4cm,0.4cm);) Active focal (randomization) units};

\end{tikzpicture}
}
\caption{Illustration of testing $H_{0k}: y_i(0,1) \ge y_i(0,2)$ using a generalized module set.
{\em Left panel:} Continuation of Figure~\ref{fig:ovlap_module} under $\Zobs$. Shaded color indicates treated units, and hatched nodes represent active focal/randomization units for the null hypothesis $H_{0k}$. Only $\mS_1$ is active under $\Zobs$.
{\em Right panel:} A possible randomized treatment $Z^{(r)}$.}
\label{fig:ovlap_module_rand}
\end{figure}

The validity of the modified test for $H_{0k}$ can be easily proved by extending the proof of Theorem~\ref{thm:focaltest_mono}. The potential benefit of this modification is the flexibility when the network is too dense to construct sufficiently many disjoint modules. In that case, a generalized module set better utilizes the network information to construct a test, which may in turn yield a more powerful test.

\subsection{Implementation details in the Medellín example}\label{appdx:implement_medellin}
In the Medellín example in Section~\ref{sec:simu} and~\ref{sec:realMed}, there are far more units that will always stay in control (the $37,055 - 967 = 36,088$ non-hotspots) compared to those with positive treatment probabilities (the $967$ hotspots). Using this fact, Algorithm~\ref{algo:focalrand_single} and~\ref{algo:focalrand_multiple} can be implemented with great simplification. 
Firstly, the edge set $E$ can be defined as $\{(i,j) \in [N]^2: d(i,j) \leq r,~\text{and at least one of}~i,j~\text{is hotspot}\}$. 
Secondly, in constructing the module sets, we can choose all $\ef(\mS)$ to be non-hotspots and all $\er(\mS)$ to be hotspots only. The requirement $\er(\mS_\ell) \cap \er(\mS_{\ell'}) = \emptyset$ is then equivalent to requiring that eligible focal units in different modules do not share the same hotspot within a distance $r$. 
Constructing uniform modules is also straightforward in regions where hotspots are relatively rare compared to non-hotspots, such as the outskirts of the city as shown in Figure~\ref{fig:Med-histdeg}, in which case all the nearby non-hotspots will be in the same $\ef(\mS)$.

Another practical consideration of the algorithm is how to construct modules to have meaningful randomizations. The randomization in Lines 10-11 of Algorithm~\ref{algo:focalrand_single} essentially draws exposures from $\{\expov_k, \expov_{k+1}\}$ on units in $\af(\Zobs; \mS_\ell)$. Ideally, the induced distribution on exposures of units in $\af(\Zobs; \mS_\ell)$ should not be degenerate, i.e., $p_\ell \in (0,1)$ for all $\ell$. 
We could enforce this condition in the construction of module sets by selecting each $\ef(\mS_\ell)$ from the set $\mathfrak S_k := \{i: \expof_i^{\text{all}} \geq \expov_{k+1},~\expof_i^{\text{none}} \leq \expov_k\}$, where $\expof_i^{\text{all}}$ and $\expof_i^{\text{none}}$ are the exposures of $i$ when all and none of $i$'s neighboring hotspots are treated, with the exception that units $j$ with $p_j = 0$ ($p_j=1$) always stay in control (treated), and conditional on observed treatments of units in $\mathbb S_{<k}$. Then by a continuity argument we know $p_\ell \in (0,1)$. Such a construction can avoid, for example, the inclusion of a unit $i$ whose $\nei_i \subseteq \mathbb S_{<k}$ into active focal units when testing $H_{0k}$, on which the conditional distribution of exposure is always degenerate since we condition on $\Zobs_{\mathbb S_{<k}}$.

The histogram of the number of eligible and active focal units when running the algorithm on multiple independent constructions of module sets is presented in Figure~\ref{fig:hist-fu} in the main text. Many eligible focal units are utilized in the randomization test by being active. Notably, although not displayed, none of the active focal units has a degenerate randomization distribution due to the construction above.

\section{A Review of the Biclique Test}\label{appdx:biclique_test}
As mentioned in the main text, the key idea behind the biclique test of~\cite{puelz2022graph} is to translate the conditioning step in the construction of a conditional randomization test into a biclique decomposition of an appropriate graphical representation of the null hypothesis $\widetilde H_{0k}$. To be specific, we need the following definitions. Denote $\mathbb U = [N]$ and $\mathbb Z = \{z \in \{0,1\}^N: P(z) > 0\}$. 
\vspace{-5px}
\newcommand{\NE}{\mathrm{ne}}
\begin{definition}[Null exposure graph and bicliques.]
The null exposure graph with respect to $\widetilde H_{0k}$ is a bipartite graph 
$\mathcal G_k^\NE =(V_k^\NE,E_k^\NE)$ such that $V_k^\NE = \mathbb U \cup \mathbb Z$, and an edge between $(i,z) \in \mathbb U\times \mathbb Z$ exists in $E_k^\NE$ if and only if $z_i=0$ and $\expof_i(z) \in \{\expov_k, \expov_{k+1}\}$.
A biclique of a null exposure graph $\mathcal G_k^\NE =(V_k^\NE,E_k^\NE)$ is a subgraph $\mathcal C = (\mathcal U, \mathcal Z)$, where $\mathcal U \subseteq \mathbb U$ and $\mathcal Z \subseteq \mathbb Z$, such that each unit $i \in \mathcal U$ is connected to all assignments in $\mathcal Z$, i.e., $\mathcal U \times \mathcal Z \subseteq E_k^\NE$.
\end{definition}
\vspace{-5px}

The null exposure graph encodes the ``imputability pattern" under the null hypothesis. 
That is, if $z$ and $z'$ are both connected to a unit $i$ in the null exposure graph, then 
$Y_i(z) = Y_i(z')$ under the null hypothesis.
As a result, the potential outcomes involved in $\widetilde H_{0k}$ for all units in a biclique can be imputed under the null hypothesis by their observed outcomes for all treatment assignments in the biclique if the observed $Z^{\mathrm{obs}}$ is also in the biclique.

In light of these imputability results, the biclique test of~\cite{puelz2022graph} proceeds in three main steps presented in Algorithm~\ref{algo:cliquetest}.
First, they build the null exposure graph that uniquely encodes the null hypothesis being tested and the particular treatment exposure function $w_i(\cdot)$ under the design $P(\cdot)$.
Next, they compute a biclique decomposition of the null exposure graph, a collection of bicliques $\{\mathcal C_j\}_{j\in J}$ with $\mathcal C_j = (\mathcal U_j, \mathcal Z_j)$, such that 
$\{\mathcal Z_j\}_{j\in J}$ forms a partition\footnote{That is, $\mathcal Z_i \cap \mathcal Z_j = \emptyset$ for $i\neq j$, and $\bigcup_{j \in J} \mathcal Z_j = \mathbb Z$.} of $\mathbb Z$. Such a decomposition can be implemented efficiently using several existing graph algorithms\footnote{For example, the ``binary inclusion-maximal biclustering" (\texttt{Bimax}) algorithm used in~\cite{puelz2022graph}, and the ``\texttt{iMBEA}" algorithm of~\cite{zhang2014imbea}.}.
Finally, a conditional randomization test is executed within the biclique $\mathcal C^\obs = (\mathcal U^\obs, \mathcal Z^\obs)$, one of the $\mathcal C_j$'s such that $\Zobs \in \mathcal Z^\obs$, which reduces to a weighted sampling over the treatment assignments in $\mathcal Z^\obs$ using $P(z)$ as the weights. $\mU^\obs$ is called the ``focal units" and $\mZ^\obs$ is called the ``focal assignments".

Theorem~2 of~\cite{puelz2022graph} establishes the finite-sample validity of the test for $\widetilde H_{0k}$. 
Moreover, using a similar reasoning as in Theorem \ref{thm:focaltest_mono} of our paper, when using an exposure-monotone test statistic in the order $(\expov_k, \expov_{k+1})$, the biclique test is also valid for testing the single monotone hypothesis $H_{0k}:y_i(0, \expov_k) \ge y_i(0, \expov_{k+1})~\forall i$.
\begin{algorithm}[t!]
\caption{\citep{puelz2022graph} Biclique test for $\widetilde H_{0k}: y_i(0, \expov_k) = y_i(0, \expov_{k+1})$}
\label{algo:cliquetest}
\begin{algorithmic}[1]
\REQUIRE Observed treatment $\Zobs$; design $P(\cdot)$ (Input).

\hspace{-24px} {\bf Output:} Finite-sample valid $p$-value for $\widetilde H_{0k}$.

\STATE Construct the null exposure graph $\mathcal G_k^\NE$ using $\mathbb U = [N]$ and $\mathbb Z$. 
\STATE Decompose $\mathcal G_k^\NE$ to obtain $\{\mathcal C_j\}_{j\in J}$, and find the unique biclique $\mathcal C^\obs = (\mathcal U^\obs, \mathcal Z^\obs) \in \{\mathcal C_j\}_{j\in J}$ such that $Z^{\mathrm{obs}} \in \mathcal Z^\obs$.
\STATE Define the randomization distribution $r(Z) \propto \mathbbm 1\{Z \in \mathcal Z^\obs\} \cdot P(Z)$.
\STATE Calculate the $p$-value as follows:
\begin{equation}\label{eq:pval_clique}
    \pval(Z^{\mathrm{obs}}, Y^{\mathrm{obs}}; \mathcal C^\obs) = \mathbbm E_{Z \sim r(\cdot)} \left[\mathbbm 1\{ t(Z, Y^{\mathrm{obs}}; \mathcal C^\obs) > t(Z^{\mathrm{obs}}, Y^{\mathrm{obs}}; \mathcal C^\obs) \} \right].
\end{equation}
\end{algorithmic}
\end{algorithm}

\begin{proposition}\label{prop:bicliquert_mono}
Consider applying Algorithm \ref{algo:cliquetest} to test $\widetilde H_{0k}$.
When an exposure-monotone test statistic in the order $(\expov_k, \expov_{k+1})$ is applied in Line 4, the resulting $p$-value in~\eqref{eq:pval_clique} is also valid under $H_{0k}$ in Equation~\eqref{eq:null_mono_k}. That is, for any $\alpha \in [0,1]$,
\[
\mathbbm P_{Z^{\mathrm{obs}} \sim P(\cdot)} \left(\pval(Z^{\mathrm{obs}}, Y^{\mathrm{obs}}; \mathcal C^\obs) \leq \alpha ~\big|~ \{\mathcal C_j\}_{j\in J}, H_{0k} \right) \leq \alpha.
\]
\end{proposition}

\section{Network Splitting and the Clique-based Algorithm}\label{appdx:assignment}
In this section we present a way to partition the network using an algorithm in the community detection literature, and subsequently find matches between hypotheses to test and the partitions leveraging the clique-based test in Section~\ref{sec:method_clique}. We also present power simulation results under the setup in Section~\ref{sec:simu}.

\subsection{Network partitioning}
As discussed in the main text, heuristically we would want the partition of the network to maximize connections between nodes within each sub-network, while minimizing the connections between nodes across different sub-networks. The Leiden algorithm~\citep{traag2019leiden} is considered a computationally efficient algorithm to detect such sub-networks (or ``communities", ``clusters") with some theoretical guarantee of the well-connectedness of the resulting sub-networks. When applying the algorithm iteratively, it is also guaranteed to converge to a partition in which all subsets of all sub-networks are locally optimal. The algorithm is widely adopted in various domains such as genetics~\citep{heumos2023best} and social science~\citep{logan2023social}.

To use the Leiden algorithm there are three tuning parameters to specify: ``resolution", ``beta", and the number of iterations. Higher resolutions lead to more and smaller communities, while lower resolutions lead to fewer and larger communities, and ``beta" affects the randomness in the algorithm.
In the Medellín network, we specify the number of iterations to be $200$, as we observed that it's enough for the algorithm to stabilize. 
For some values of resolution and beta, the average number of communities detected across $5,000$ repetitions of the algorithm are presented in Table \ref{tab:Leiden_number}. 
From the results, in the subsequent analysis we choose ``resolution" in $\{10^{-3}, 10^{-4}\}$ and ``beta" in $\{10^{-1}, 10^{-2}, 10^{-3}\}$ to avoid too many small communities. 
\begin{table}[!ht]
\centering
\begin{tabular}{ccc}
\hline
Resolution & Beta & \# of communities detected \\
\hline
$10^{-2}$ & $10^{-1}$ & $1841.00$ \\
$10^{-3}$ & $10^{-1}$ & $107.50$ \\
$10^{-4}$ & $10^{-1}$ & $54.77$ \\
$10^{-2}$ & $10^{-2}$ & $1826.92$ \\
$10^{-3}$ & $10^{-2}$ & $107.02$ \\
$10^{-4}$ & $10^{-2}$ & $54.79$ \\
$10^{-2}$ & $10^{-3}$ & $1824.32$ \\
$10^{-3}$ & $10^{-3}$ & $105.47$ \\
$10^{-4}$ & $10^{-3}$ & $54.70$ \\
\hline
\end{tabular}
\caption{Leiden algorithm: resolution and beta and number of communities, averaged across $5,000$ repetitions of the algorithm.}
\label{tab:Leiden_number}
\end{table}

\subsection{Assigning communities to hypotheses}\label{appdx:cls_assignment}
Given the detected communities, or sub-networks $\mathcal G_c = (V_c, E_c)$ with $c \in [C]$, we need to decide which sub-networks are used to test which of the $K-1$ hypotheses $H_{0k}$. We will formulate such a decision problem into an optimization problem that can be solved efficiently.

Denote $S_c = |V_c|$ the size of the community $c$, $\mathrm{NE}_{c,k}$ the matrix representation of the null exposure graph built within $\mathcal G_c$ under hypothesis $H_{0k}$, with $\mathrm{dim}(\mathrm{NE}_{c,k}) = S_c \times N_{\text{rand}}$ where $N_{\text{rand}}$ is the number of randomizations drawn from $P_k(\cdot)$, and the $(i,j)$-th entry $\mathrm{NE}_{c,k}[i,j] = 1$ if and only if the $i$-th node is connected to the $j$-th randomization. Let $M_{c,k}$ be a measurement of the ``informativeness" of $\mathrm{NE}_{c,k}$, such as
\begin{itemize}
    \item density of $\mathrm{NE}_{c,k}$: $M_{c,k} = \sum_{i,j} \mathrm{NE}_{c,k}[i,j] / (S_c \cdot N_{\text{rand}})$;
    \item average standard deviations across rows: $M_{c,k} = \sum_{i} \mathrm{sd}(\mathrm{NE}_{c,k}[i,]) / S_c$;
    \item average standard deviations across columns: $M_{c,k} = \sum_{j} \mathrm{sd}(\mathrm{NE}_{c,k}[,j]) / N_{\text{rand}}$.
\end{itemize}
Denote $A_{c, k}$ be the indicator of community $c$ being assigned to test $H_{0k}$. We propose to solve the assignment $A = (A_{c, k})$ through the following integer programming problem:
\begin{equation}\label{eq:assign_ip}
\begin{split}
\max_{A}~ &\min_{k \in [K-1]} \quad \sum_{c\in[C]} A_{c, k} M_{c, k} S_c \\
\text{s.t.} \quad & \sum_{k \in [K-1]} A_{c,k} = 1, ~\forall c \\
            \quad & \sum_{c \in [C]} A_{c,k} \geq 1, ~\forall k \\
            \quad & A_{c, k} \in \{0,1\},~\forall c, k,
\end{split}
\end{equation}
or equivalently the following MILP:
\begin{equation}\label{eq:assign_milp}
\begin{split}
&\max_{A, ~t} ~ t \\
\text{s.t.} \quad & \sum_{c\in[C]} A_{c, k} M_{c, k} S_c \geq t, ~\forall k\\
            \quad & \sum_{k \in [K-1]} A_{c,k} = 1, ~\forall c \\
            \quad & \sum_{c \in [C]} A_{c,k} \geq 1, ~\forall k \\
            \quad & A_{c, k} \in \{0,1\},~\forall c, k.
\end{split}
\end{equation}
Without prior information, we treat each hypothesis equally by maximizing the minimum of a weighted score of the informativeness in testing hypothesis $k$. For example, when $M_{c,k}$ is the density of the null exposure graph, $M_{c,k} S_c$ calculates the sum of the NE graph densities for testing $H_{0k}$ weighted by the size $S_c$, or equivalently the number of connections in the NE graphs that are used to test $H_{0k}$ divided by $N_{\text{rand}}$. The constraint $\sum_{k} A_{c,k} = 1$ imposes that each community $c$ is used to test exactly one $H_{0k}$, and the constraint $\sum_{c} A_{c,k} \geq 1$ imposes that there should be at least one community assigned to test $H_{0k}$.
Both \eqref{eq:assign_ip} and \eqref{eq:assign_milp} can be solved efficiently using optimization packages. An approximate solution is sufficient for our purpose.

Figure~\ref{fig:Leiden_para1_rep1} and~\ref{fig:Leiden_para2} display some realizations of the network splitting and the focal units from solving the above optimization problem. For each of them, the left panel displays the network partitioning output from the Leiden algorithm, the middle panel displays the assignment result from solving the MILP, and the right panel displays the focal units in the biclique that contains $\Zobs$, i.e., $\mU^\obs$.
As a baseline, we also present a naive way to split the network in Figure~\ref{fig:prev_split} where the two break points in the $x$-coordinate are the $60\%$ and $70\%$ quantiles of the $x$-coordinates of the $967$ hotspots, and the break point in the $y$-coordinate is the median of the $y$-coordinates of the $967$ hotspots. Generally for all realizations there are fewer focal units for $H_{02}$ and $H_{03}$. 

\begin{figure}[!h]
    \centering
    \includegraphics[width = 0.32\textwidth]{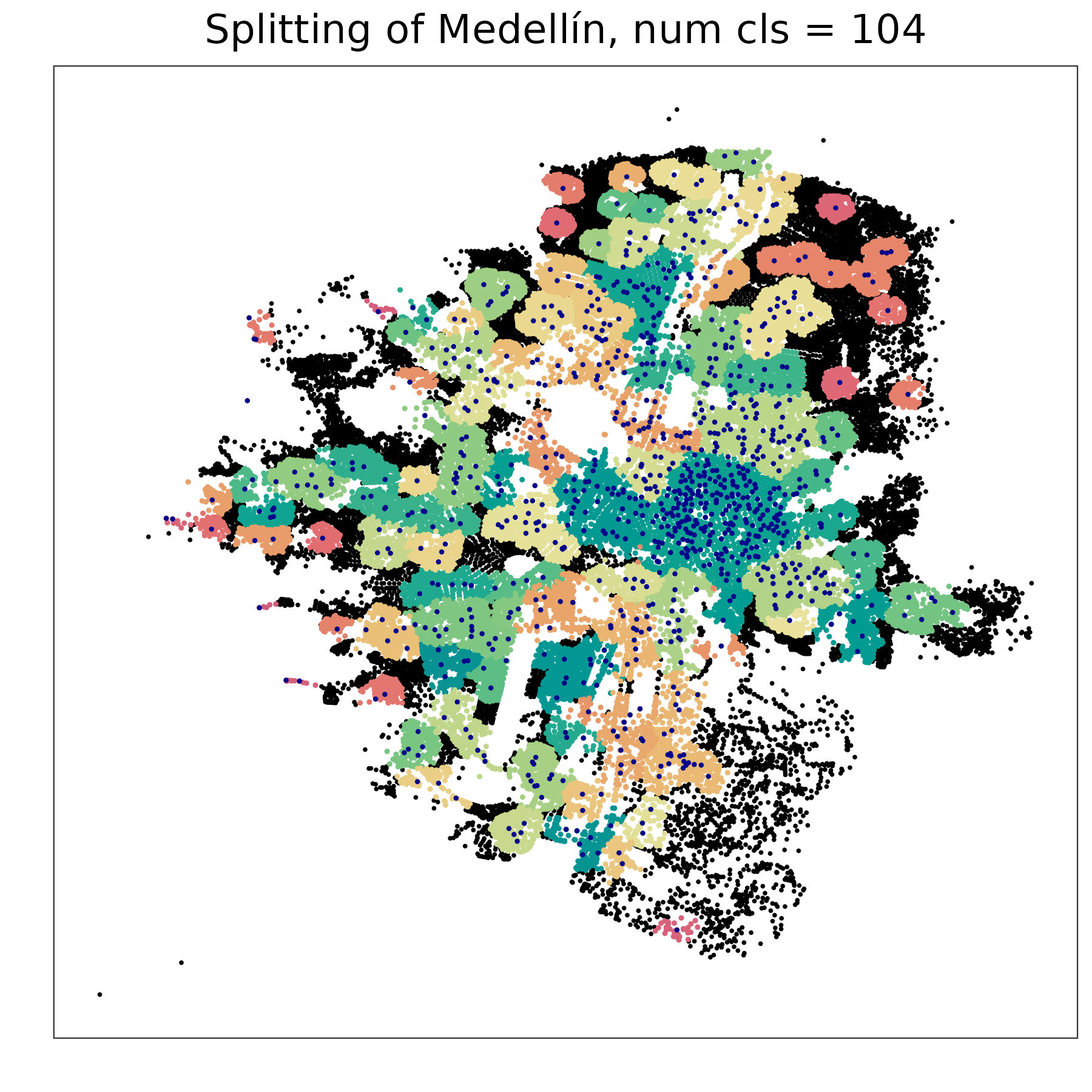}
    \includegraphics[width = 0.32\textwidth]{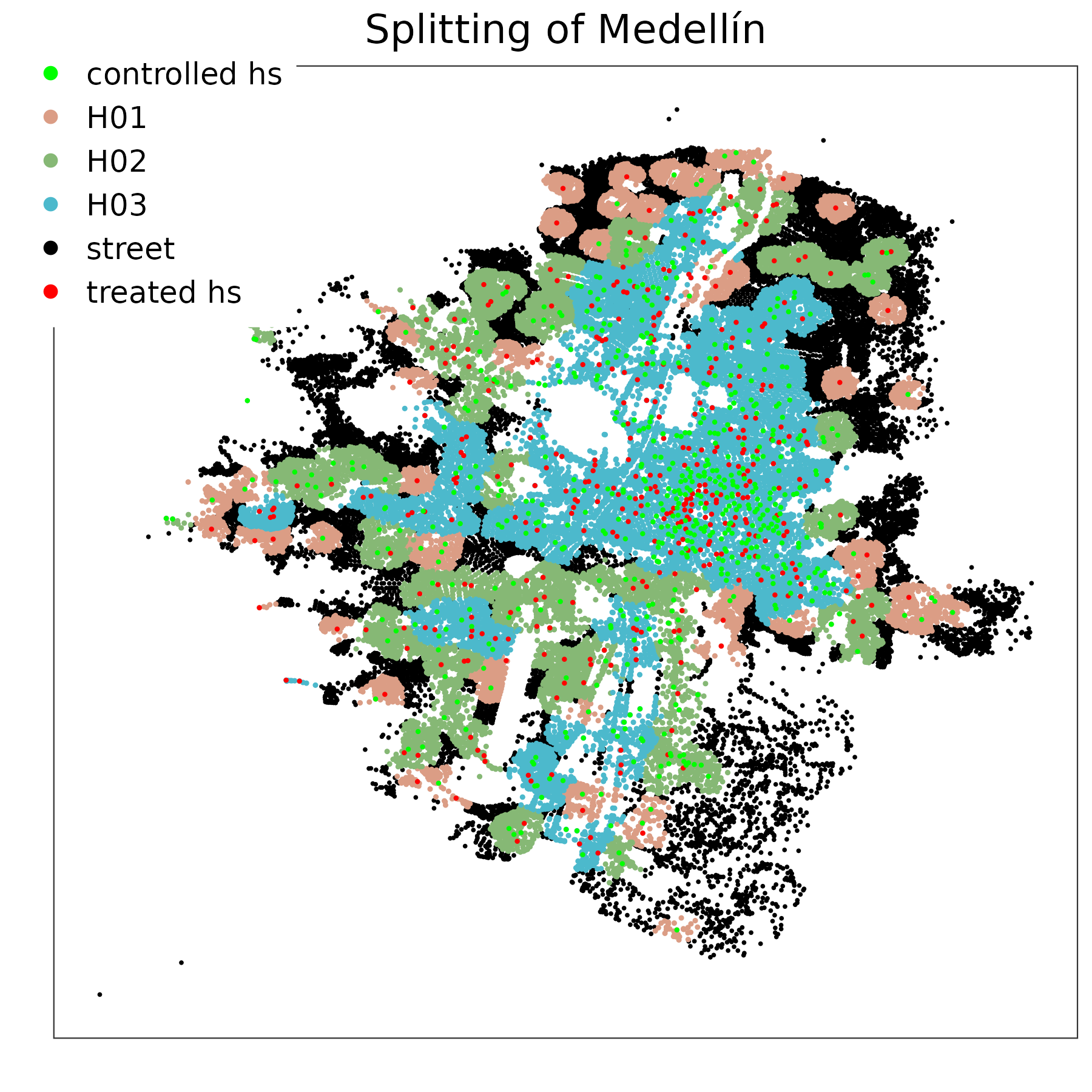}
    \includegraphics[width = 0.32\textwidth]{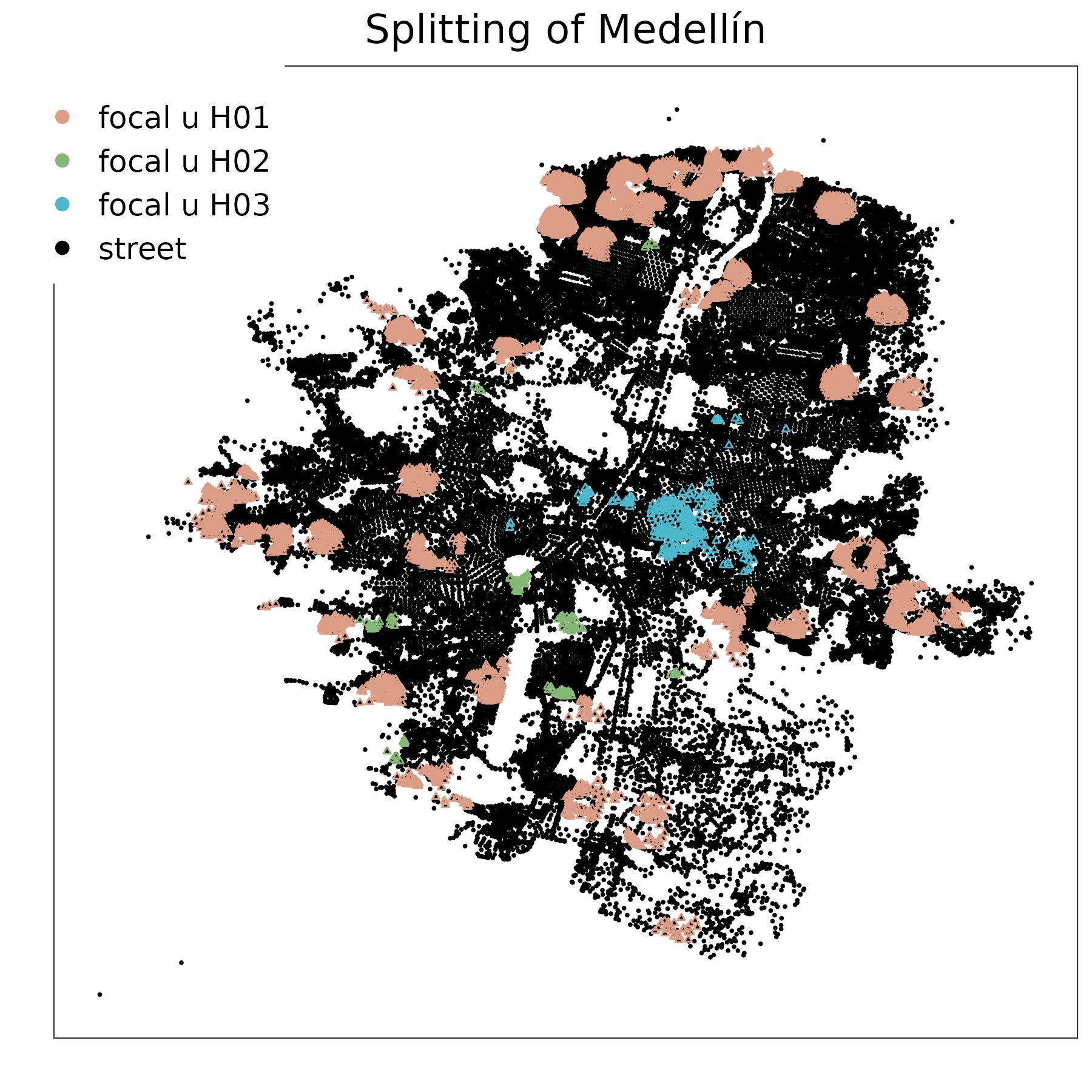}
    \caption{Leiden Splitting (Res, Beta) = $(10^{-3}, 10^{-1})$ results from one realization. Left: the network partitioning output from the Leiden algorithm. Middle: the assignment result from solving the MILP. ``H01", ``H02" and ``H03" denote the sub-networks for testing $H_{01}$, $H_{02}$ and $H_{03}$ respectively. Right: the focal units in the biclique that contains $\Zobs$, for testing $H_{01}$, $H_{02}$ and $H_{03}$ respectively.}
    \label{fig:Leiden_para1_rep1}
\end{figure}

\begin{figure}[!h]
    \centering
    \includegraphics[width = 0.32\textwidth]{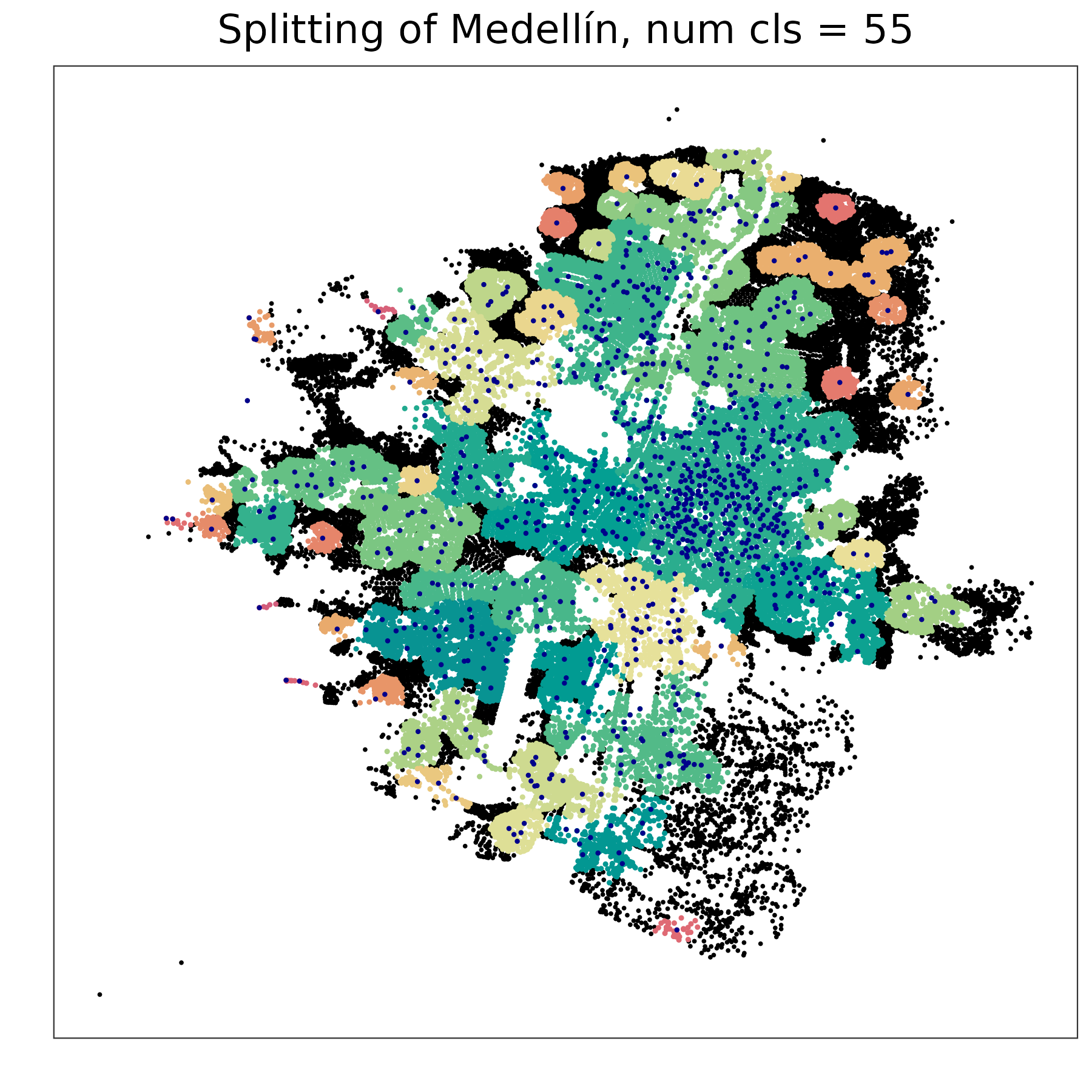}
    \includegraphics[width = 0.32\textwidth]{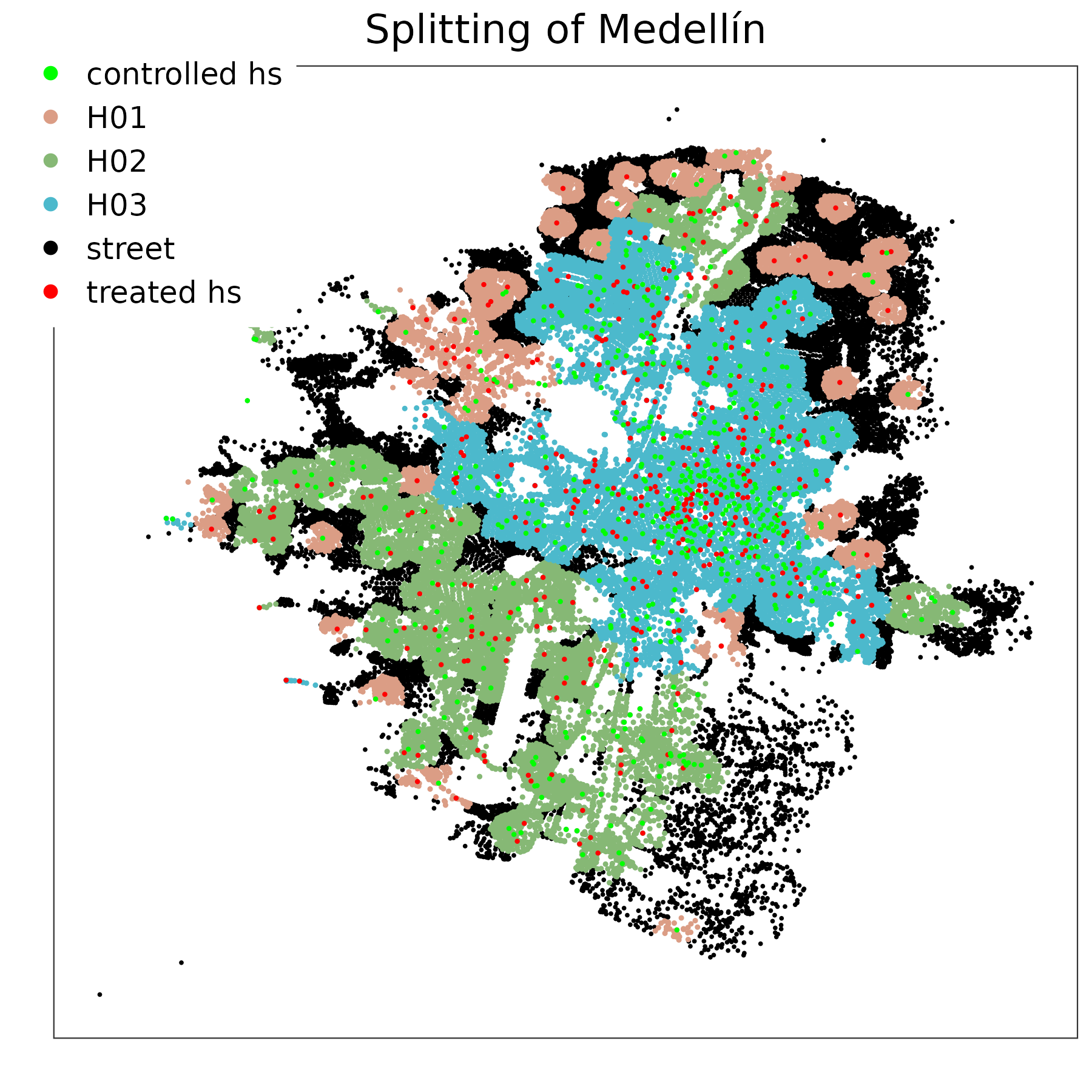}
    \includegraphics[width = 0.32\textwidth]{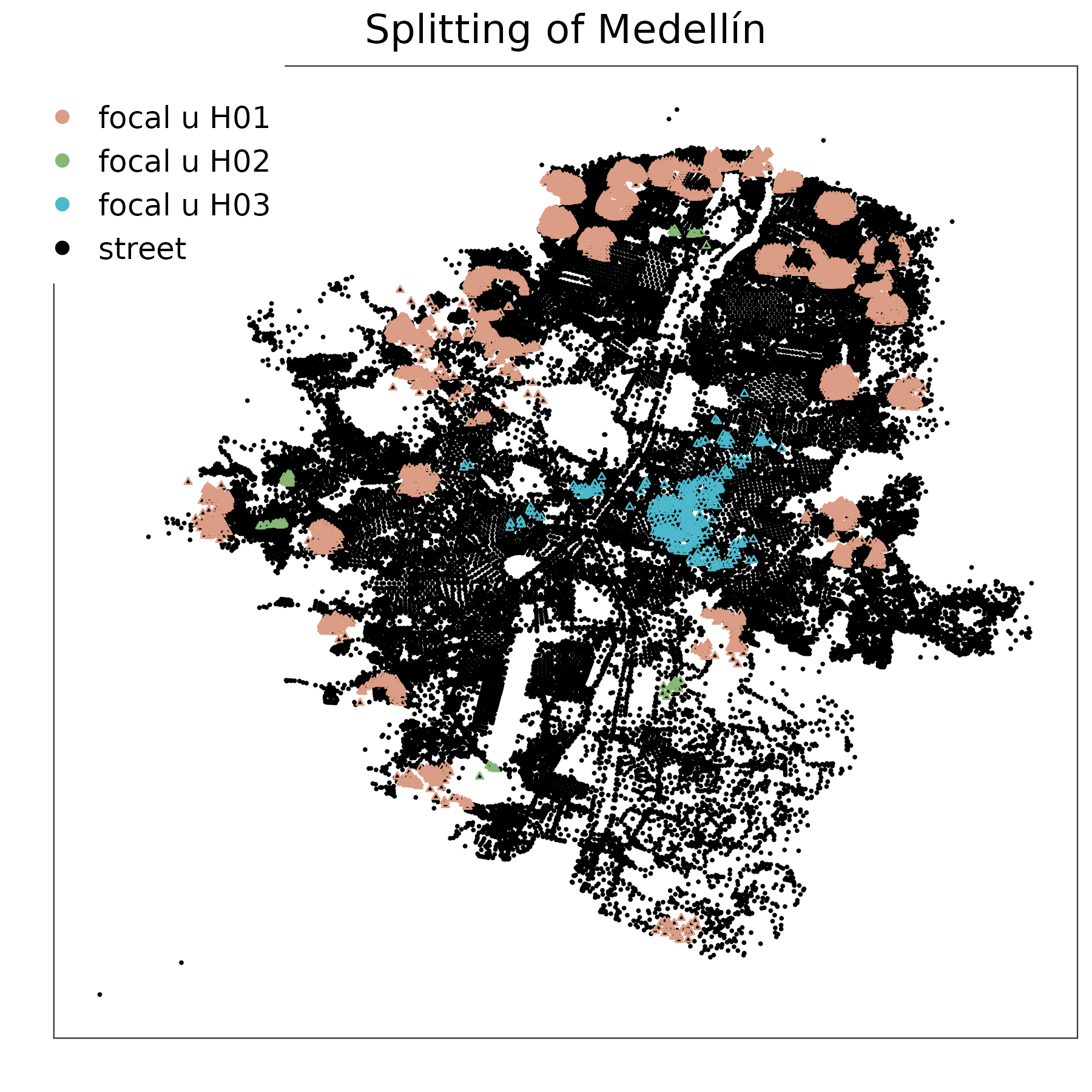}
    \caption{Leiden Splitting (Res, Beta) = $(10^{-4}, 10^{-1})$ results from one realization. Left: the network partitioning output from the Leiden algorithm. Middle: the assignment result from solving the MILP. ``H01", ``H02" and ``H03" denote the sub-networks for testing $H_{01}$, $H_{02}$ and $H_{03}$ respectively. Right: the focal units in the biclique that contains $\Zobs$, for testing $H_{01}$, $H_{02}$ and $H_{03}$ respectively.}
    \label{fig:Leiden_para2}
\end{figure}

\begin{figure}[!h]
    \centering
    \includegraphics[width = 0.4\textwidth]{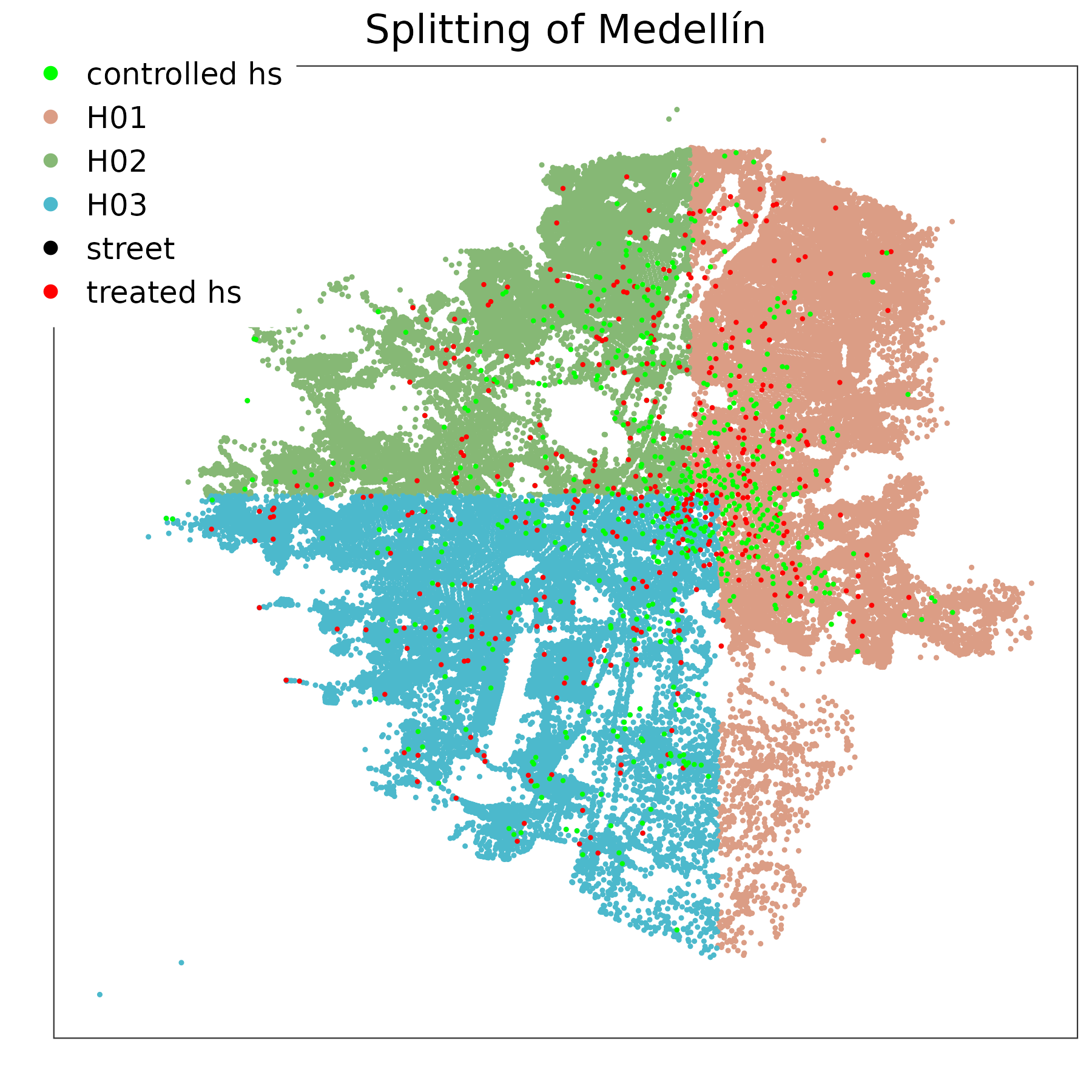}
    \includegraphics[width = 0.4\textwidth]{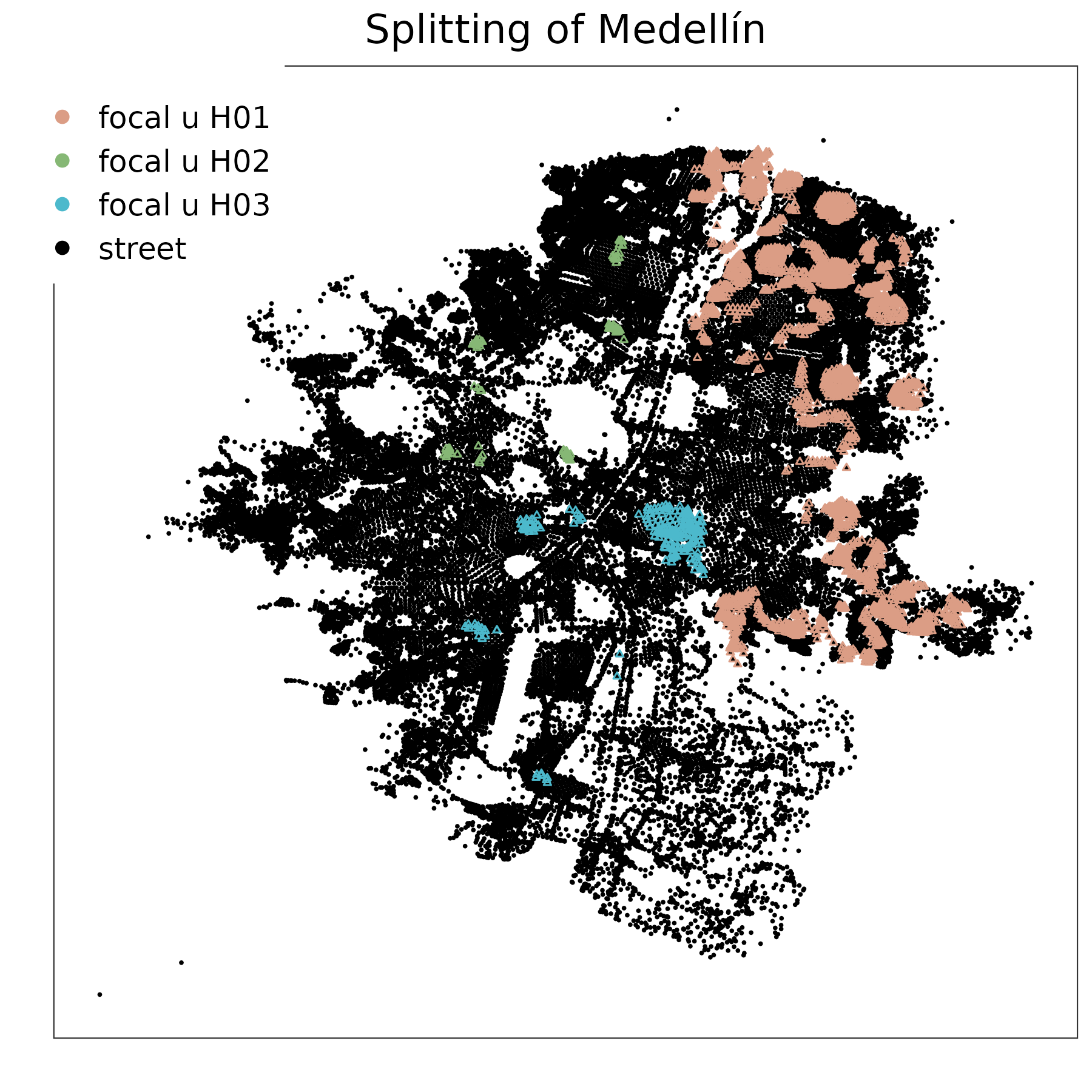}
    \caption{Naive splitting results. Left: the sub-networks for testing $H_{01}$, $H_{02}$ and $H_{03}$, denoted as ``H01", ``H02" and ``H03" respectively. Right: the focal units in the biclique that contains $\Zobs$, for testing $H_{01}$, $H_{02}$ and $H_{03}$ respectively.}
    \label{fig:prev_split}
\end{figure}

\subsection{Power simulation}\label{appdx:assign_power}
We use the DGP1 outcome model in Section \ref{sec:simu_dgp}, and $10,000$ randomizations in constructing the bicliques. The results are in Figure \ref{fig:power_DGP1_Leiden_para1} for Leiden algorithm (Res, Beta) = $(10^{-3}, 10^{-1})$ and Table \ref{tab:Leiden_power} for other parameters where we only present the power at $\tau \in \{-0.5, 0, 0.2, 0.5, 1\}$ with difference-in-means as the test statistic and $p$-values combined by Fisher's rule.

\begin{figure}[!h]
    \centering
    \includegraphics[width = 0.4\textwidth]{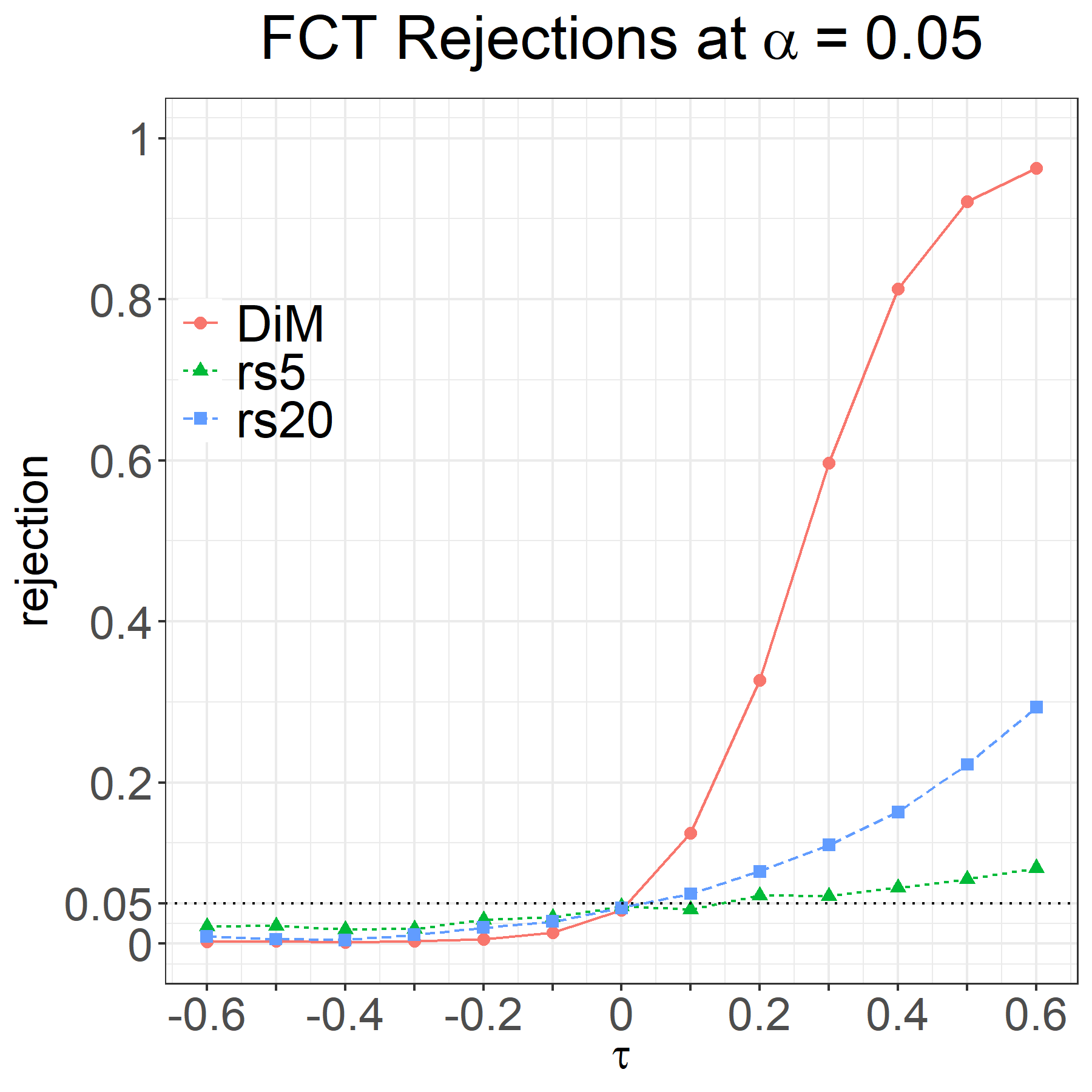}
    \includegraphics[width = 0.4\textwidth]{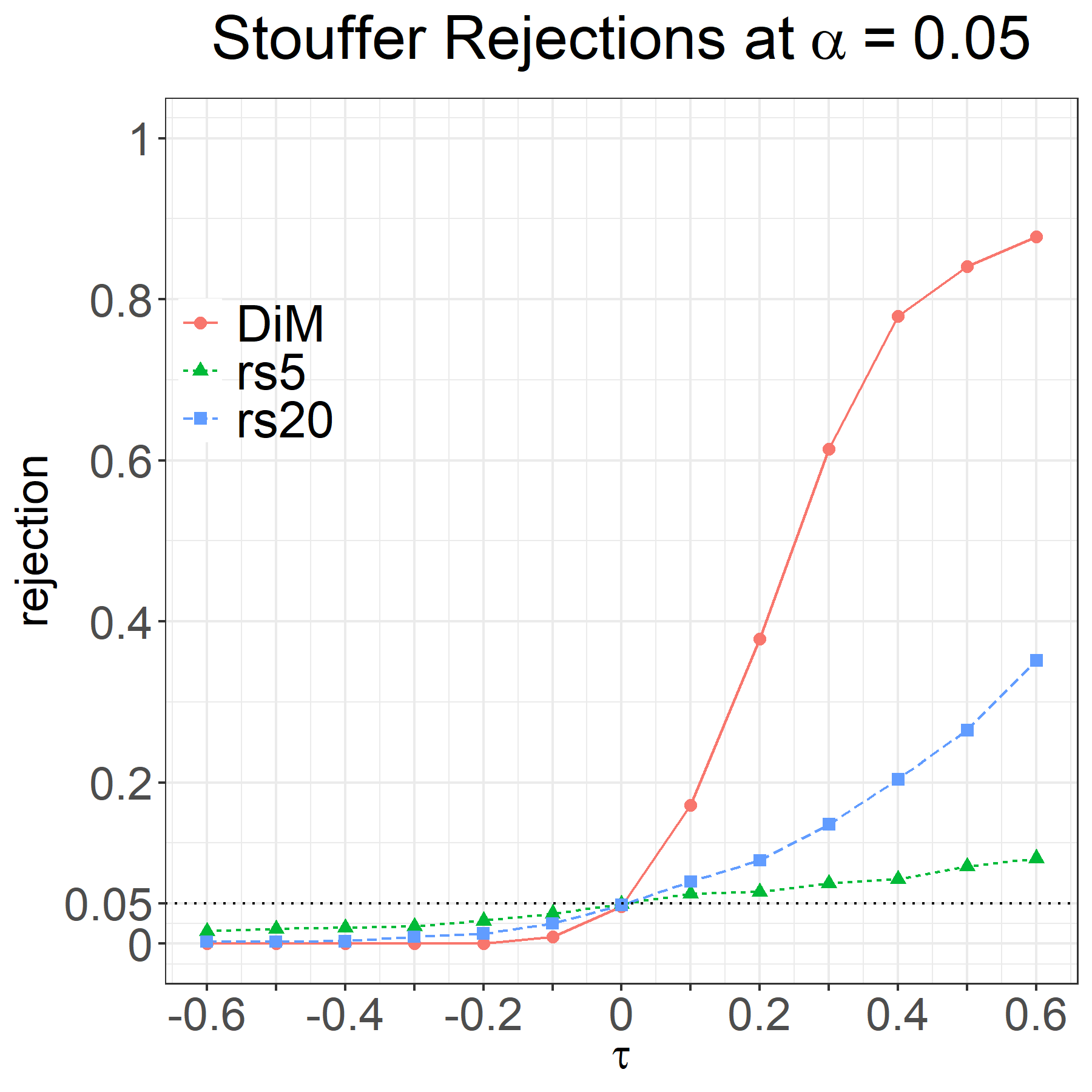}
    \includegraphics[width = 1.0\textwidth]{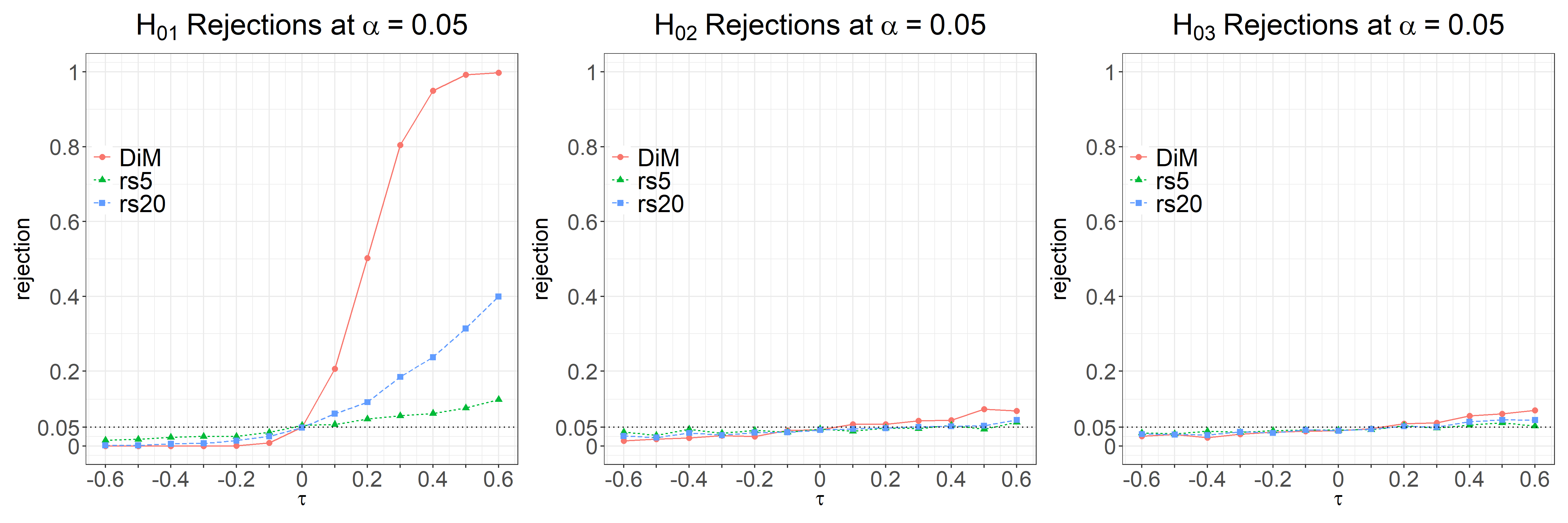}
    \caption{Leiden Splitting (Res, Beta) = $(10^{-3}, 10^{-1})$. Rejection of overall monotone hypothesis and individual hypotheses under DGP1 in Section~\ref{sec:simu_dgp}.}
    \label{fig:power_DGP1_Leiden_para1}
\end{figure}

\begin{table}[!h]
\centering
\begin{tabular}{lccccc}
  \hline
  Method & $\tau = -0.5$ & $\tau = 0$ & $\tau = 0.2$ & $\tau = 0.5$ & $\tau = 1$ \\ 
  \hline
Leiden $(10^{-3}, 10^{-1})$ & 0.31 & 4.19 & 32.74 & 92.16 & 99.12 \\ 
  Leiden $(10^{-4}, 10^{-2})$ & 0.10 & 4.14 & 30.56 & 88.45 & 97.95 \\ 
  Leiden $(10^{-3}, 10^{-3})$ & 0.21 & 3.82 & 31.21 & 91.87 & 98.72 \\ 
  Leiden $(10^{-4}, 10^{-1})$ & 0.42 & 3.38 & 29.96 & 89.22 & 97.65 \\ 
  Leiden $(10^{-3}, 10^{-2})$ & 0.16 & 3.98 & 30.54 & 91.39 & 98.53 \\ 
  Leiden $(10^{-4}, 10^{-3})$ & 0.43 & 3.91 & 29.89 & 90.06 & 98.00 \\ 
  Naive splitting & 0.30 & 3.81 & 22.88 & 77.49 & 96.71 \\ 
  Module-based test & 0.00 & 5.15 & 62.25 & 100.00 & 100.00 \\
   \hline
\end{tabular}
\caption{Simulation results for different splitting configurations under DGP1 in Section~\ref{sec:simu_dgp}. All tests are conducted at the 5\% level and 
the reported values are rejection probabilities (in \%).}
\label{tab:Leiden_power}
\end{table}

In general, there is little difference in power for different Leiden algorithm parameters. All of them are better than the naive way of splitting. Due to the lack of focal units, tests for $H_{02}$ and $H_{03}$ are of low power, making the clique-based test less powerful compared to the module-based test (Algorithm~\ref{algo:focalrand_multiple} of the main text).

\section{More Empirical Results}~\label{appdx:more_empirial}
\subsection{A check of grouping exposure levels}\label{appdx:grouping3}
In this section, we provide an empirical check for the assumption that the (controlled) potential outcomes of a street are the same for all exposure levels of at least three. Specifically, consider the following null hypothesis
\begin{equation}\label{eq:grouping3}
    H_0^{[\geq3]}:\quad y_i(0, \expov) ~=~ y_i(0, \expov'),\quad \forall \expov, \expov' \geq 3,~\forall i,
\end{equation}
where the exposure function is the number of treated neighbors within 225 meters as in the main text.
We consider the following randomization $p$-value to assess~\eqref{eq:grouping3}:
\begin{equation}\label{eq:AIC_grouping}
    \begin{split}
        &\pval(Z^\obs) = \mathbbm P_{Z \sim P(\cdot)} \big( t(Z, \Yobs) \geq t(\Zobs, \Yobs) \big),\\
        &t(Z, Y) = \mathrm{AIC}(\text{model 1}(Z,Y)) - \mathrm{AIC}(\text{model 2}(Z,Y)),
    \end{split}
\end{equation}
where $P(\cdot)$ is the design, and the test statistic is the difference in Akaike Information Criteria (AIC) between the OLS fits of the following two saturated linear models with grouping threshold $3$ (model 1) and $10$ (model 2)\footnote{From Figure~\ref{fig:Med-histdeg} of the main text, a threshold of $10$ is sufficient as there are few units with degree at least $10$.}, both fitted on non-hotspot streets only so that $Z_i=0$ for all $Z \sim P(\cdot)$:
\begin{equation}\label{eq:sat_exps_linear}
    \begin{split}
        &\text{model 1}(Z,Y):\quad Y_i = \sum_{k=0}^2 \theta_k \mathbbm 1\{\expof_i(Z) = k\} + \theta_3 \mathbbm 1\{\expof_i(Z) \geq 3\} + \varepsilon_i,\\
        &\text{model 2}(Z,Y):\quad Y_i = \sum_{k=0}^9 \theta_k \mathbbm 1\{\expof_i(Z) = k\} + \theta_{10} \mathbbm 1\{\expof_i(Z) \geq 10\} + \varepsilon_i.
\end{split}
\end{equation}
Intuitively, under the null $H_0^{[\geq3]}$ the two models are indistinguishable, so $t(\Zobs, \Yobs)$ should not lie at the tails of the randomization distribution --- the distribution of $t(Z, \Yobs)$ induced by $Z\sim P(\cdot)$ --- so that the resulting $p$-value should not be small. When $H_0^{[\geq3]}$ does not hold, however, model 2 is expected to fit the outcome better, so $t(\Zobs, \Yobs)$ is expected to lie in the right tail of the randomization distribution, giving a small $p$-value. The randomization distributions and the $p$-values for the six types of outcomes considered in the main text is presented in Figure~\ref{fig:group3_pval_dist}. We observe that the $p$-values are all moderate for all outcomes, supporting our decision to group exposure levels beyond $3$ into a single category.
\begin{figure}[!h]
    \centering
    \includegraphics[width=0.8\linewidth]{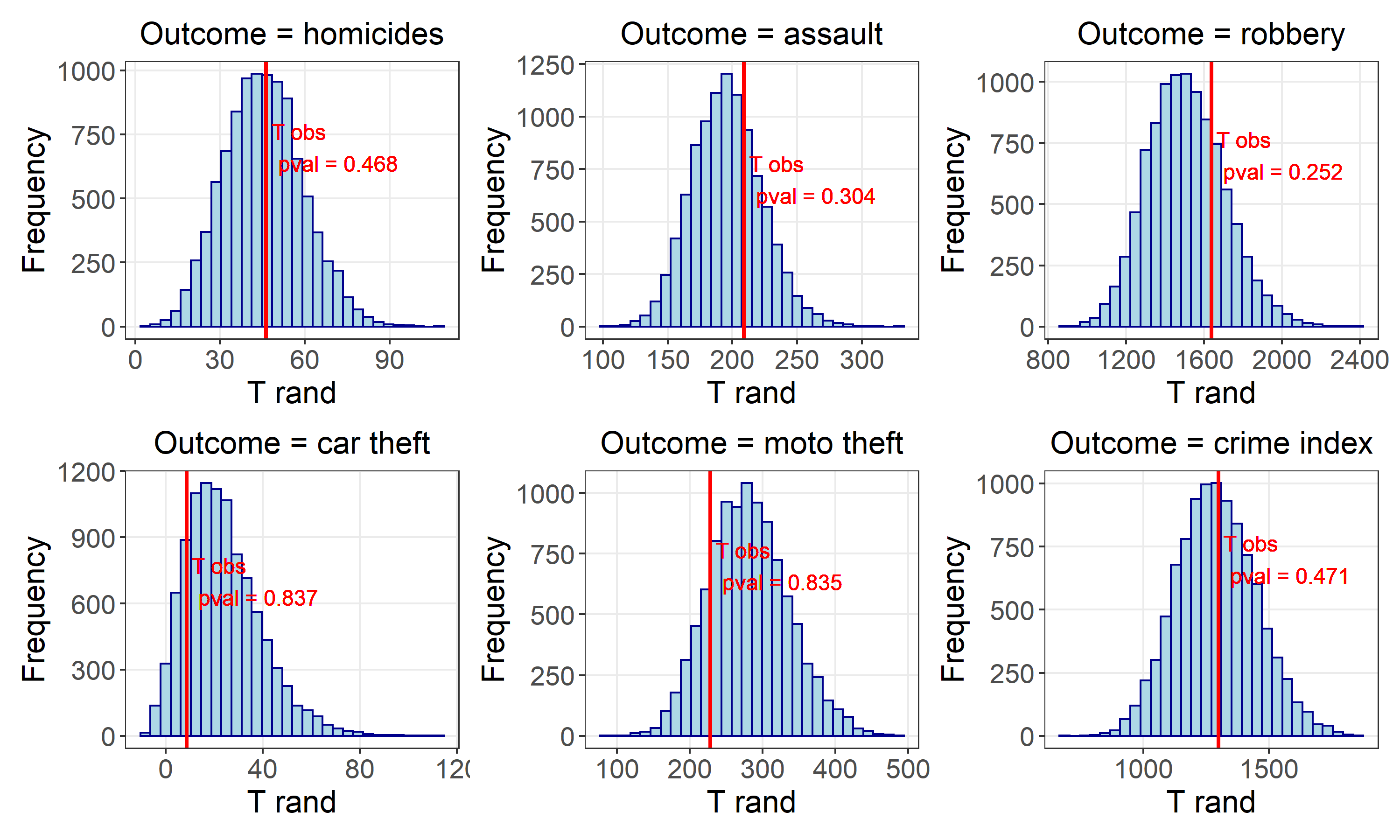}
    \caption{Randomization distributions and $p$-values from~\eqref{eq:AIC_grouping}.}
    \label{fig:group3_pval_dist}
\end{figure}
One final remark is that, instead of the AIC, one can use other (penalized) goodness-of-fit statistics, such as the difference in mean-squared errors or Bayesian Information Criteria, potentially applied to models other than the linear models considered here. 

\subsection{Confounding factors}\label{appdx:degree_medellin}
In the Medellín setting, when a street has many hotspots around, it is expected that the crime level in that street is high because of the proximity to the disturbing streets. Indeed, the two boxplots in Figure~\ref{fig:degree_boxplots} shows that under $\Zobs$, degrees differ systematically across treatment and exposure levels. This suggests that $d_i$ may act as a confounder when regressing the outcome on treatment and exposure, even if one fully saturates at exposure as in~\eqref{eq:sat_exps_linear}, and motivates the DGP3 in Section~\ref{sec:simu_dgp} of the main text.
We also note that many such confounders may exist in real-life settings, making standard inference methods unreliable.
\begin{figure}[!h]
    \centering
    \includegraphics[width = 0.3\textwidth]{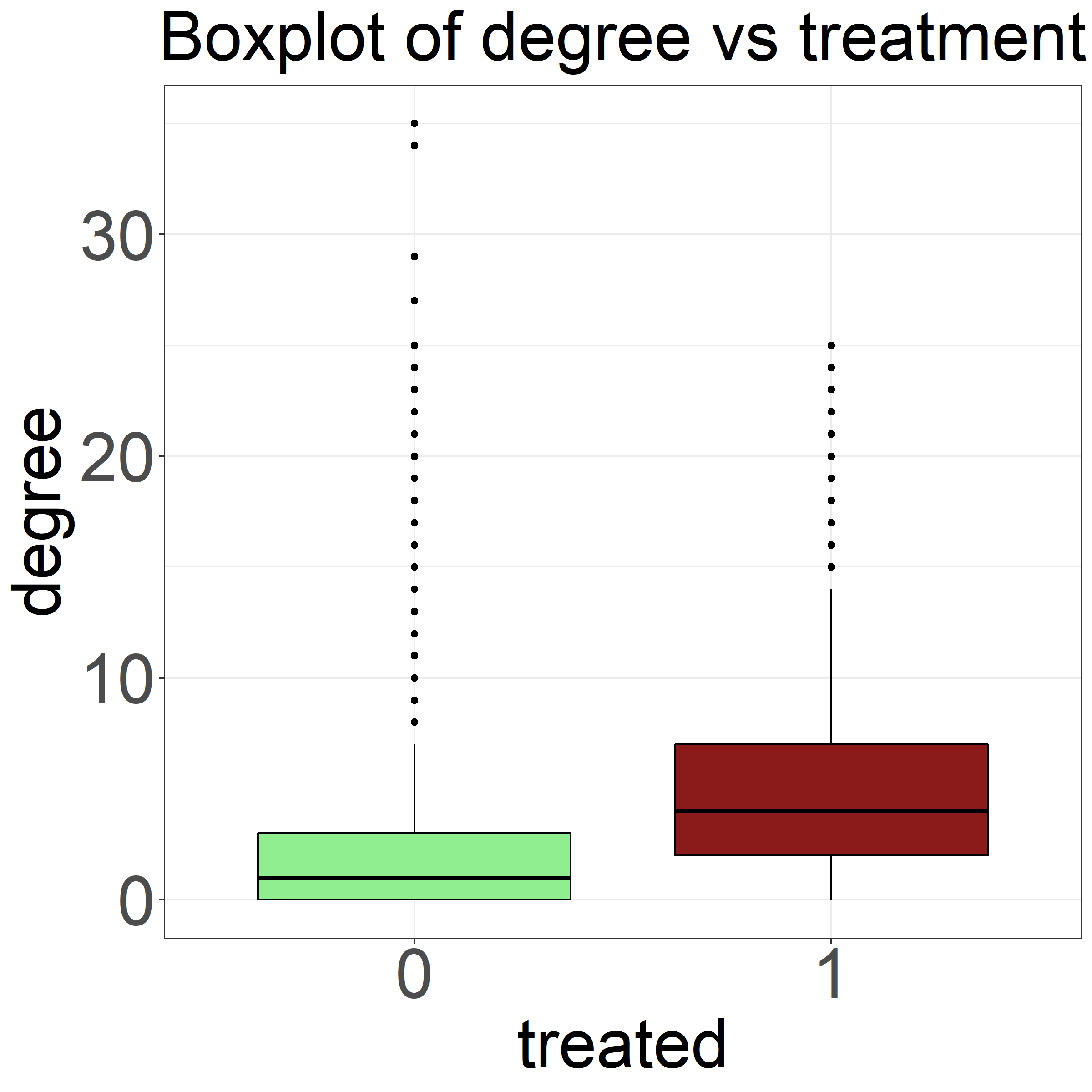}
    \includegraphics[width = 0.3\textwidth]{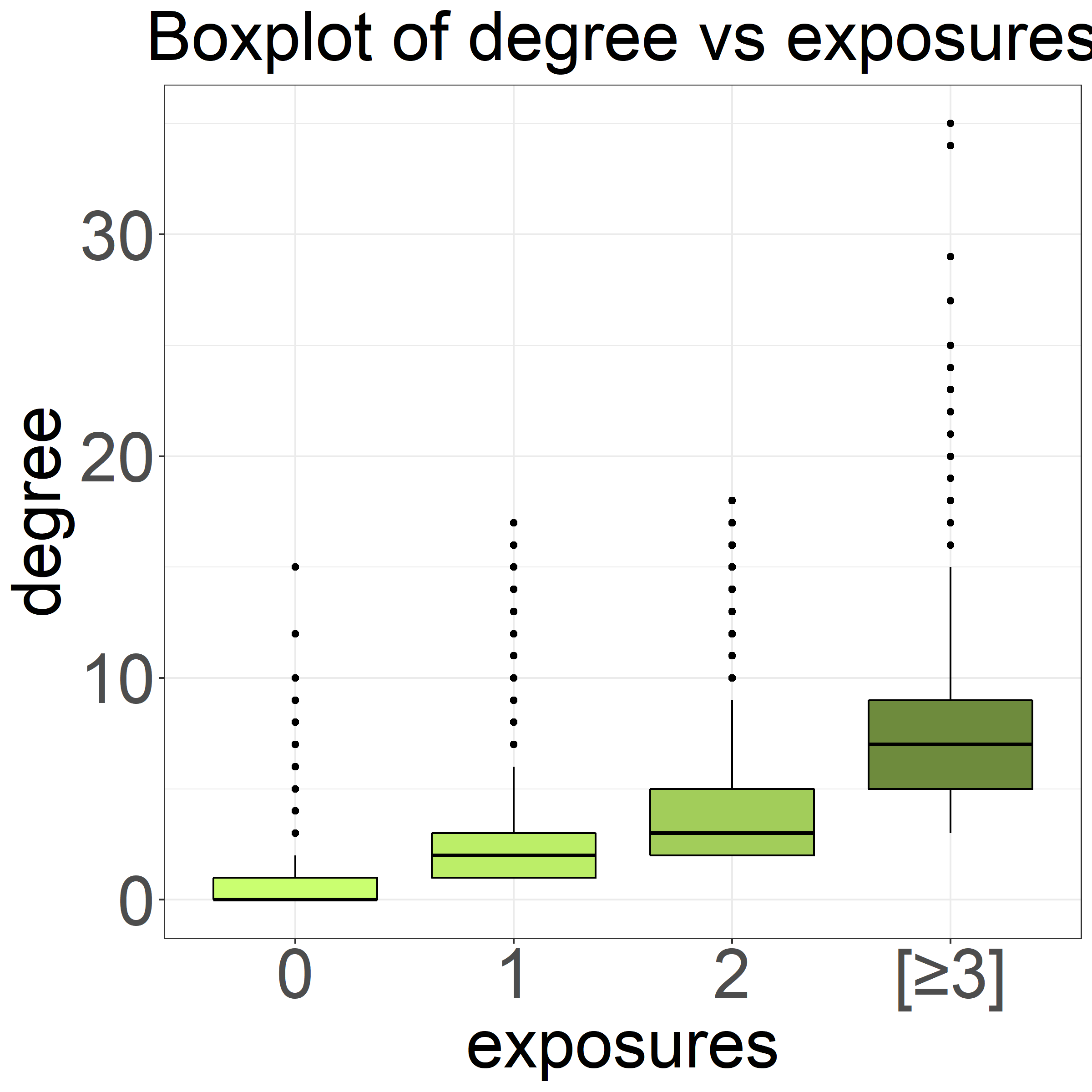}
    \caption{Boxplots of degree vs.\ treatment and exposure levels.}
    \label{fig:degree_boxplots}
\end{figure}

\subsection{More results on the monotone hypothesis}\label{appdx:more_emp_3H}
\paragraph{Results from different constructions of module sets.} 
Figure~\ref{fig:hist-dcrimeindex} shows the histograms of the $p$-values resulting from $2,000$ independent constructions of module sets for the crime index outcome. Table~\ref{tab:2medpval_deltaY} presents twice the median of the $2,000$ $p$-values for all outcomes. Both are for the beneficial spillover hypothesis~\eqref{eq:mononull_med}. We reject~\eqref{eq:mononull_med} at 5\% level for all outcomes using the rank-sum test statistic.
\begin{figure}[!h]
    \centering
    \includegraphics[width = 0.7\textwidth]{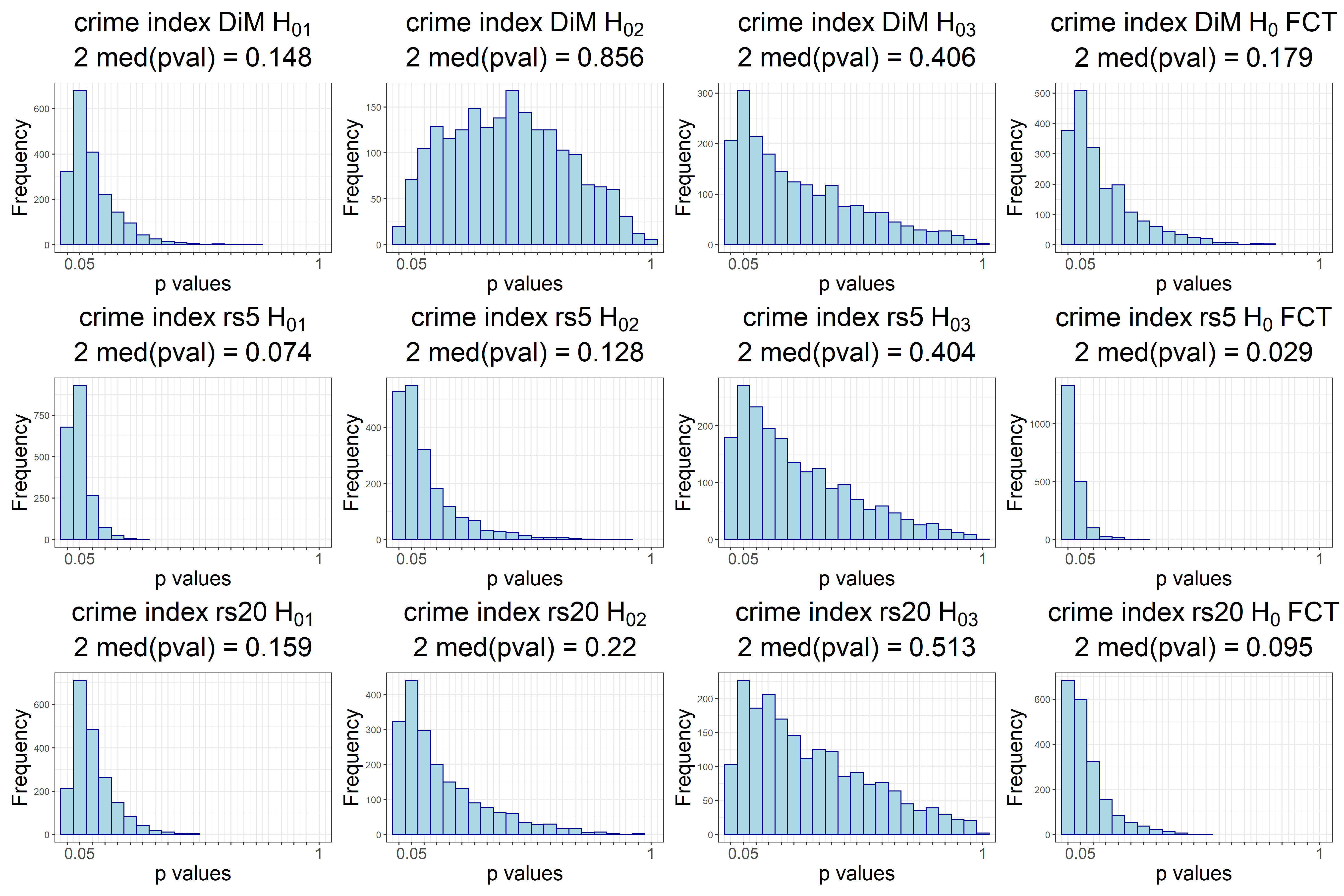}
    \caption{Histograms of $p$-values for crime index across the $2,000$ constructions.}
    \label{fig:hist-dcrimeindex}
\end{figure}

\begin{table}[!h]
\centering
\begin{tabular}{llllllll}
  \hline
Hypothesis & test stat & homicides & assault & robbery & car theft & moto theft & crime index \\ 
  \hline
$H_{01}$ & \multirow{4}{*}{DiM} & 0.1880 & 0.3741 & 1.4453 & 0.1712 & 0.6491 & 0.1480 \\ 
  $H_{02}$ &  & 1.7592 & 0.2204 & 0.4907 & 0.4739 & 0.6805 & 0.8564 \\ 
  $H_{03}$ &  & 1.0492 & 0.7910 & 0.6401 & 0.5173 & 0.5121 & 0.4063 \\ 
  $H_0$ by FCT &  & 0.6871 & 0.2251 & 0.7674 & 0.1456 & 0.4888 & 0.1790 \\ 
  \hline
  $H_{01}$ & \multirow{4}{*}{rs5} & 0.0608 & 0.0602 & 0.0844 & 0.0684 & 0.0752 & 0.0742 \\ 
  $H_{02}$ &  & 0.1418 & 0.1056 & 0.1388 & 0.1386 & 0.1210 & 0.1282 \\ 
  $H_{03}$ &  & 0.4283 & 0.4059 & 0.3675 & 0.3885 & 0.3509 & 0.4035 \\ 
  $H_0$ by FCT &  & \textbf{0.0281} & \textbf{0.0212} & \textbf{0.0308} & \textbf{0.0283} & \textbf{0.0256} & \textbf{0.0291} \\ 
  \hline
  $H_{01}$ & \multirow{4}{*}{rs20} & \textbf{0.0500} & 0.0530 & 0.2120 & 0.1016 & 0.1780 & 0.1588 \\ 
  $H_{02}$ &  & 0.1902 & 0.0868 & 0.2661 & 0.2695 & 0.1668 & 0.2200 \\ 
  $H_{03}$ &  & 0.5715 & 0.4875 & 0.5373 & 0.4827 & 0.4343 & 0.5131 \\ 
  $H_0$ by FCT &  & \textbf{0.0369} & \textbf{0.0183} & 0.1298 & 0.0651 & 0.0700 & 0.0951 \\ 
   \hline
\end{tabular}
\caption{Twice the median of $p$-values across repetitions. \textbf{Bold} denotes value below $0.05$.} 
\label{tab:2medpval_deltaY}
\end{table}

\paragraph{Results from the module sets that maximize the expected number of active focal units.}
Table~\ref{tab:crime_displacement} shows the test results for the beneficial spillover hypothesis~\eqref{eq:mononull_med} conducted on the module sets that maximize the expected number of active focal units among the $2,000$ ones for all outcomes. Additionally, Table~\ref{tab:crime_displacement2} presents the results for the crime displacement hypothesis~\eqref{eq:H0_crime_displacement} for all outcomes, conducted on the same module sets. We observe that conclusions similar to those in the main text for the crime index apply to other outcomes as well.
\begin{table}[!h]
\centering
\begin{tabular}{llllllll}
  \hline
Hypothesis & test stat & homicides & assault & robbery & car theft & moto theft & crime index \\ 
  \hline
$H_{01}$ & \multirow{4}{*}{DiM} & 0.1166 & 0.0990 & 0.5243 & \textbf{0.0358} & 0.0926 & \textbf{0.0108} \\ 
  $H_{02}$ &  & 0.8854 & 0.0620 & 0.1690 & 0.1798 & 0.6947 & 0.5379 \\ 
  $H_{03}$ &  & 0.7137 & 0.3837 & 0.1948 & 0.0746 & 0.2356 & 0.2002 \\ 
  $H_0$ by FCT &  & 0.5164 & 0.0597 & 0.2295 & \textbf{0.0182} & 0.2116 & \textbf{0.0356} \\ 
  \hline
  $H_{01}$ & \multirow{4}{*}{rs5} & \textbf{0.0102} & \textbf{0.0092} & \textbf{0.0130} & \textbf{0.0114} & \textbf{0.0106} & \textbf{0.0080} \\ 
  $H_{02}$ &  & \textbf{0.0418} & \textbf{0.0294} & 0.0558 & \textbf{0.0398} & \textbf{0.0448} & 0.0554 \\ 
  $H_{03}$ &  & 0.3059 & 0.2294 & 0.1938 & 0.2284 & 0.2144 & 0.1816 \\ 
  $H_0$ by FCT &  & \textbf{0.0065} & \textbf{0.0036} & \textbf{0.0069} & \textbf{0.0054} & \textbf{0.0053} & \textbf{0.0044} \\ 
  \hline
  $H_{01}$ & \multirow{4}{*}{rs20} & \textbf{0.0072} & \textbf{0.0054} & \textbf{0.0270} & \textbf{0.0112} & \textbf{0.0224} & \textbf{0.0108} \\ 
  $H_{02}$ &  & 0.0556 & \textbf{0.0226} & 0.1702 & 0.0834 & 0.0766 & 0.1640 \\ 
  $H_{03}$ &  & 0.4417 & 0.1516 & 0.2507 & 0.1672 & 0.2016 & 0.2797 \\ 
  $H_0$ by FCT &  & \textbf{0.0083} & \textbf{0.0013} & \textbf{0.0353} & \textbf{0.0075} & \textbf{0.0141} & \textbf{0.0186} \\ 
   \hline
\end{tabular}
\caption{Test $p$-values for the beneficial spillover hypothesis~\eqref{eq:mononull_med} using the module sets that maximize the expected number of active focal units. \textbf{Bold} denotes value below $0.05$.} 
\label{tab:crime_displacement}
\end{table}

\begin{table}[!h]
\centering
\begin{tabular}{llllllll}
  \hline
Hypothesis & test stat & homicides & assault & robbery & car theft & moto theft & crime index \\ 
  \hline
$H_{01}$ & \multirow{4}{*}{DiM} & 0.8834 & 0.9010 & 0.4757 & 0.9642 & 0.9074 & 0.9892 \\ 
  $H_{02}$ &  & 0.1146 & 0.9380 & 0.8310 & 0.8202 & 0.3053 & 0.4621 \\ 
  $H_{03}$ &  & 0.2863 & 0.6163 & 0.8052 & 0.9254 & 0.7644 & 0.7998 \\ 
  $H_0$ by FCT &  & 0.3133 & 0.9714 & 0.8913 & 0.9960 & 0.7957 & 0.9186 \\ 
  \hline
  $H_{01}$ & \multirow{4}{*}{rs5} & 0.9898 & 0.9908 & 0.9870 & 0.9886 & 0.9894 & 0.9920 \\ 
  $H_{02}$ &  & 0.9582 & 0.9706 & 0.9442 & 0.9602 & 0.9552 & 0.9446 \\ 
  $H_{03}$ &  & 0.6941 & 0.7706 & 0.8062 & 0.7716 & 0.7856 & 0.8184 \\ 
  $H_0$ by FCT &  & 0.9911 & 0.9964 & 0.9969 & 0.9960 & 0.9965 & 0.9974 \\ 
  \hline
  $H_{01}$ & \multirow{4}{*}{rs20} & 0.9928 & 0.9946 & 0.9730 & 0.9888 & 0.9776 & 0.9892 \\ 
  $H_{02}$ &  & 0.9444 & 0.9774 & 0.8298 & 0.9166 & 0.9234 & 0.8360 \\ 
  $H_{03}$ &  & 0.5583 & 0.8484 & 0.7493 & 0.8328 & 0.7984 & 0.7203 \\ 
  $H_0$ by FCT &  & 0.9720 & 0.9990 & 0.9854 & 0.9970 & 0.9954 & 0.9842 \\ 
   \hline
\end{tabular}
\caption{Test $p$-values for the crime displacement hypothesis~\eqref{eq:H0_crime_displacement} using the module sets that maximize the expected number of active focal units. \textbf{Bold} denotes value below $0.05$.} 
\label{tab:crime_displacement2}
\end{table}

\subsection{Monotone null hypothesis with six exposure levels}\label{sec:Med_5H}
Instead of grouping exposures beyond $3$ and test the monotone hypothesis in~\eqref{eq:mononull_med}, here we group exposures at a higher threshold $5$, so that $\exposet = \{0,1,2,3,4,[\geq 5]\}$, and test the following monotone hypothesis:
\begin{equation}\label{eq:mononull_med_5H}
\begin{gathered}
    H_0 = \bigcap_{k \in [5]} H_{0k}, ~\text{where}\\
    H_{01}: y_i(0, 0) \geq y_i(0, 1),\quad H_{02}: y_i(0, 1) \geq y_i(0, 2), \quad H_{03}: y_i(0, 2)\geq y_i(0, 3),\\
    \quad H_{04}: y_i(0, 3)\geq y_i(0, 4), \quad H_{05}: y_i(0, 4)\geq y_i(0, [\geq 5]), \quad \forall i.
\end{gathered}
\end{equation}
We are still testing the beneficial (harmless) hypothesis that when more nearby units are treated the crime in the current street is lower. We again adjust $\Delta Y_i = Y_i^{\text{post}} - Y_i^{\text{pre}}$ for the five types of crime and the crime index, and consider the difference-in-means and the Stephenson rank sum test statistics with $s=5, 20$, with $p$-values combined by Fisher's rule.

Figure~\ref{fig:hist-fu-5H} shows the histograms of the number of eligible and active focal units across the $2,000$ random constructions of module sets, and Table~\ref{tab:2medpval_deltaY_5H} shows twice the median of $p$-values resulting from applying Algorithm~\ref{algo:focalrand_multiple} on these $2,000$ module sets. Two of the $2,000$ realizations are presented in Figure~\ref{fig:5H_realization}. 
We observe similar patterns as in testing four exposure levels in Section~\ref{sec:realMed}, that focal units used to test hypotheses involving lower exposure contrast tend to spread at the outskirts, while focal units for hypotheses involving higher exposure contrast tend to be in the center of the city.
Table~\ref{tab:maxExAFUs_deltaY_5H} reports the result when we apply the test on the module sets with the largest expected number of active focal units. The results from Table~\ref{tab:maxExAFUs_deltaY_5H} still reject the beneficial null~\eqref{eq:mononull_med_5H}, and the rejection signals again mainly come from individual contrast hypotheses involving lower exposure levels. In this case, however, twice the median of $p$-values across repetitions become less powerful. 

\begin{figure}[!ht]
    \centering
    \includegraphics[width=0.7\linewidth]{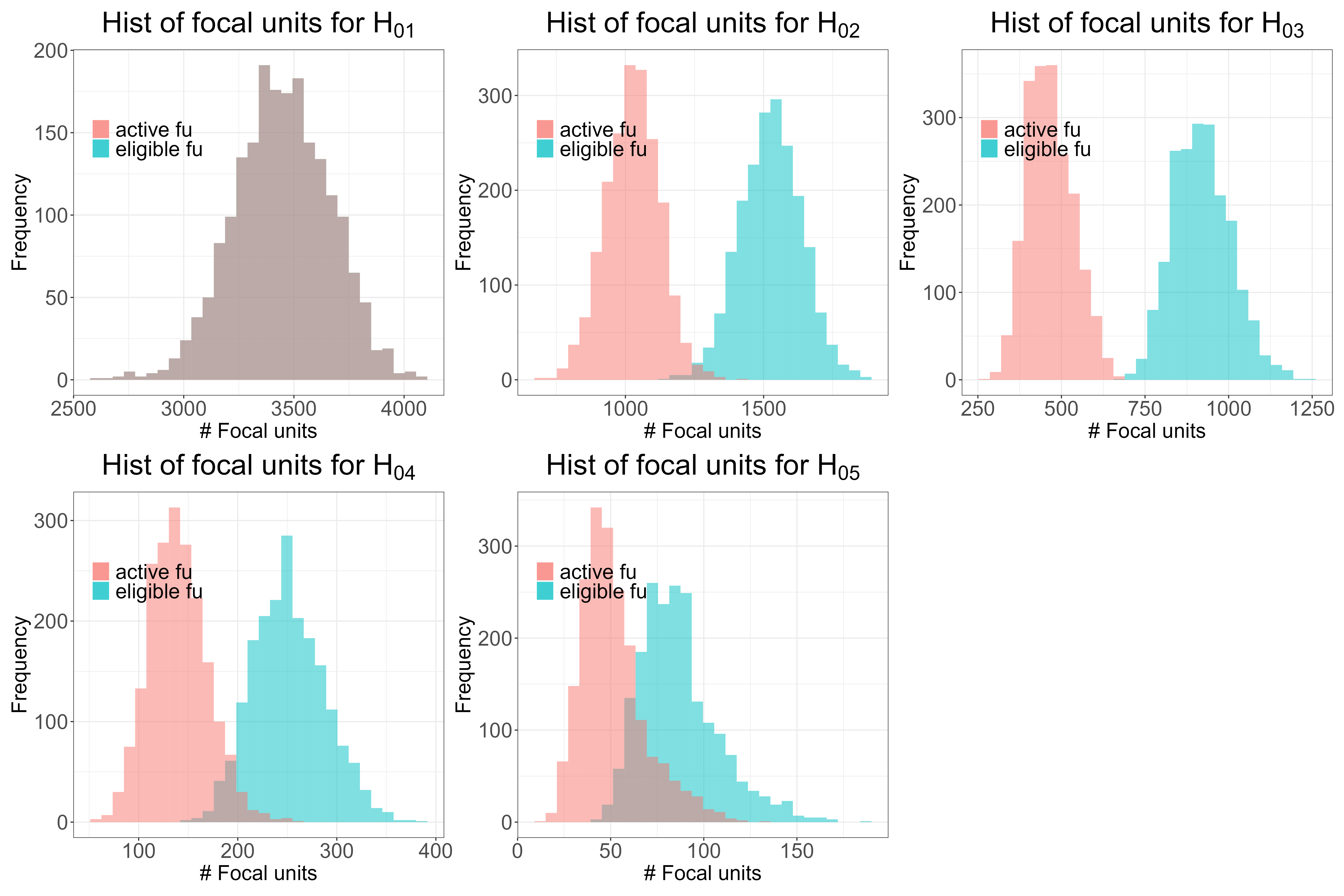}
    \caption{Histograms of number of focal units across the $2,000$ repetitions, for the null~\eqref{eq:mononull_med_5H}.}
    \label{fig:hist-fu-5H}
\end{figure}

\begin{figure}[!ht]
    \centering
    \includegraphics[width = 0.4\textwidth]{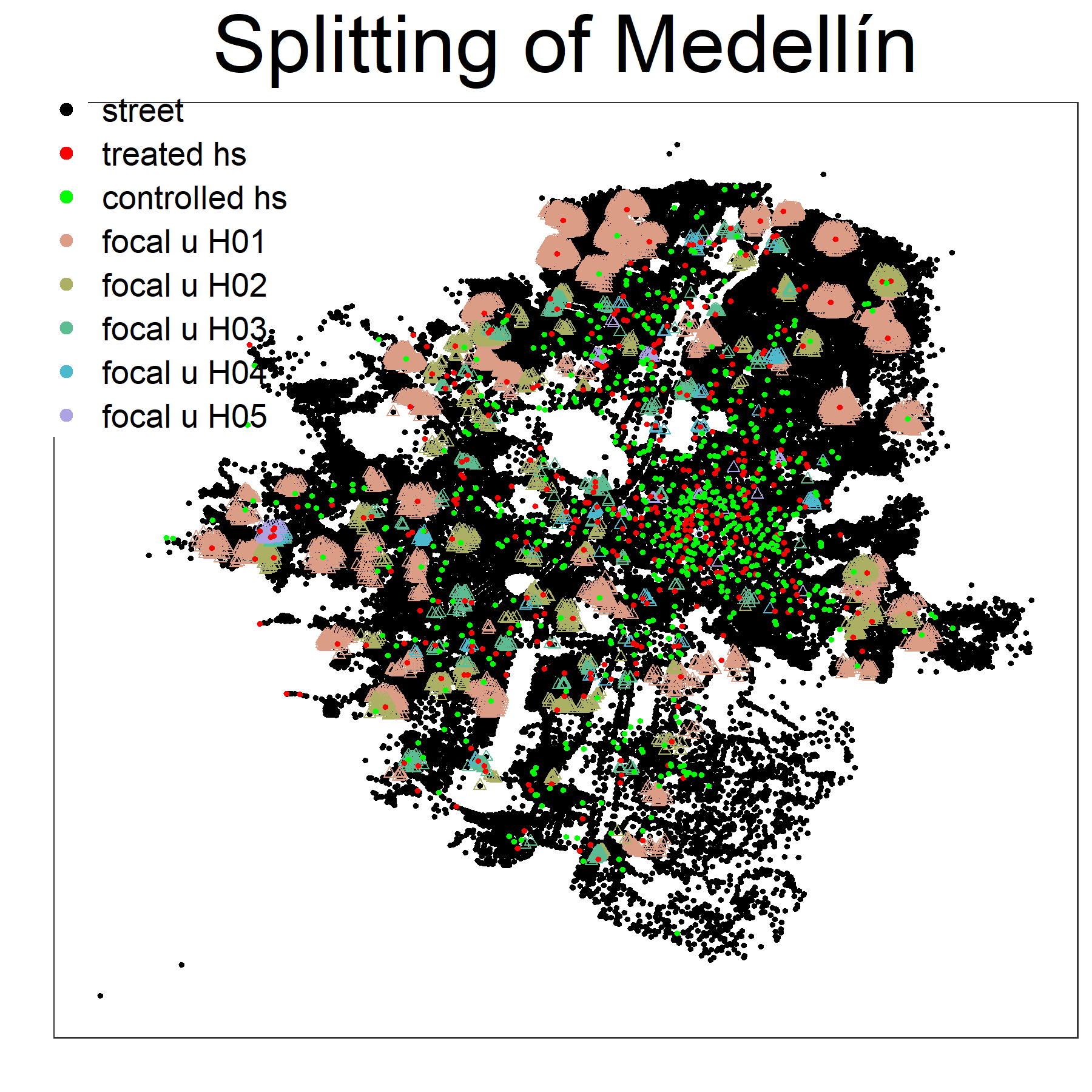}
    \includegraphics[width = 0.4\textwidth]{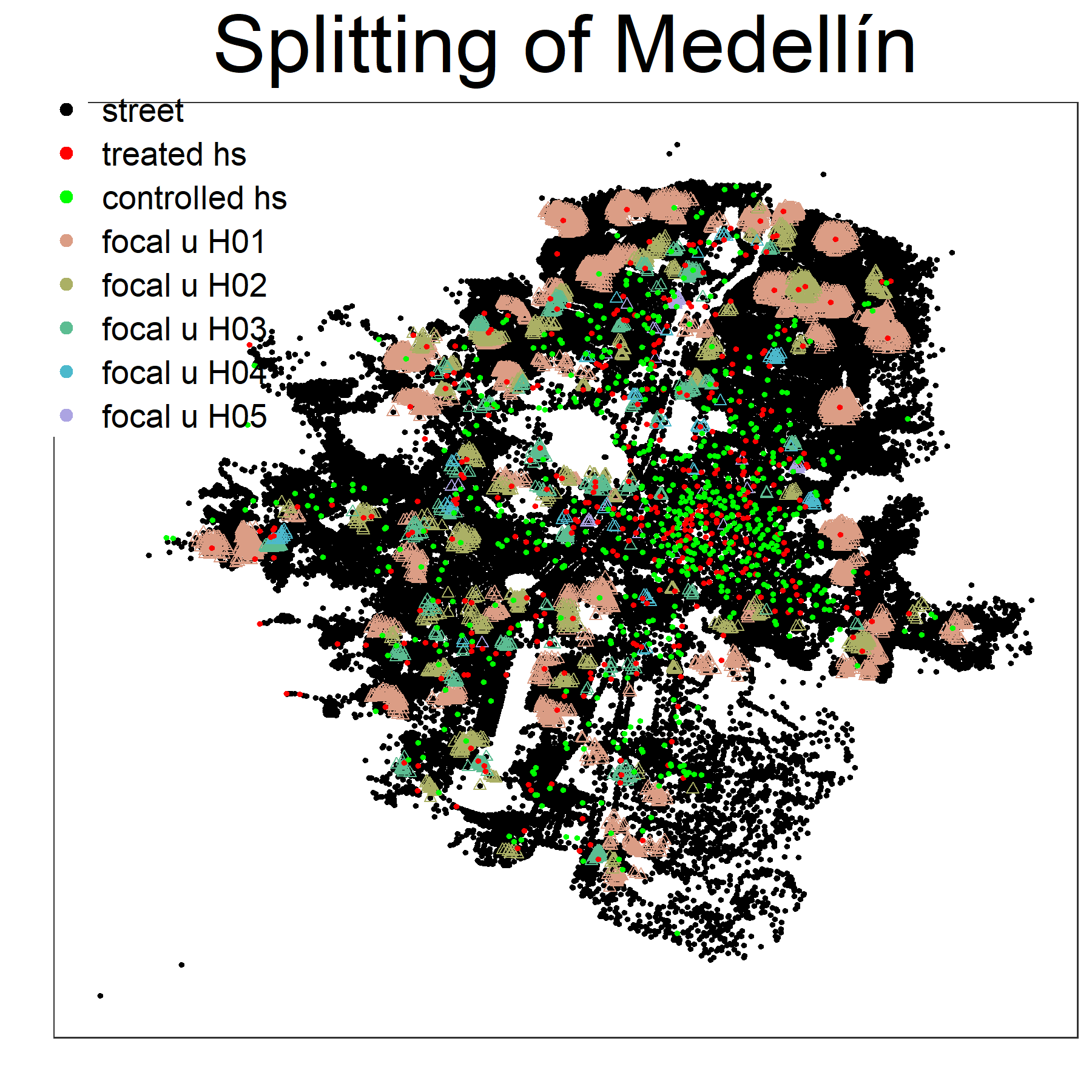}
    \caption{Two realizations of focal units for testing under $\exposet = \{0, 1, 2, 3, 4, [\geq5]\}$.}
    \label{fig:5H_realization}
\end{figure}

\begin{table}[ht]
\centering
\begin{tabular}{llllllll}
  \hline
Hypothesis & test stat & homicides & assault & robbery & car theft & moto theft & crime index \\ 
  \hline
$H_{01}$ & \multirow{6}{*}{DiM} & 0.2559 & 0.3885 & 1.4425 & 0.2190 & 0.5819 & 0.1750 \\ 
  $H_{02}$ &  & 1.7962 & 0.1744 & 0.4449 & 0.4027 & 0.7315 & 0.8566 \\ 
  $H_{03}$ &  & 1.2304 & 0.7407 & 0.3919 & 0.3747 & 0.8376 & 0.5235 \\ 
  $H_{04}$ &  & 1.1424 & 1.1740 & 0.4782 & 0.9213 & 0.5259 & 0.5455 \\ 
  $H_{05}$ &  & 1.1852 & 1.1622 & 1.3965 & 1.2170 & 1.2034 & 1.3709 \\ 
  $H_0$ by FCT &  & 1.0546 & 0.3616 & 0.5379 & 0.2351 & 0.6006 & 0.3128 \\ 
  \hline
  $H_{01}$ & \multirow{6}{*}{rs5} & 0.0860 & 0.0856 & 0.1118 & 0.0928 & 0.0996 & 0.0988 \\ 
  $H_{02}$ &  & 0.1438 & 0.1046 & 0.1392 & 0.1372 & 0.1252 & 0.1308 \\ 
  $H_{03}$ &  & 0.3763 & 0.3609 & 0.2875 & 0.3297 & 0.3349 & 0.3693 \\ 
  $H_{04}$ &  & 1.0216 & 1.2274 & 1.0192 & 1.0564 & 0.9492 & 0.9618 \\ 
  $H_{05}$ &  & 0.8530 & 0.8204 & 1.0296 & 0.8250 & 0.8656 & 0.9980 \\ 
  $H_0$ by FCT &  & 0.0580 & 0.0523 & 0.0590 & 0.0526 & 0.0517 & 0.0610 \\ 
  \hline
  $H_{01}$ & \multirow{6}{*}{rs20} & 0.0748 & 0.0776 & 0.2272 & 0.1224 & 0.1902 & 0.1784 \\ 
  $H_{02}$ &  & 0.1936 & 0.0872 & 0.2613 & 0.2426 & 0.1754 & 0.2509 \\ 
  $H_{03}$ &  & 0.4795 & 0.4223 & 0.3131 & 0.3913 & 0.4269 & 0.4421 \\ 
  $H_{04}$ &  & 0.9496 & 1.4525 & 1.1354 & 1.1928 & 0.9202 & 0.9762 \\ 
  $H_{05}$ &  & 0.8760 & 0.8258 & 1.1994 & 0.8004 & 0.8858 & 1.1252 \\ 
  $H_0$ by FCT &  & 0.0764 & 0.0618 & 0.1845 & 0.1069 & 0.1261 & 0.1660 \\ 
   \hline
\end{tabular}
\caption{Twice the median of $p$-values across repetitions for the null~\eqref{eq:mononull_med_5H}.} 
\label{tab:2medpval_deltaY_5H}
\end{table}

\begin{table}[ht]
\centering
\begin{tabular}{llllllll}
  \hline
Hypothesis & test stat & homicides & assault & robbery & car theft & moto theft & crime index \\ 
  \hline
$H_{01}$ & \multirow{6}{*}{DiM} & 0.1566 & 0.2352 & 0.7564 & 0.0980 & 0.0614 & \textbf{0.0346} \\ 
  $H_{02}$ &  & 0.7097 & \textbf{0.0466} & 0.0816 & 0.0548 & 0.6769 & 0.2649 \\ 
  $H_{03}$ &  & 0.9242 & 0.4089 & 0.3915 & 0.2613 & 0.1602 & 0.3707 \\ 
  $H_{04}$ &  & 0.8126 & 0.6397 & 0.1552 & 0.1718 & 0.5839 & 0.4265 \\ 
  $H_{05}$ &  & 0.3877 & 0.4073 & 0.8802 & 0.3885 & 0.1872 & 0.3825 \\ 
  $H_0$ by FCT &  & 0.7384 & 0.1967 & 0.3252 & \textbf{0.0463} & 0.1533 & 0.1322 \\ 
  \hline
  $H_{01}$ & \multirow{6}{*}{rs5} & 0.0732 & 0.0754 & 0.0936 & 0.0804 & 0.0872 & 0.0804 \\ 
  $H_{02}$ &  & \textbf{0.0016} & \textbf{0.0012} & \textbf{0.0008} & \textbf{0.0014} & \textbf{0.0016} & \textbf{0.0008} \\ 
  $H_{03}$ &  & 0.1846 & 0.1510 & 0.1212 & 0.1530 & 0.1152 & 0.1246 \\ 
  $H_{04}$ &  & 0.4359 & 0.5433 & 0.3847 & 0.3973 & 0.3999 & 0.3909 \\ 
  $H_{05}$ &  & 0.3077 & 0.2308 & 0.4267 & 0.3349 & 0.2314 & 0.2705 \\ 
  $H_0$ by FCT &  & \textbf{0.0045} & \textbf{0.0031} & \textbf{0.0028} & \textbf{0.0038} & \textbf{0.0028} & \textbf{0.0018} \\ 
  \hline
  $H_{01}$ & \multirow{6}{*}{rs20} & 0.0566 & 0.0712 & 0.1790 & 0.1050 & 0.1830 & 0.1246 \\ 
  $H_{02}$ &  & \textbf{0.0020} & \textbf{0.0006} & \textbf{0.0030} & \textbf{0.0006} & \textbf{0.0070} & \textbf{0.0022} \\ 
  $H_{03}$ &  & 0.2478 & 0.1332 & 0.1574 & 0.1732 & 0.0866 & 0.1398 \\ 
  $H_{04}$ &  & 0.4359 & 0.7129 & 0.3567 & 0.4397 & 0.4213 & 0.3321 \\ 
  $H_{05}$ &  & 0.3163 & 0.0646 & 0.7343 & 0.4503 & 0.1484 & 0.2671 \\ 
  $H_0$ by FCT &  & \textbf{0.0055} & \textbf{0.0008} & \textbf{0.0182} & \textbf{0.0036} & \textbf{0.0083} & \textbf{0.0050} \\ 
   \hline
\end{tabular}
\caption{Test $p$-values for the null~\eqref{eq:mononull_med_5H} using the module sets that maximize the expected number of active focal units. \textbf{Bold} denotes value below $0.05$.}
\label{tab:maxExAFUs_deltaY_5H}
\end{table}

\end{document}